\documentclass{article}

\usepackage[main, final]{neurips_2026}
\usepackage{microtype}
\usepackage{graphicx}
\usepackage{float}
\usepackage{subfigure}
\usepackage{booktabs}
\usepackage{algorithm}
\usepackage{algpseudocode}
\usepackage{bm}
\usepackage{amssymb}
\usepackage{bbding}
\usepackage{pifont}
\usepackage{enumitem}
\usepackage{hyperref}
\usepackage{amsmath}
\usepackage{wrapfig}
\usepackage{amsthm}

\newtheorem{theorem}{Theorem}[section]
\newtheorem{proposition}[theorem]{Proposition}
\newtheorem{lemma}[theorem]{Lemma}
\newtheorem{corollary}[theorem]{Corollary}
\theoremstyle{definition}
\newtheorem{definition}[theorem]{Definition}

\theoremstyle{remark}
\newtheorem{remark}[theorem]{Remark}

\usepackage[utf8]{inputenc} % allow utf-8 input
\usepackage[T1]{fontenc}    % use 8-bit T1 fonts
\usepackage{hyperref}       % hyperlinks
\usepackage{url}            % simple URL typesetting
\usepackage{booktabs}       % professional-quality tables
\usepackage{amsfonts}       % blackboard math symbols
\usepackage{nicefrac}       % compact symbols for 1/2, etc.
\usepackage{microtype}      % microtypography
\usepackage{xcolor}         % colors

\newif\ifrevisionmarks
\revisionmarkstrue
\newcommand{\revcolor}{\ifrevisionmarks\color{black}\fi}
\DeclareRobustCommand{\rev}[1]{\ifrevisionmarks\textcolor{black}{#1}\else#1\fi}
\DeclareRobustCommand{\grm}[1]{\ifrevisionmarks\textcolor{black!50!black}{#1}\else#1\fi}
\makeatletter
\newcommand{\revbib}[1]{\@for\rev@key:=#1\do{\expandafter\let\csname rev@newbib@\rev@key\endcsname\@empty}}
\newif\ifrev@bibblue
\newcommand{\revmarkbibitems}{% color the entries (labels included) of references added in the camera-ready
  \let\rev@origbibitem\bibitem
  \def\bibitem##1{\ifcsname rev@newbib@##1\endcsname\revcolor\rev@bibbluetrue
    \else\ifrev@bibblue\normalcolor\rev@bibbluefalse\fi\fi\rev@origbibitem{##1}}}
\makeatother
\revbib{pst13,jiang2024mixtral,grattafiori2024llama3,johnson2023mimiciv,krizhevsky2009learning,krizhevsky2012imagenet,hu2022lora,yang2025qwen3}

\graphicspath{{figures/}}

\begin{document}
% Note. For the workshop paper template, both \title{} and \workshoptitle{} are required, with the former indicating the paper title shown in the title and the latter indicating the workshop title displayed in the footnote. 
\title{TensorCommitments: A Lightweight Verifiable Inference for Large Language Models}

% The \author macro works with any number of authors. There are two commands
% used to separate the names and addresses of multiple authors: \And and \AND.
%
% Using \And between authors leaves it to LaTeX to determine where to break the
% lines. Using \AND forces a line break at that point. So, if LaTeX puts 3 of 4
% authors names on the first line, and the last on the second line, try using
% \AND instead of \And before the third author name.

\author{%
  \revcolor Oguzhan Baser$^{1}$\thanks{\revcolor Correspondence to: Oguzhan Baser \texttt{<oguzhanbaser@utexas.edu>}.} \quad
  Elahe Sadeghi$^{1}$ \quad
  Eric Wang$^{2}$ \quad
  Nico Vergauwen$^{2}$ \\
  \revcolor\textbf{Sam Kazemian}$^{2}$ \quad
  \textbf{Hong Kang}$^{3}$ \quad
  \textbf{Sandeep P. Chinchali}$^{1}$ \quad
  \textbf{Sriram Vishwanath}$^{4}$ \\
  \revcolor$^{1}$Electrical and Computer Engineering, The University of Texas at Austin \\
  \revcolor$^{2}$Theseus AI Labs \\
  \revcolor$^{3}$Electrical and Computer Engineering, McGill University \\
  \revcolor$^{4}$Electrical and Computer Engineering, Georgia Institute of Technology \\
}

\maketitle

\begin{abstract}
Most large language models (LLMs) run on external clouds: users send a prompt, pay for inference, and must trust that the remote GPU executes the LLM without any adversarial tampering. We critically ask how to achieve verifiable LLM inference, where a prover (the service) must convince a verifier (the client) that an inference was run correctly without rerunning the LLM. Existing cryptographic works are too slow at the LLM scale, while non-cryptographic ones require a strong verifier GPU. We propose TensorCommitments (TCs), a tensor-native proof-of-inference scheme. TC binds the LLM inference to a commitment, an irreversible tag that breaks under tampering, organized in our multivariate Terkle Trees. For LLaMA2, TC adds only 0.97\% prover and 0.12\% verifier time over inference while improving robustness to tailored attacks by up to 48\% over the best prior work requiring a verifier GPU.
\end{abstract}

\section{Introduction}\label{sec:intro}

As large language models (LLMs) move \textit{from chatbots to decision-makers}, we increasingly have to \textbf{trust} remote inference runs that we cannot see or replay. Today, LLMs help draft code and execute tools in software pipelines, coordinate multi-agent workflows, and support decisions in finance, healthcare, and law \cite{survey1}. In many of these settings, an LLM runs on external infrastructure, produces a stream of tokens, and \textit{downstream systems act on those tokens as if they were ground truth}. If the underlying computation is \textit{silently corrupted} - by hardware faults, misconfiguration, or even a compromised service - agents may trade on the wrong market signal, a medical assistant may summarize the wrong patient record, or an autonomous system may reason over tampered state \cite{survey1}.

Current deployments provide almost \textbf{no traceability} of the intermediate activations or hidden states of these billion-parameter models; only the final text is visible, and re-running the full model just to check every answer is prohibitively expensive. As LLM usage scales to persistent, tool-using, and multi-agent systems, the gap between the \textit{impact} of a single faulty inference and our ability to cheaply verify it is widening \cite{li2024survey}. This motivates the \textbf{central question} of this work: \textit{can we design a practical mechanism that allows lightweight clients to check whether a complex LLM inference was executed correctly, without \grm{re-running} the model or exposing its internal states?} \rev{The question, and our mechanism, extend to any setting where one party pays for computation it cannot observe, e.g., compute markets, chained multi-party inference, federated aggregation, and agents (App.~\ref{app:settings}).}

Existing works on verifiable LLMs %addresses different facets of this challenge with distinct tradeoffs. 
\textit{only partially} address this and \grm{run} into \textit{sharp trade-offs} between scalability, privacy, and verifier cost. Cryptographic systems based on \textit{zero-knowledge} (ZK) proofs \cite{chen2024zkml, kang2023scaling, lee2021privacy, liu2021zkcnn, zkllm, weng2021mystique} can certify end-to-end correctness of an inference while hiding model parameters, but by compiling every tensor operation into a large constraint system. For LLM-scale models, this leads to minutes of prover \grm{runtime} per query and huge extra compute compared to plain inference, which is hard to amortize in interactive or multi-agent workloads. In parallel, \textit{learning-based works} \cite{qi2025verifierq, sun2024svip} \grm{train} auxiliary models to detect perturbed outputs, or ``critique'' reasoning traces, but their guarantees are statistical\rev{, with no cryptographic binding}. Also, they need retraining or \grm{fine-tuning} whenever the base model or domain shifts. A very recent strong method, TOPLOC \cite{TOPLOC}, takes a more practical route and fits last activations to a polynomial to detect model, prompt, or precision changes with low storage overhead, but the verifier still needs to recompute the forward pass with access to model weights and substantial GPU memory. It exposes internal states and thus violates user privacy. Further, these schemes flatten the compute into vectors, checking correctness per token instead of maintaining the tensor layout (e.g., spatiotemporal grids) \cite{safetynets}. This vector-centric view inflates the proof size with each new dimension and offers no direct way to reason about localized changes, motivating a tensor-native scheme that scales to LLMs with light-compute verifiers.

\begin{wrapfigure}{r}{0.5\textwidth}\centering\vspace{-0.5cm}\includegraphics[width=\linewidth]{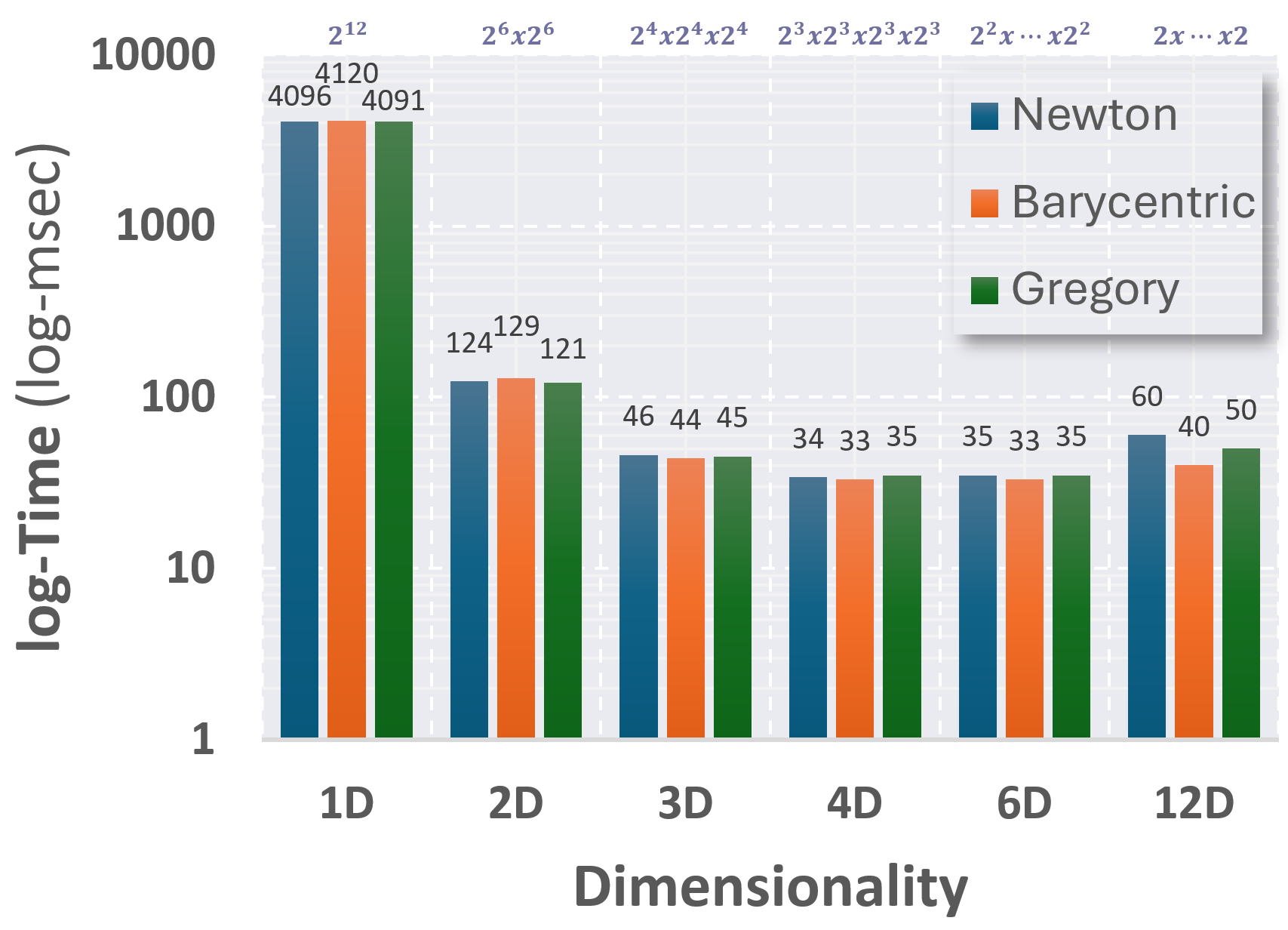}
\caption{\small\textbf{Our key observation:} \textit{multivariate interpolation is faster.} We plot log-runtime to interpolate a polynomial over a fixed grid of $N\!=\!2^{12}$ samples, reshaped from 1D ($2^{12}$ points) to $m$D grids ({\small $2^\frac{12}{m}\!\times\cdots\times 2^\frac{12}{m}$}) using Newton, Barycentric, and Gregory interpolation. Across all, moving from univariate ($m=1$) to bivariate ($m=2$) cuts runtime from 4.1s to 0.125s (over \textbf{30}$\times$ speedup), with further reductions as dimension increases. Its time bound {\tiny $\mathcal{O}\!\big(\binom{m+\lfloor N^{1/m}\rfloor}{m}^2\big)$}\grm{, derived} in App.~\ref{app:complexity}, decreases sharply as the tensor dimension $m$ grows.}\label{fig:polydecrease}\vspace{-0.1cm}
\end{wrapfigure}

Our \textbf{key technical insight} is that commitment schemes are conceptually well matched to verifiable inference: if each hidden state and parameter block is bound to a succinct commitment, a remote run can be \rev{audited} with a few proofs rather than replaying the full model. However, existing constructions are fundamentally \textit{vector}-centric: they encode long flat vectors with \textit{univariate polynomials} \cite{kgz}. For billion-parameter LLMs, the polynomial degree and the committed vectors grow with every layer and token, so interpolation rapidly dominates the cost, even without multi-agent interactions. Our \textbf{key observation} is that this bottleneck stems from \textit{the univariate view}, not from commitments themselves. By viewing parameters as tensors and committing to \emph{multivariate} polynomials to respect their natural axes (e.g., layers and heads), we preserve tensor structure and exploit fast multivariate evaluation to keep commitments and proofs efficient.  Fig.~\ref{fig:polydecrease} shows this core scaling behavior: \textit{for fixed data, multivariate interpolation becomes substantially more efficient than \grm{the} univariate one as dimensionality increases}. This makes tensor-native commitments \textbf{a natural fit} for LLM workloads. Fig.~\ref{fig:sysarch} gives the overview: a model host attaches \textit{TensorCommitments}, organizes them into \textit{an authenticated tree}, and streams compact proofs that \textit{a lightweight client} can check without GPUs.  In light of prior work, our contributions are \grm{fourfold}:

\begin{itemize}
    \item We design a \textbf{multivariate commitment scheme} tailored to \textit{verifiable LLM inference}. %and long sequences.
    \item We introduce \textbf{Terkle Trees}, a tensor-native authentication structure that tracks evolving hidden states and tensors with just \textit{a single root of multivariate proofs.}
    \item We design a \textbf{layer selection algorithm} to verify the \textit{most critical layers} for the model output, cutting overhead while \rev{improving detection of} targeted attacks.
    \item We \textit{systematically study} these mechanisms on several LLMs and \textit{tailored attacks}, and \textbf{open source our implementation}\footnote{\rev{\url{https://github.com/UTAustin-SwarmLab/TensorCommitments}}\label{github}} ready to merge into inference pipelines.
\end{itemize}

\begin{figure*}[!t]
    \hspace{1cm}
    \centering
    \includegraphics[width=0.79\linewidth]{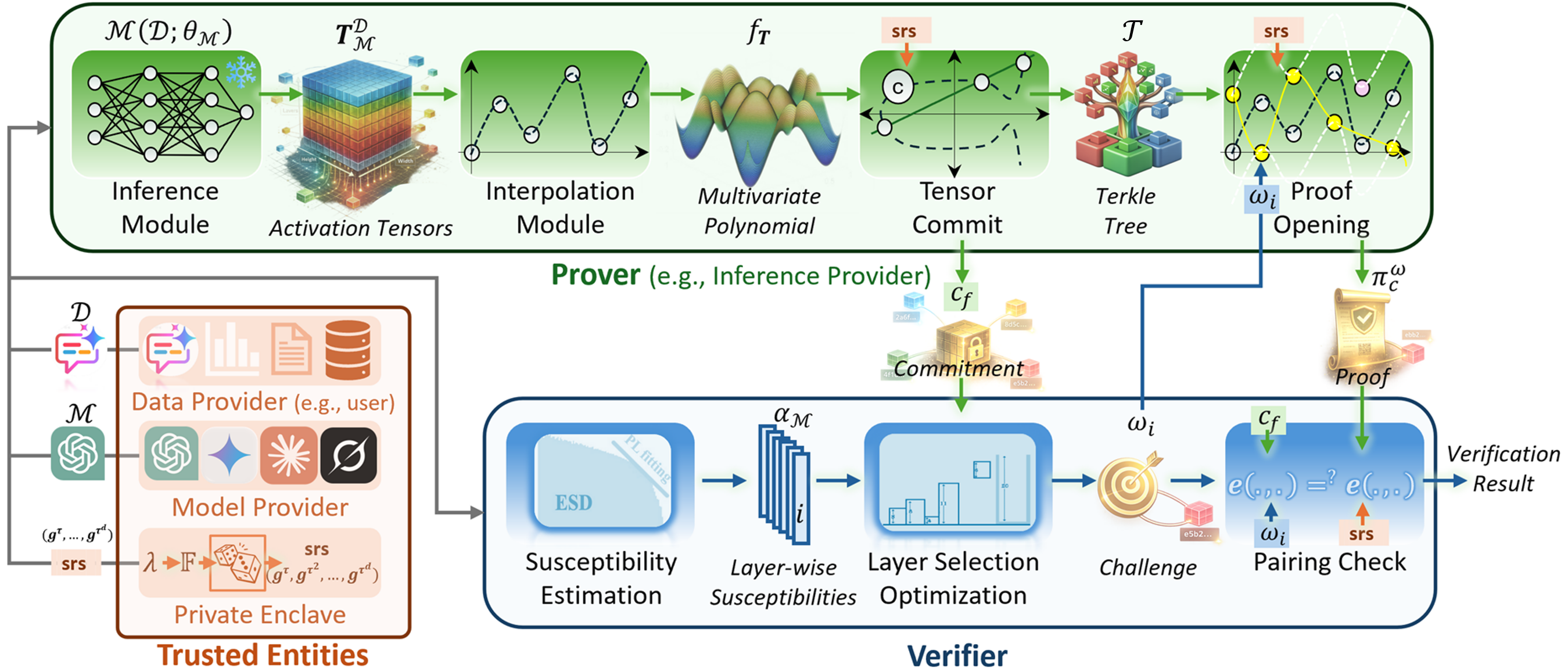}
 \caption{\small \textbf{How does our verifiable inference pipeline work?} 
\emph{Trusted setup} publishes structured reference string $\mathsf{srs}$ \rev{(Sec.~\ref{sec:method})} with the model and data. \emph{Prover} runs model {\small $\mathcal{M}(\mathcal{D};\theta_{\mathcal{M}})$} to produce activation tensors $T^{\mathcal{D}}_{\mathcal{M}}$; interpolates into \grm{multivariate} polynomial $f_{T_{\mathcal{M}}^{\mathcal{D}}}$; commits to obtain $C_{f}$ and builds Terkle tree $\mathcal{T}$; upon verifier challenge $\omega_i$, provides opening proofs $\pi^{\omega_i}_{C}$. \emph{Verifier} uses spectral heavy-tail scores $\alpha_{\mathcal{M}}$ to rank layers, solves the interval selector, Problem~\ref{problem}.5, to choose challenge $\{\omega_i\}$, and checks pairing $e(\cdot,\cdot)$, accepting only if all checks pass. The prover does all heavy work. The verifier checks pairings and \emph{does not} re-run full inference.}
\label{fig:sysarch}
\end{figure*}

\section{Background}
Here, we outline the main concepts shaping our framework. \textbf{Vector Commitments (VCs)} bind a vector to a single group element with succinct \emph{proofs} at single points. Let $\mathbb{F}$ be a finite field, and let $G, G_T$ be pairing-friendly \grm{elliptic curves} of prime order $p$ equipped with a bilinear map $e : G \times G \to G_T$. \grm{App.}~\ref{app:pairing-EC} explains these algebraic structures in more detail. Let $g \in G$ be a generator and let $[d]:=\{1,\dots,d\}$. We use boldface $\mathbf{v} \in \mathbb{F}^d$ for vectors of dimension $d$. 

\begin{definition}
\label{def:security}
Let $\lambda \in \mathbb{N}$ be a \emph{security parameter}, typically $\lambda \geq 128$ bits, determining the computational hardness of breaking the scheme. The security is satisfied when all probabilities of adversarial success \grm{are} negligible in $\lambda$, i.e., \rev{$\mathrm{negl}(\lambda)$ decays faster than any inverse polynomial in $\lambda$}.
\end{definition}

\begin{definition}
\label{def:vc}
\emph{Vector \grm{commitments}} $\mathsf{VC}$ \grm{have} four methods: {\small $\mathsf{Setup}_{\mathsf{VC}}(.),\mathsf{Com}_{\mathsf{VC}}(.),\mathsf{Open}_{\mathsf{VC}}(.),\mathsf{Ver}_{\mathsf{VC}}(.)$}, where $\mathbf{v}=(v_1,\dots,v_d)\!\in\!\mathbb{F}^d$ is the committed vector, $i \in [d]$ is an index, $C\!\in\!G$ is a commitment, and $\pi_i$ is an opening proof for coordinate $i$ with the claimed value $y\!\in\!\mathbb{F}$.

\noindent\underline{\small $\mathsf{Setup}_{\mathrm{vc}}(\lambda, d) \to \mathsf{pp}$:} It
samples trapdoor $\tau \in \mathbb{F}$ uniformly via the security parameter $\lambda$ (detailed in App. \ref{app:security-parameter}) and publishes the structured reference string {\small $\mathsf{srs}\!=\!(g^{\tau^0},\dots,g^{\tau^d})\!\in\!G^{d+1}$} together with evaluation domain $\Omega\!=\!\{\omega_1,\dots,\omega_d\}\!\subset\!\mathbb{F}$ of distinct points, where public parameters $\mathsf{pp}\!=\!(\mathsf{srs},\Omega)$.

\noindent\underline{\small $\mathsf{Com}_{\mathrm{vc}}(\mathsf{pp}, \mathbf{v}) \to C$:} 
Given $\mathbf{v}\in\mathbb{F}^d$, it interpolates a univariate polynomial $f\in\mathbb{F}[X]$ of degree $\!<\!d$ satisfying $f(\omega_i)\!=\!v_i, \forall i\!\in\![d]$, and outputs $C\!=\!g^{f(\tau)} \!\in\!G$ obtained by multi-exponentiation using $\mathsf{srs}$.

\noindent\underline{\small $\mathsf{Open}_{\mathrm{vc}}(\mathsf{pp}, \mathbf{v}, i) \to (y, \pi_i)$:} 
To prove $\mathbf{v}$ at position $i\in[d]$, \grm{it} sets $y:=v_i$ and computes the quotient polynomial $\rev{q}(X):=(f(X)-y)/(X-\omega_i)$. It outputs $\pi_i := g^{\rev{q}(\tau)}$.

\noindent\underline{\small $\mathsf{Ver}_{\mathrm{vc}}(\mathsf{pp}, C, i, y, \pi_i)$:} 
If {\small $e(C\!\cdot\! g^{-y},\,\! g)\!\stackrel{?}{=}\!e(\pi_i,\, \!g^{\tau-\omega_i})$}, it accepts $y$ by the pairing. $g^{\tau-\omega_i}$ comes from $\mathsf{srs}$. 
\end{definition}

\begin{definition}
A $\mathsf{VC}$ is \emph{correct} if, for any honestly generated $(\mathsf{pp},\mathbf{v},C)$ and $\forall i\!\in\![d]$, the opening {\small $(y,\pi_i)\!\leftarrow\!\mathsf{Open}_{\mathrm{vc}} (\mathsf{pp},\mathbf{v},i)$} satisfies $y\!=\!v_i$ and {\small $\mathsf{Ver}_{\mathrm{vc}}(\mathsf{pp},C,i,y,\pi_i)\!=\!1$}. Honest {\small $(C,\pi)$ }never get rejected. 
\end{definition}
\begin{definition}
A $\mathsf{VC}$ is \emph{position-binding} if, for a probabilistic polynomial-time (PPT) adversary $\mathcal{A}$, the probability that
$\mathcal{A}$ outputs
\(
(y, \pi) \neq (y', \pi') 
\)
for $(\mathsf{pp}, C, i)$,
such that
{\small $\mathsf{Ver}_{\mathrm{vc}}(\mathsf{pp}, C, i, y, \pi)\!=\!1$} and
{\small $\mathsf{Ver}_{\mathrm{vc}}(\mathsf{pp}, C, i, y', \pi')\!=\!1$}
is negligible in $\lambda$. Once $C$ is fixed, changes in any entry fail the verification.
\end{definition}
Intuitively, even a \grm{compute-heavy} adversary cannot take $C$ and later convince the verifier of a different $v_i$ at $i$. Hence, $C$ behaves like a binding handle to a unique vector.

\underline{\textbf{Merkle Trees}} (MTs) are cryptographic structures that use a hash function $H$ to commit to a data vector $\mathbf{x}=(x_i)_{i\in[n]}$ \cite{merkletrees}. $x_i$ become leaves of a full $B$-ary MT of depth $L=\lceil\log_B n\rceil$ with \rev{root $u_0$}.
\begin{definition} We set $h(u)\!:=\!H(x_i)$ for each leaf node $u$ holding value $x_i$ and \(
h(u)\!:=\!H\big(h(c(u,1))\!\parallel\!\dots\!\parallel\!h(c(u,B))\big)\) for each internal node $u$ holding concatenated 
hashes of its children $c(u,1),\dots,c(u,B)$. Then, the \emph{Merkle commitment} to $\mathbf{x}$ is defined as the root hash \(C^{\mathrm{M}} := h(\rev{u_0})\).
\end{definition}
\begin{definition}
Let $i\in[n]$ index a leaf with root-to-leaf path
$(u_0,\dots,u_L)$, \rev{where} $u_L$ stores $x_i$, and let $j_d\in[B]$
be the position of $u_{d+1}$ among the children of $u_d$.  For each level
$d\in\{0,\dots,L-1\}$, define the sibling multiset \(\mathbf{\Psi}_d := \big(h(c(u_d,j))\big)_{j\in[B]\setminus\{j_d\}}.\)  Then the \emph{Merkle proof} for $x_i$ is defined as: \(\pi_i^{\mathrm{M}} := (\mathbf{\Psi}_0,\dots,\mathbf{\Psi}_{L-1}).\)
\end{definition}
Given $(C^{\mathrm{M}}, i, x_i, \pi_i^{\mathrm{M}})$, the verifier recomputes
the leaf hash $H(x_i)$, then hashes upward using the sibling sets
$\mathbf{\Psi}_d$ to reconstruct a root $\hat{C}$ and accepts if
$\hat{C}=C^{\mathrm{M}}$.
\begin{remark}\label{th:merklecomplexity}
For a $B$-ary Merkle tree with $n$ leaves, each proof
$\pi_i^{\mathrm{M}}$ contains $\mathcal{O}((B-1)\log_B n)$ hashes and
verification performs $\mathcal{O}(\log_B n)$ evaluations of $H$.
\end{remark}
MTs are cheap and common (e.g., in blockchains) when $n$ is moderate. For LLMs with billions of tensors, Remark~\ref{th:merklecomplexity} forces smaller $B$ to keep proofs small, leading to very deep trees and large bandwidth. Also, hashing destroys tensor relations and algebraic structure:  verifying related tensor entries (e.g., adjacent tokens) requires independent proofs with no amortization across the tensors.

\begin{figure*}[!t]
    \centering
    \includegraphics[width=0.9\linewidth]{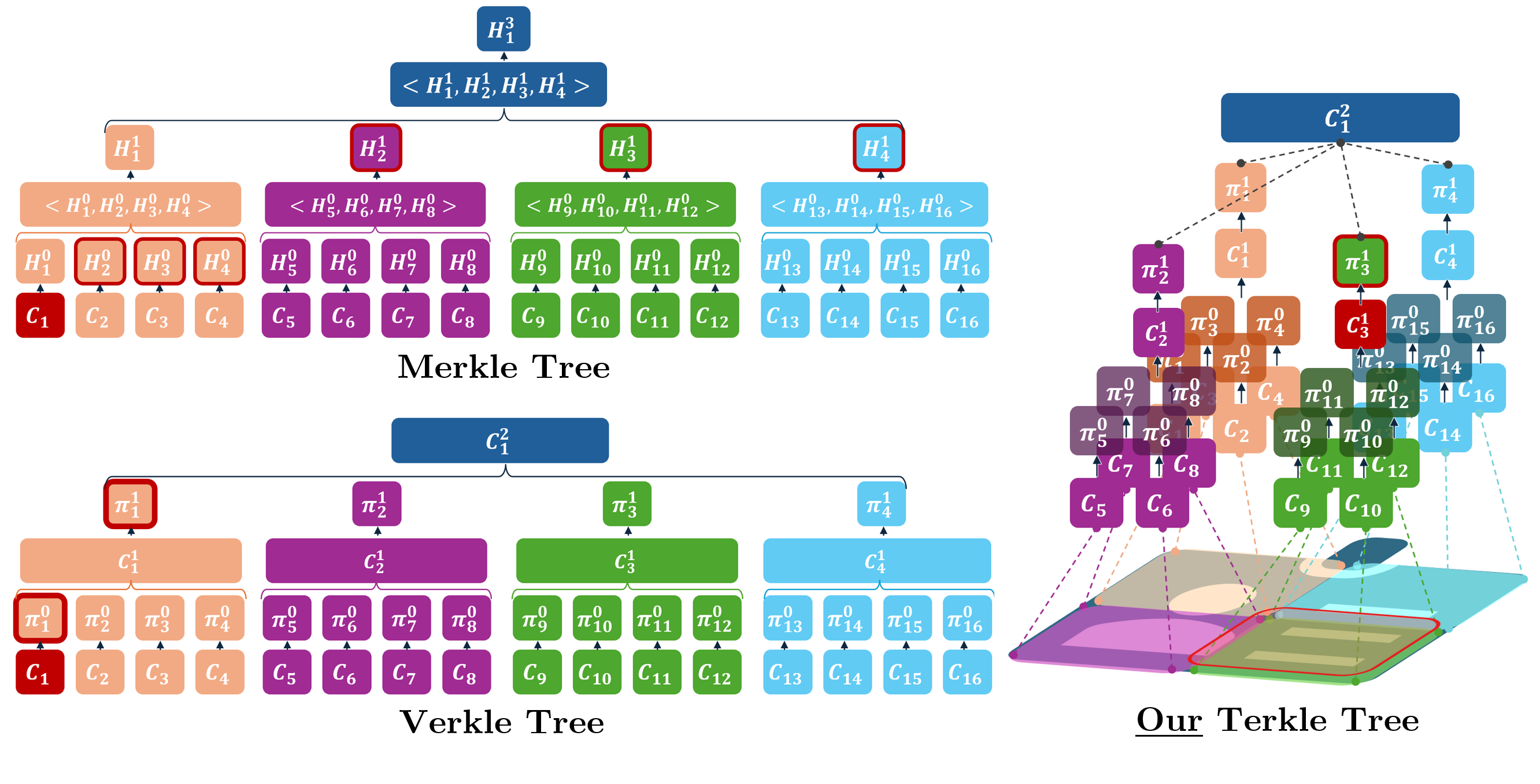}
    \caption{\small \textbf{Which verifiable tree aligns best with the tensor structure while keeping the proofs succinct?} \textit{(Top)} A $B$-ary Merkle tree commits to leaf values $C_i$ via hash labels $H_j^{d}$; a membership proof for a single leaf (highlighted in red) must include all sibling hashes along the path from leaf to root, treating the state as a flat list of values. \textit{(Bottom)} A Verkle tree replaces hash parents with vector commitments $C_j^{1}$ and per-level opening proofs $\pi_j^{d}$, reducing proof size to $\rev{\mathcal{O}}(\log_B n)$ but still indexing children in one dimension, without exploiting any tensor structure in the underlying model or feature map. \textit{(Right)} A Terkle tree commits at the root $C_1^{2}$ to a tensor-shaped grid of parameters or features (illustrated as colored regions on the base plane); each internal node $C_j^{1}$ corresponds to a multi-dimensional block, and each $\pi_j^{d}$ is a multivariate opening at a specific tensor index. This tensor-native organization allows authenticating the entire LLM or multi-agent states with a single root while informing about structured subsets (e.g., spatial patches) using fewer openings.}\label{fig:trees}
\end{figure*}

\underline{\textbf{Verkle Trees}} (VTs) build on MTs by replacing the hash function with VCs \cite{verkletree} with 4 functions given in Def. \ref{def:vc} and \grm{we} reuse the same notation \grm{as} in MTs and VCs.
\begin{definition}
For each node $u$, VT collects its $B$ children into
\(
\mathbf{z}_u := (z_{u,1},\dots,z_{u,B}) \in \mathbb{F}^B
\)
and sets the node commitment \(C_u := \){\small \(\mathsf{Com}_\mathsf{VC}\)}\((\mathbf{z}_u).\) $z_{u,j}$ is data $x_i$ at leaves and it is the commitment of child $c(u,j)$ at internal nodes. Then, \textit{Verkle root commitment} to $\mathbf{x}$ is defined as $C^{\mathrm{V}} := C_{\rev{u_0}}$.
\end{definition}
\begin{definition} For a leaf $i\in[n]$ with \grm{root-to-leaf} path $(u_0,\dots,u_L)$ and child positions $j_0,\dots,j_{L-1}$, the prover opens {\small \(\pi_d := \mathsf{Open}_\mathsf{VC}(\mathbf{z}_{u_d}, j_d)\)} for each level \(d=0,\dots,L-1\).  Then, its \textit{Verkle proof} is defined as {\small \(
\pi_i^{\mathrm{V}} := \big(\pi_0,\dots,\pi_{L-1}\big).
\)  }
\end{definition}
Given {\small $(C^{\mathrm{V}}\!, i, x_i, \pi_i^{\mathrm{V}})$}, the verifier walks the same MT path, checks
{\small $\mathsf{Ver}_\mathsf{VC}$}$(C_{u_d}, j_d, z_{u_d,j_d}, \pi_d)\!=\!1, \forall d$.
\begin{remark}
\label{th:verklecomplexity}
A VT over $n$ \grm{leaves} % with depth $\lceil \log_B n\rceil$
has $\rev{\mathcal{O}}(\log_B n)$ proof size and $\rev{\mathcal{O}}(\log_B n)$ verification time.
\end{remark}
This states VTs remove the {\small$B$}-1 factor in Remark~\ref{th:merklecomplexity} and enable large-arity trees with succinct proofs.

\section{Problem \rev{statement} \label{sec:problem}}

\begin{figure}[!t]
    \centering
    \includegraphics[width=1.0\linewidth]{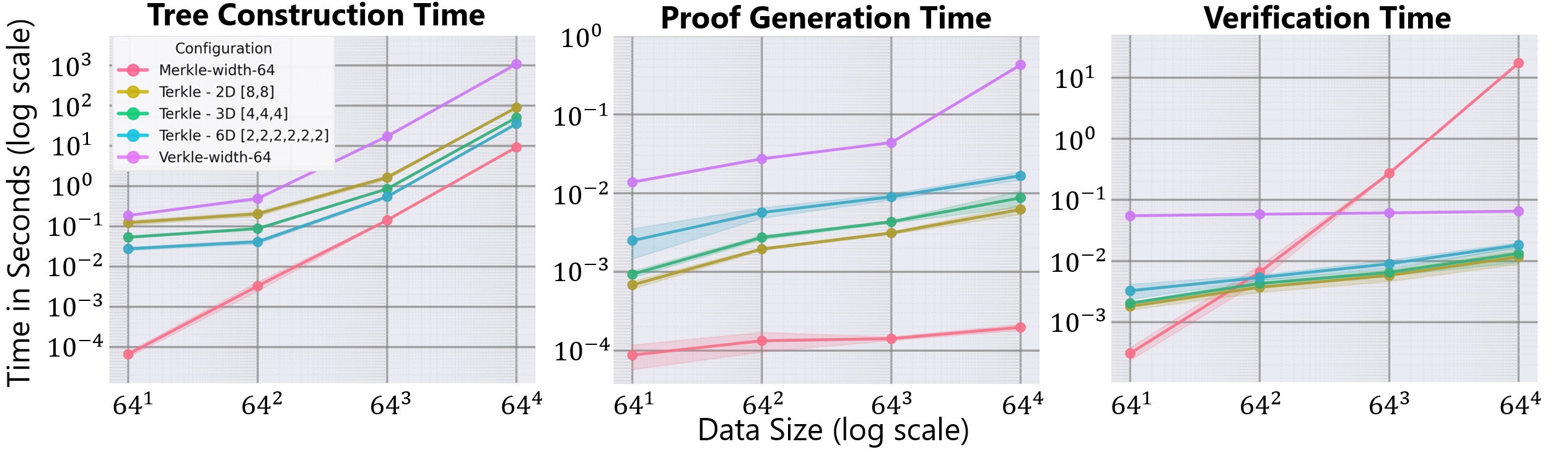}
    \caption{\small \textbf{How do Terkle trees scale better than Merkle and Verkle trees?} Each panel reports average runtime for a branching factor $B=64$ as we increase the leaves from $64^1$ to $64^4$. \emph{(Left) Tree construction time:} Merkle (pink) is fastest to build, while Verkle (purple) is two orders of magnitude slower at $64^4$. Terkle provides up to \textbf{29$\times$} \grm{speedup} compared to Verkle and the smoothest scaling. \emph{(Middle) Proof generation time:} Terkle reduces the proving time by up to \textbf{67$\times$} compared to Verkle. Merkle is the fastest but has a \textbf{63$\times$} larger proof size without privacy guarantees as $B=64$. \emph{(Right) Verification time:} Merkle verification cost increases steeply with data size since it must process $(B-1)$ sibling hash proofs per level while others need only a few proofs for the entire path. Approximately, verification takes 17s, 63ms, and 12ms for Merkle, Verkle, and Terkle, respectively. Hence, we speed them up by \textbf{1416}$\times$ and \textbf{14}$\times$ respectively. Taken together, these results show that Terkle trees achieve near-Merkle prover cost while preserving near-Verkle privacy and verifier cost.}\label{fig:treeresult}
\end{figure}

Consider LLM $\mathcal{M}\!:\!\mathcal{X} \!\times\!\Theta\!\to\!\mathcal{Y}$ with parameters $\theta_{\mathcal{M}}\!\in\!\Theta$ and layers $\mathcal{L}$. For input $x\sim\mathcal{X}$, an honest inference produces output $y\!=\!\mathcal{M}(x;\theta_{\mathcal{M}})$ and a collection of activation tensors $T_{\mathcal{M}}^{x}\!=\!\{T_\ell(x)\}_{\ell=1}^\mathcal{L}$, where $T_\ell(x)\!\in\!\mathbb{F}^{d^{(\ell)}_1\times \cdots \times d^{(\ell)}_{m_\ell}}$. As shown in Fig.~\ref{fig:sysarch}, a trusted setup uses security parameter $\lambda$ to form structured reference string $\textsf{srs}$ for a tensor commitment scheme, then publishes it while hiding its trapdoor(s). Untrusted compute-rich prover $\mathcal{P}$ (e.g., a GPU server) receives $x$, runs a (possibly altered) model to produce response $\tilde{y}$ and activation tensors $\tilde{T}^{x}$ \rev{(its transcript)}, and returns to light verifier $\mathcal{V}$ pair $(\tilde{y}, \tilde{C})$, where $\tilde{C}$ is \grm{the} commitment to $\tilde{T}^{x}$. Later, upon a challenge on index set $I$, $\mathcal{P}$ must provide proofs $\pi_I$ \rev{that let} $\mathcal{V}$ \rev{audit} $(\tilde{y},\tilde{C})$ \rev{against} {\small $(\mathcal{M},\theta_{\mathcal{M}})$} and $x$.

\textit{As an example}, let $x$ be a clinical record and $\mathcal{M}$ a medical LLM hosted by a third-party cloud. A light hospital workstation $\mathcal{V}$ sends $x$ to the cloud $\mathcal{P}$ and receives a diagnosis report $y$ plus a short commitment $C$ to all activations $T_{\mathcal{M}}^{x}$. Re-running $\mathcal{M}$ with local GPUs is expensive, but accepting a wrong or forged $y$ is disastrous. Instead, $\mathcal{V}$ uses only $\textsf{srs}$ \grm{and} tensor queries (e.g., selected layers) to \grm{check} the committed computation matches \grm{the} original $(\mathcal{M},\theta_{\mathcal{M}})$, without seeing all hidden states.

\textbf{Our objective is} to design a verifiable inference protocol $(\mathcal{P},\mathcal{V})$ such that, for any input $x$, completeness and \rev{output} soundness hold for any adversarial PPT prover $\mathcal{P}$.

\textit{Completeness:} If $\mathcal{P}$ honestly runs $\mathcal{M}(x;\theta_{\mathcal{M}})$ and commits to $T_{\mathcal{M}}^{x}$ to obtain $(y, C, \pi)$ via the protocol, then: $\Pr\!\big[\mathcal{V}(\mathsf{pp}, x, y, C, \pi)=1\big] = 1.$

\textit{\rev{Output soundness}:} For any malicious $\mathcal{P}$, \rev{{\small $\Pr\left[\mathcal{V}(\mathsf{pp}, x, \tilde{y}, \tilde{C}, \tilde{\pi})=1\wedge\tilde{y}\neq\mathcal{M}(x;\theta_{\mathcal{M}})\right] \leq \varepsilon$}, where $\varepsilon$, the worst case of Prop.~\ref{prop:guarantee} over wrong-output transcripts, is statistical and challenge-distribution dependent. Only commitment checks have $\mathrm{negl}(\lambda)$} \rev{error (App.~\ref{app:proof}). They do not establish provenance}.

Our work, in Fig.~\ref{fig:sysarch}, has \textbf{three types of entities}:\\
\textbf{\underline{The Trusted Entity}} runs a one-time procedure with $\lambda$ to generate $\mathsf{pp}$ containing $\textsf{srs}$. Besides $\mathsf{pp}$, it publishes model $(\mathcal{M}, \theta_{\mathcal{M}})$, and data $x$.\\
\textbf{\underline{The Prover}} $\mathcal{P}$ is a powerful server that receives an input $x \in \mathcal{X}$, runs an LLM \(y\!=\!\mathcal{M}(x; \theta_{\mathcal{M}}),\) and produces the associated activations $T_{\mathcal{M}}^{x}\!=\!\{T_\ell(x)\}_{\ell=1}^\mathcal{L}$. As detailed in Sec.~\ref{subsec:tc}, it interpolates $T_\ell(x)$ into \grm{a} multivariate polynomial $f_{T_\ell}$, commits, organizes those in a tree, and returns a short commitment $C_f$, the claimed output $y$, and \grm{later,} upon receiving a challenge, proofs $\pi^{\omega}_C$ for selected points.\\
\textbf{\underline{The Verifier}} $\mathcal{V}$ has limited resources. It gets $(\textsf{srs},\mathsf{pp})$ and model $(\mathcal{M},\theta_{\mathcal{M}})$ \rev{but} cannot re-run $\mathcal{M}$ on $x$. Given $x$ and $(y, C_f)$, it selects \grm{a set} $I$ of layers $\{\boldsymbol{\omega}_i\}_{i \in I}$, requests their proofs $\pi^{\omega}_C$, and runs $\rev{\mathcal{O}}(|I|)$ pairing checks. \rev{It accepts only if all checks pass, which serve four roles: \textit{binding}, as $C_f$ precedes the challenge; \textit{homomorphic audits} of linear relations among committed tensors, run on commitments; \textit{arithmetic audits}, which check a challenged entry's inputs against $C_f$ and recompute it from} \rev{$\theta_{\mathcal{M}}$ up to tolerance $\epsilon_{\mathrm{tol}}$ (App.~\ref{app:audits}); and \textit{coverage}}, \rev{as only sampled violations can be detected}.

\textbf{Threat Model\rev{:}} \rev{$\mathcal{P}$ may inject} noise into weights or activations $T_\ell(x)$, \rev{\grm{run a}} cheaper proxy model $\widetilde{\mathcal{M}}$ with {\small $\theta_{\mathcal{\widetilde{M}}}\!\neq\!\theta_{\mathcal{M}}$}, or \rev{tamper \grm{with}} output $y$ \rev{(}App.~\ref{app:attacks}\rev{), and commit consistently to the incorrect transcript}.

\section{Method}\label{sec:method}

Our pipeline comprises three key components: TensorCommitments binding the entire inference \rev{transcript}, Terkle Trees aggregating commitments of evolving LLM states, and a layer selection algorithm finding the most sensitive parts of LLMs.

\label{subsec:tc}
\textbf{\underline{TensorCommitments}} (TCs) are tensor-native generalizations of VCs, but TC commits to a \emph{multivariate} polynomial\rev{, instantiating the multivariate KZG commitment~\cite{pst13} over tensor grids}. We reuse the notation from the VCs detailed in Def.~\ref{def:vc}: a finite field $\mathbb{F}$, curves $G, G_T$ of prime order $p$ with pairing $e : G \times G \to G_T$, and generator $g \in G$. TC's input is a tensor $T$, with order $m$ and shape $\mathbf{d} = (d_1,\dots,d_m)$, denoted as: {\small $
T\!\in\!\mathbb{F}^{d_1 \times \cdots \times d_m},
T[\omega_1,\dots,\omega_m]\!\in\!\mathbb{F},
\omega_j\!\in\![d_j]$}, where {\small $\boldsymbol{\omega}\!=\!(\omega_1,\dots,\omega_m)$} is an evaluation point in the domain $\mathbf{\Omega}$.
Following this notation, we define the TCs as follows.

\begin{definition}\label{def:tc}
A \emph{tensor commitment scheme} $\mathsf{TC}$ for the tensors of order $m$ consists of four multivariate polynomial-based algorithms listed as {\small $\mathsf{Setup}_{\mathsf{TC}}(.), \mathsf{Com}_{\mathsf{TC}}(.), \mathsf{Open}_{\mathsf{TC}}(.),$} and {\small$\mathsf{Ver}_{\mathsf{TC}}(.)$}
\end{definition}

\underline{{\small $\mathsf{Setup}_{\mathsf{TC}}(\lambda, \mathbf{d}) \to \mathsf{pp}$}:} Given security parameter $\lambda$ and shape $\mathbf{d}$, {\small $\mathsf{Setup}$} forms random \textit{trapdoors \grm{per~axis}} $\tau_1,\dots,\tau_m\!\in\!\mathbb{F}$ and builds an $\mathsf{srs}$ for all monomials up to $d_j\!\in\!\mathbf{d}$ degrees along each axis $j\!\in\!m$.
For each monomial with tuple of $m$-variable degrees $\mathbf{i} = (i_1,\dots,i_m)$, where $0\le i_j\leq d_j$, we denote the per-monomial trapdoors as \(\boldsymbol{\tau}_{\mathbf{i}}=\tau_1^{i_1} \cdots \tau_m^{i_m} \in \mathbb{F}\), and \(g_{\mathbf{i}}=g^{\boldsymbol{\tau}_{\mathbf{i}}} \in G\). Then $\mathsf{srs}$ is $(g_{\mathbf{i}})_{\mathbf{i} \in \mathcal{I}}$ where {\small$\mathcal{I} := \{0,\dots,d_1\!-\!1\} \times \cdots \times \{0,\dots,d_m\!-\!1\}$}. $\mathsf{srs}$ gives public access to $g^{\tau_1^{i_1}\cdots\tau_m^{i_m}}$ for every monomial $X_1^{i_1}\cdots X_k^{i_m}$, without revealing the trapdoors $\tau_1,\dots,\tau_m$. Then, it publishes parameters $\mathsf{pp} = (\mathsf{srs}, \mathbf{\Omega})$ for domain $\mathbf{\Omega}$.

{\small \underline{$\mathsf{Com}_{\mathsf{TC}}(\mathsf{pp}, T) \to C_T$}:} Given function values $T$ on grid $\mathbf{\Omega}$, {\small $\mathsf{Com}$} forms \grm{a} \textit{multivariate polynomial} $f_T \in \mathbb{F}[X_1,\dots,X_m]$ with degree $d_j$ in each variable $X_j$, where {\small $1\!\leq\!j\!\leq\!m$}, by Lagrange interpolation such that {\small \( f_T\bigl(\boldsymbol{\omega}\bigr)\!=\!T[\boldsymbol{\omega}], \forall \boldsymbol{\omega}\!\in\!\mathbf{\Omega}\)}. We can express $f_T$ in the monomial basis, as denoted in:
\begin{equation}
f_T(X_1,\dots,X_m)=\sum_{k_1=0}^{d_1-1} \cdots \sum_{k_m=0}^{d_m-1}
a_{k_1,\dots,k_m}\,
X_1^{k_1}\cdots X_m^{k_m}.\nonumber
\end{equation}
The coefficients are interpolated from $T$ on the grid $\mathbf{\Omega}$. Then, it commits by evaluating $f_T(\tau_1,\dots,\tau_m)$, shown as:
{\small \begin{equation}
    C_T=g^{f_T(\tau_1,\dots,\tau_m)}=g^{\sum_{k_1=0}^{d_1-1} \cdots \sum_{k_m=0}^{d_m-1} a_{k_1,\dots,k_m}\, \tau_1^{k_1}\cdots\tau_m^{k_m}} =\prod_{k_1=0}^{d_1-1} \cdots \prod_{k_m=0}^{d_m-1}
    \bigl( g^{\tau_1^{k_1}\cdots\tau_m^{k_m}} \bigr)^{ a_{k_1,\dots,k_m}} \in G.\nonumber
\end{equation}}

Using $\mathsf{srs}$ and linearity in the exponent, we commit without knowing secrets $\{\tau_j|\forall j\!\leq\!m\}$. $C_T$ is a single group element serving as a succinct handle for the entire $T$.

\underline{{\small $\mathsf{Open}_{\mathsf{TC}}(\mathsf{pp}, T, \boldsymbol{\omega}) \to (y, \pi_{\boldsymbol{\omega}})$}:}
Given public parameters $\mathsf{pp}$, tensor $T$, and challenge point $\boldsymbol{\omega} = (\omega_1,\dots,\omega_m)$, it outputs $y = f_T(\boldsymbol{\omega})$ and proof $\pi_{\boldsymbol{\omega}}$ to show $y$ is consistent with previously published $C_T$. To build $\pi_{\boldsymbol{\omega}}$ without revealing $f_T$, we iteratively ``peel off'' each variable using a univariate polynomial division. We call the remainder polynomial $r_{i}(.)$ a polynomial independent of $\{X_1,\dots,X_i\}$. We set $r_0 := f_T$ and iteratively divide $r_{i-1}(.)$ by $(X_i - \omega_i)$, as in:
\begin{equation}
        r_{i-1}(X_i,\dots,X_m)=q_i(X_i,\dots,X_m)\,(X_i - \omega_i)+r_i(X_{i+1},\dots,X_m),
    \label{eq:iter-div-step}\nonumber
\end{equation}

where $q_i$ is the $i$-th quotient. After $m$ steps, the final remainder is scalar $r_m\!=\!y$. For this recursion, we use Algorithm~\ref{algo:polydiv}  (\grm{App.~\ref{app:polydiv}}). Using $\mathsf{srs}$, the prover commits to each quotient via:
\rev{\[
\pi_{\omega_i} := g^{q_i(\tau_i,\dots,\tau_m)} \in G,
\quad i = 1,\dots,m,
\]}\ignorespaces
Concretely, if we write the quotient in monomial form as:
{\small \begin{equation}
    q_i(X_i,\dots,X_m)=\sum_{k_i=0}^{d_i-1} \cdots \sum_{k_m=0}^{d_m-1}b^{(i)}_{k_i,\dots,k_m}X_i^{k_i}\cdots X_m^{k_m},\nonumber
\end{equation}}
\grm{using} $\mathsf{srs}$ and \grm{linearity}, we obtain the opening as follows:
\begin{equation}
    \pi_{\omega_i}=\prod_{k_i=0}^{d_i-1} \cdots \prod_{k_m=0}^{d_m-1}
    \bigl( g^{\tau_i^{k_i}\cdots\tau_m^{k_m}} \bigr)^{b^{(i)}_{k_i,\dots,k_m}}.\nonumber
\end{equation}

The whole proof is $\pi_{\boldsymbol{\omega}}\!:=\!(\pi_{\omega_1},\dots,\pi_{\omega_m})$. \rev{Summing the} \rev{$m$} division \rev{steps telescopes ($r_0-r_m\!=\!f_T-y$) into the coordinate-wise identity {\small $f_T(X_1,\dots,X_m)-y=\sum_{i=1}^m q_i(X_i,\dots,X_m)\,(X_i-\omega_i)$}, whose evaluation at $\boldsymbol{\omega}$ gives $y\!=\!f_T(\boldsymbol{\omega})$}.

\underline{{\small $\mathsf{Ver}_{\mathsf{TC}}(\mathsf{pp}, C_T, \boldsymbol{\omega}, y, \pi_{\boldsymbol{\omega}}) \to \rev{\{0,1\}}$}:} Given a commitment $C_T$ and a challenge $\boldsymbol{\omega}$, it checks $y\!=\! f_T(\boldsymbol{\omega})$ for the \emph{right} $f_T$ underlying $C_T$ via $\mathsf{pp}$ and proof $\pi_{\boldsymbol{\omega}}$ without seeing $T$ or $f_T$, by pairing {\small \rev{\(e(C_T\!\cdot\!g^{-y},\, g)\!\stackrel{?}{=}\!\prod_{i=1}^m e\bigl(\pi_{\omega_i},\, g^{\tau_i - \omega_i}\bigr),\)}} \rev{where $g^{\tau_i-\omega_i}\!=\!g^{\tau_i}\!\cdot\!g^{-\omega_i}$ uses the $\mathsf{srs}$ element of monomial $X_i$.} If it holds, it outputs $1$. Evaluating \rev{the identity above at the trapdoors} and using {\small $e(g^a,g^b)\!=\!e(g,g)^{ab}$}, \rev{both sides equal {\small $e(g,g)^{f_T(\tau_1,\dots,\tau_m)-y}$}, so honest proofs always pass}. App.~\ref{app:proof} \rev{proves correctness and reduces evaluation binding to the $q$-Strong Bilinear Diffie-Hellman ($q$-SBDH) assumption, as in~\cite{pst13}}.

\underline{\textbf{Terkle Trees}} (TTs) extend TCs from a single inference to an entire LLM dialogue. Authenticating these with a single root is better than sending separate commitments for each prompt. New states and their proofs can reuse the structure of previous ones, instead of forcing a full recomputation. Verifiable LLMs need a compact way to bind \emph{all} prompts and outputs. TTs address this by arranging TCs in a high-arity tree: each node commits to an \grm{inference,} and the root acts as a single global commitment for the entire dialogue. The prover opens only a few root-to-leaf paths selected by our layer policy, and the verifier checks consistency from leaves to the root. As demonstrated in Fig.~\ref{fig:trees}, unlike MTs and VTs, TTs preserve tensor structure, enabling localized tensor-shaped queries with low prover and verifier cost.

\begin{definition}
A \emph{Terkle Tree} on a payload \grm{$\mathbf{X}$,} with shape  $\boldsymbol{d}\!=\!(d_1,\dots,d_m)$ and order $m$, is a $B$-ary tree where $B\!=\!\prod_{t=1}^{m} d_t$. Each internal node $u$ has $B$ children indexed by $\boldsymbol{k}\!=\!(k_1,\dots,k_m)\!\in\!\prod_{t=1}^{m} [d_t]$, using $c(u,\boldsymbol{k})$ for the child at position $\boldsymbol{k}$. We gather these children in a tensor $\mathbf{Z}_u$ of shape $\boldsymbol{d}^{(u)}$ whose entries satisfy $\mathbf{Z}_u[\boldsymbol{k}]\!=\!X[\boldsymbol{k}]$ if $c(u,\boldsymbol{k})$ is a leaf, and $\mathbf{Z}_u[\boldsymbol{k}]\!=\!C_{c(u,\boldsymbol{k})}$ if $c(u,\boldsymbol{k})$ is an internal node with commitment $C_{c(u,\boldsymbol{k})}$. A node stores $C_u\!:=\! \mathsf{Com}_\mathsf{TC}(\mathbf{Z}_u)$, and the \emph{Terkle root commitment} to $\mathbf{X}$ is $C^{\mathrm{T}} := C_{\rev{u_0}}$.
\end{definition}
\begin{definition}
Let $i \in [n]$ index a leaf in a TT with root-to-leaf path $(u_0,\dots,u_L)$, \rev{where} $u_L$ stores the commitment of entry $X[\boldsymbol{k}]$. For each level $d\!\in\!\{0,\dots,L-1\}$, let $\boldsymbol{\kappa}_d$ be the local child multi-index such that $u_{d+1}\!=\! c(u_d,\boldsymbol{\kappa}_d)$. \grm{The} \textit{Terkle membership proof} for leaf $i$ is the sequence
$\pi_i^{\mathrm{T}}\!:=\!(\pi_0,\dots,\pi_{L-1})$, where
$\pi_d\!:=\!\mathsf{Open}_\mathsf{TC}(\mathbf{Z}_{u_d},\boldsymbol{\kappa}_d)$ is a TC proof for entry $\mathbf{Z}_{u_d}[\boldsymbol{\kappa}_d]$ at level $d$.
\end{definition}
Given $(C^{\mathrm{T}}, i, X[\boldsymbol{k}], \pi_i^{\mathrm{T}})$, the verifier follows the same root-to-leaf path $(u_0,\dots,u_L)$ as in VTs. At each level $d$, it uses the public multi-index $\boldsymbol{\kappa}_d$ to check ${\small \mathsf{Ver}_{\mathsf{TC}}\big(C_{u_d}, \boldsymbol{\kappa}_d, \mathbf{Z}_{u_d}[\boldsymbol{\kappa}_d], \pi_d\big)\!=\!1}$ and that $\mathbf{Z}_{u_d}[\boldsymbol{\kappa}_d]$ matches $C_{u_{d+1}}$. The proof is accepted if all checks pass and $C_{u_0} = C^{\mathrm{T}}$. TT verifications mirror VTs' but replace 1D vector openings with tensor-native multi-index openings.
\begin{proposition}
Consider a TT over $n$ items with arity $B$ across $m$ tensor
dimensions, so each internal node arranges its $B$ children as an order-$m$
tensor with side length $d$ per axis ($B = d^m$). The depth is
$L = \lceil \log_{B^m} n \rceil$. \textit{(\textbf{i})}
A Terkle membership proof $\pi_i^{\mathrm{T}}$ for any leaf $i \in [n]$ contains
one tensor opening per level, so
\(
|\pi_i^{\mathrm{T}}|
= \mathcal{O}(L)
= \mathcal{O}\!\big(\tfrac{1}{m}\log_B n\big),
\)
i.e., the same logarithmic dependence on $n$ as VTs but with $1/m$ as many levels. \textit{(\textbf{ii})} The cost to prove a TT node is \(\Theta\!\big(\tbinom{m + d}{m}^2\big)\) for \(d\!=\!B^{1/m}\), while a VT node with the same $B$ (flattened) children costs $\Theta(B^2)$. For large $d$, $\tbinom{m + d}{m} \sim d^m/m!\!=\!B/m!$, so TTs cost $\Theta(B^2/(m!)^2)$. Thus, for fixed $B$ and $m>1$, TTs yield proofs shorter by $\Theta(1/m)$ and per-node prover work smaller by $\Theta(1/(m!)^2)$ than VTs.
\end{proposition}

\textbf{\underline{Layer Selection Algorithm:}}
\label{subsec:layer-selection} LLMs have many blocks, indexed by $i$, with weights
$W_i$ and activations $T_i$. Verifying the whole model (e.g., 70B) is prohibitive. As we empirically show in Fig.~\ref{fig:layersensitivies}, layer-wise sensitivity to noise is highly non-uniform: only a subset of layers induces large shifts in outputs. Hence, we design a robustness-aware layer selection scheme to verify \grm{the} most impactful blocks. Each layer $i$ is assigned a nonnegative \emph{benefit}
score $\nu_i\!\ge\!0$ and a verification \emph{cost} $\phi_i\!>\!0$. Following
$\alpha$-pruning~\cite{alphapruning}, we use its importance metric in the
\emph{reverse} direction: instead of pruning low-importance layers, we
prioritize verifying those with the highest $\nu_i / \phi_i$ under a fixed budget. For each block $i$, we form the weight correlation matrix $\rev{W_i^\top W_i}\!\in\!\mathbb{R}^{d_i \times d_i}$ with eigenvalues sorted as {\small $0 \!\le\! \Lambda_1\! \le\! \dots\! \le\! \Lambda_{d_i}$} to obtain the empirical spectral density $\mu_{\rev{i}}\! =\! \frac{1}{d_i}\sum_j \delta_{\Lambda_j}$. Then, we fit a power-law tail {\small $P(\Lambda) \!\propto \!\Lambda^{-\alpha_i}$} on the largest $\chi$ eigenvalues via a Hill estimator as:
{\small $
\widehat{\alpha}_i
\;=\;
1 + \chi/ \sum_{r=1}^{\chi}\ln(
\Lambda_{d_i-r+1}/\Lambda_{d_i-\chi})
.
$}
Smaller $\widehat{\alpha}_i$ (heavier tail) empirically corresponds to layers strongly changing the output. We normalize and invert these scores and weight them by parameter size $P_i$, as in {\small \(
\nu_i
=
(\widehat{\alpha}_{\max}-\widehat{\alpha}_i)/(\widehat{\alpha}_{\max}-\widehat{\alpha}_{\min})
\cdot (P_i/\sum{}_j P_j).
\)}
\noindent $\nu_i$ is largest for blocks where an
adversary gains the most leverage and a verifier spends \grm{the most}.

\textbf{Problem 4.5:\label{problem}} Let $\phi_i>0$ be the cost (e.g., the number of parameters) of verifying layer block $i$, and $M$ verifiers with budgets $\beta_1,\dots,\beta_M$ can attest \grm{to} at most one contiguous interval from $\mathcal{L}$ layers. By setting a binary variable, $\gamma_{k,i}$, to assign a block $i$ to verifier $k$ with consecutive assignments of start $s_{k,i}$ and end $\rev{o_{k,i}}$, we maximize total verified robustness $\sum_{k}\sum_{i} \nu_i\,\gamma_{k,i}$ with an integer linear program as below. A dynamic program (DP) solves this, as given in App.~\ref{app:dp}, with prefix sums and a monotone sliding window in $\mathcal{O}(M\mathcal{L})$ \grm{runtime} and $\mathcal{O}(\mathcal{L})$ memory usage via a rolling array in $k$.

\begin{figure}[t]
\centering

\begin{minipage}[t]{0.52\textwidth}
{\small
\begin{align}
\max_{\gamma,s,\rev{o}}\quad
& \sum_{k=1}^M\sum_{i=1}^\mathcal{L} \nu_i\,\gamma_{k,i}&\nonumber
\\[0.25em]
\text{s.t.}\quad& \sum_{i=1}^\mathcal{L} \phi_i\,\gamma_{k,i}\le \beta_k,
\quad \forall k\in[M]
&\text{\scriptsize \tt <limited budget>}
% \label{eq:ilp-budget}
\nonumber\\
& \sum_{k=1}^M \gamma_{k,i} \le 1,
\quad \forall i\in[\mathcal{L}]
&\text{\scriptsize \tt <no overlap>}
% \label{eq:ilp-no-overlap}
\nonumber\\
& \sum_{i=1}^\mathcal{L} s_{k,i} \le 1,
\quad \forall k\in[M]
&\text{\scriptsize \tt <one start>}
% \label{eq:ilp-one-start-end}
\nonumber\\
& \sum_{i=1}^\mathcal{L} \rev{o_{k,i}} \le 1,
\quad \forall k\in[M]
&\text{\scriptsize \tt <one end>}
% \label{eq:ilp-one-start-end}
\nonumber\\
& \gamma_{k,i}-\gamma_{k,i-1}=s_{k,i}-\rev{o_{k,i-1}},
&
% \label{eq:ilp-contiguity}
\nonumber\\
& \gamma_{k,1}=s_{k,1},
\quad \forall k,\; i=2{:}\mathcal{L}
&\text{\scriptsize \tt <contiguity>}
% \label{eq:ilp-contiguity}
\nonumber\\
& \gamma_{k,i},s_{k,i},\rev{o_{k,i}}\in\{0,1\}.
\nonumber
\end{align}
}
\end{minipage}
\hfill
\begin{minipage}[t]{0.46\textwidth}
\vspace{0pt}
\centering
\includegraphics[width=0.99\linewidth]{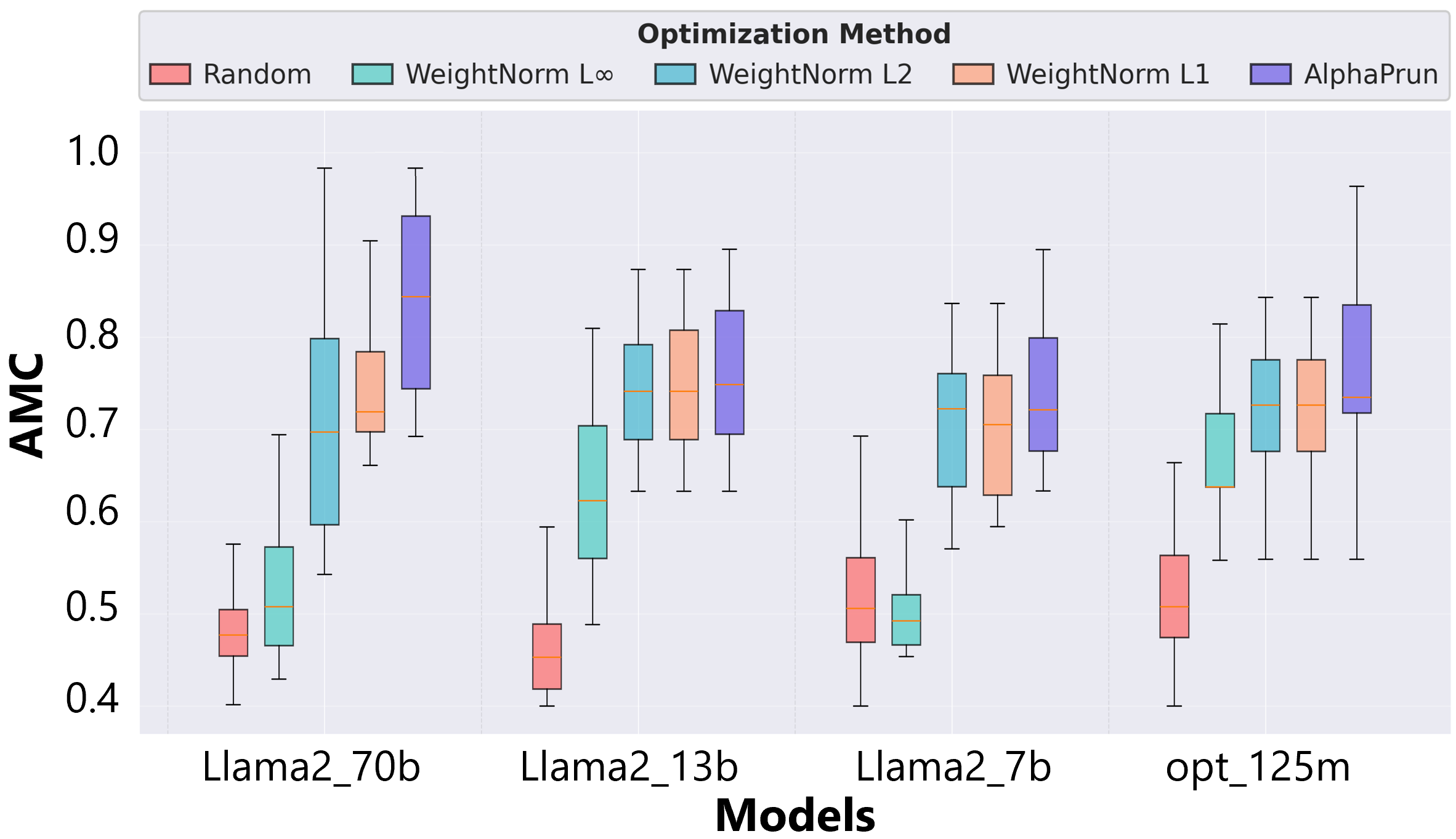}
\caption{\small \textbf{Do spectral scores reveal critical layers?} We carry out $1{,}000$ noise-injection attacks per model, randomly altering half of the layers in $10$ prompts and $100$ seeds, ensuring the output token changes.
For a given benefit function, we solve Problem~\ref{problem}.5 and measure AMC, defined in Sec.~\ref{sec:experiments} as the fraction of attacked layers falling inside the selected interval(s).
Our objective consistently outperforms alternatives, improving median coverage by up to \textbf{75\%}, with larger gains on bigger models.}
\label{fig:amcexps}
\end{minipage}
\end{figure}

\rev{\textbf{\underline{Randomized Selection against Adaptive Provers:}} A deterministic selector audits layer $\ell$ with probability $p_\ell\!=\!\mathbf{1}[\ell\!\in\!S_\alpha]\!\in\!\{0,1\}$, where $S_\alpha$ is the solution of Problem~4.5 with $\alpha$-scores, so a policy-aware prover can deviate outside $S_\alpha$, where it is never audited. After the root is published, $\mathcal{S}_\gamma$ picks $S_\alpha$ with probability $\gamma\!\in\![0,1]$ (not the ILP's $\gamma_{k,i}$), else a uniformly random feasible allocation, which covers layer $\ell$ with probability $c_\ell$. Then, {\small $p_\ell\!=\!\gamma\,\mathbf{1}[\ell\!\in\!S_\alpha]\!+\!(1\!-\!\gamma)\,c_\ell\!\ge\!(1\!-\!\gamma)\,c$} with $c\!:=\!\min_\ell c_\ell$, positive iff every layer is in some feasible allocation ($\phi_\ell\!\le\!\max_k\beta_k$ $\forall\ell$). This floor holds even if the prover knows the policy, $\alpha$-scores, and budgets, and lower-bounds $p_\ell$ in Prop.~\ref{prop:guarantee} (App.~\ref{app:adaptive}).}

{\revcolor \stepcounter{theorem}% the hard-coded Problem 4.5 takes number 4.5
\begin{proposition}\label{prop:guarantee}
Let $T$ be an incorrect committed transcript and $\delta(T)$ the probability that $\mathcal{V}$ samples a model relation of $T$ with error above $\epsilon_{\mathrm{tol}}$. If sampled relations are checked correctly up to $\epsilon_{\mathrm{tol}}$, then {\small $\Pr[\mathcal{V}\text{ accepts }T]\le 1-\delta(T)+\mathrm{negl}(\lambda)\le 1-p_\ell[1-(1-\rho_\ell)^k]+\mathrm{negl}(\lambda)$} for any violated layer $\ell$ selected by the policy with probability $p_\ell$, where a fraction $\rho_\ell$ of its auditable coordinates has error above $\epsilon_{\mathrm{tol}}$ and $k$ of them are sampled uniformly, independently of the layer draw (proof in App.~\ref{app:guarantee-proof}).
\end{proposition}}

\section{Experiments}\label{sec:experiments}
We assess \textit{TC} against four \underline{benchmarks}: \textbf{zkLLM} \cite{zkllm} is a cryptographic approach for end-to-end correctness with ZK proofs, the strongest integrity with the highest overhead. \textbf{SVIP} \cite{sun2024svip} is a learned inspector training an auxiliary classifier for tampering, fast but \grm{model-specific}. \textbf{TOPLOC} \cite{TOPLOC} uses late-layer top-k activations to detect deviations with small overhead. \textbf{Raw Activations} \grm{uses} last-layer activations for a high-accuracy, non-cryptographic upper bound under generous bandwidth.

\begin{figure}[t]
    \begin{minipage}[t]{0.5\linewidth}
        \centering
        \includegraphics[width=1.1\linewidth]{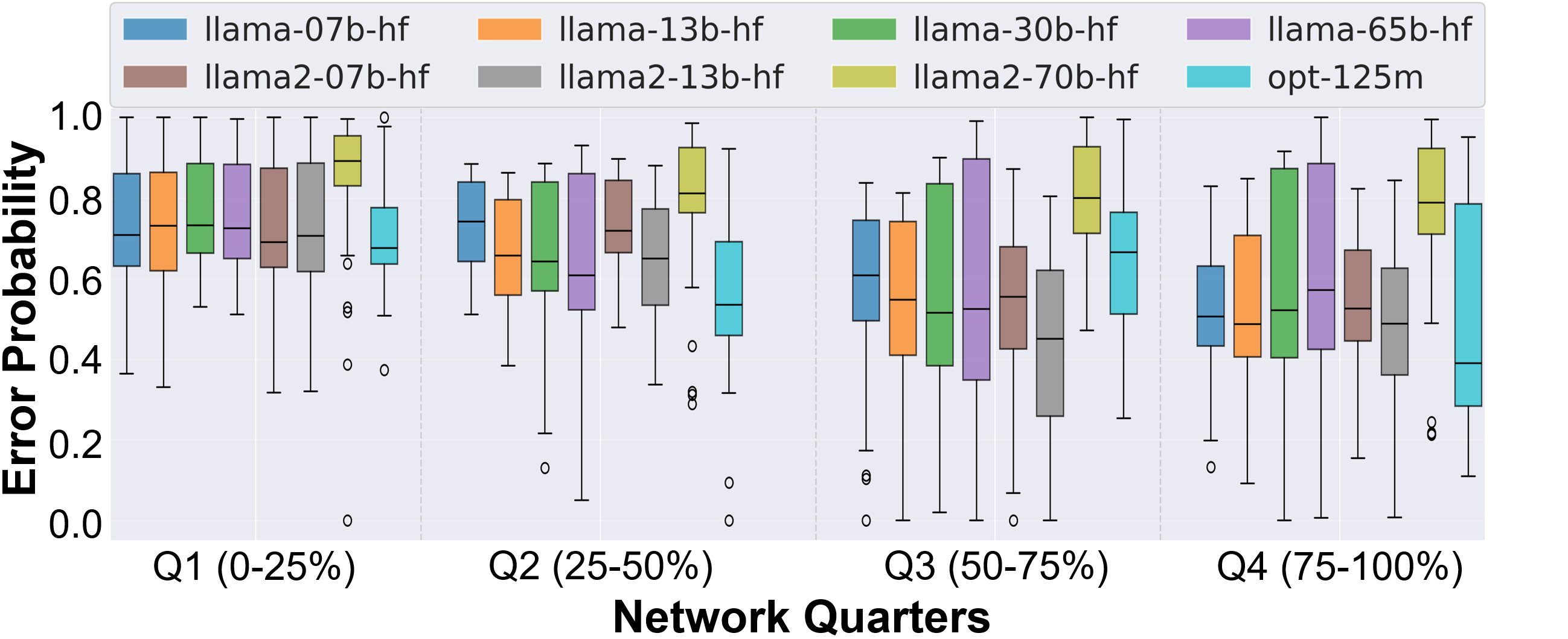}
        \caption{\small \textbf{Where do LLMs break most easily?}
        Gaussian noise, scaled by each layer's $\ell_2$ norm, is injected one layer at a time in each model to measure the change in output-token probability.
        \textit{\textbf{Sensitivity varies significantly across depth and architectures}}, as shown by per-layer sensitivities aggregated over network quarters.
        Notably, sensitivity peaks differently for LLaMA2-7B \cite{touvron2023llama2} \textit{(brown)} in Q2, LLaMA2-13B \cite{touvron2023llama2} \textit{(gray)} in Q1, and OPT-125M \cite{opt} \textit{(turquoise)} in Q4.}
        \label{fig:layersensitivies}
    \end{minipage}
    \hfill
    \begin{minipage}[t]{0.455\linewidth}
        \centering
        \includegraphics[width=0.86\linewidth]{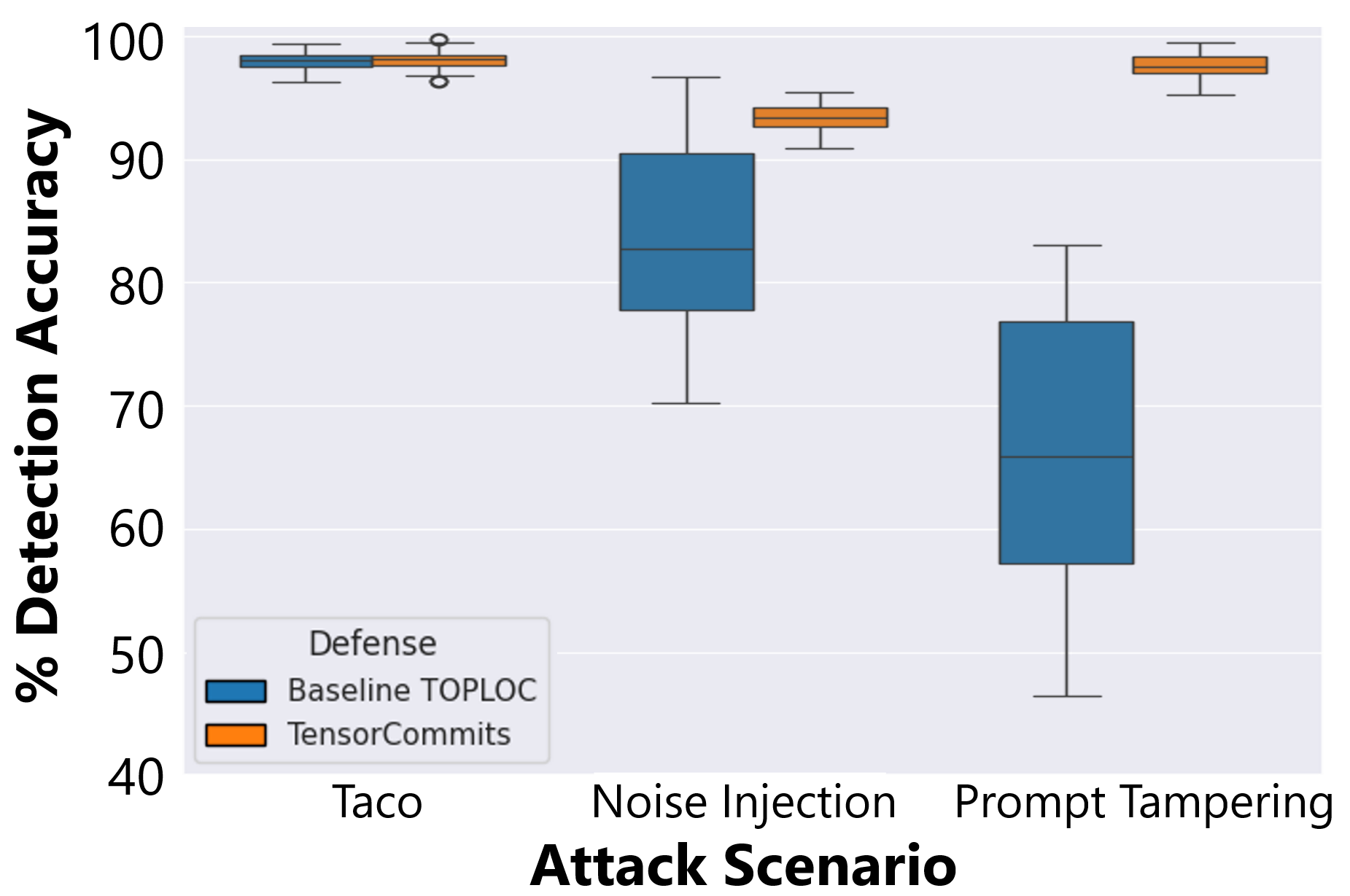}
        \caption{\small \textbf{How do we beat the SoTA under tailored attacks?}
       We conduct 100 LLM attacks with 10 seeds each, assessing TOPLOC \textit{(blue)} and TC \textit{(orange)} detection accuracy. TOPLOC's \emph{Taco} attacks are near-ceiling.
        \emph{Noise Injection} causes broad degradation, while TC boosts the median by \textbf{12}\%.
        \emph{Prompt Tampering} yields high variance, yet TC improves by \textbf{48}\%.
        SoTA misses slow-burn deviations and has fixed thresholds vulnerable to forgeries.}
        \label{fig:toploccomparison}
    \end{minipage}
\end{figure}

We monitor the benchmarks along three \underline{metrics}: system overhead, robustness, and scalability. \textbf{(Overhead)} For each method, we report the commitment size per token, prover time, verification time, and GPU memory overhead for both prover and verifier. \textbf{(Robustness)} We report the attack detection accuracy and AMC for how well the chosen layers overlap with the targets. \textbf{(Scalability)} For all trees, we measure time to build, prove, and verify as the entries and tensor dimension grow. 

\textit{ \underline{Attack Manipulation Coverage (AMC):}} Let $\mathcal{M}\!:\!\mathcal{X}\!\to\!\mathcal{Y}$ be a model with $\mathcal{L}$ layers, and let $\mathcal{A}$ be a family of inference-time attacks. For input $x\!\in\!\mathcal{X}$ and attack $A\!\in\!\mathcal{A}$, write $\tilde{\mathcal{M}}_A(x)$ for the attacked model output, and let $\xi(x,A)\!\subseteq\![\mathcal{L}]$ be the set of layers whose weights are explicitly manipulated by $A$.  A layer-selection strategy chooses a subset $S\!\subseteq\![\mathcal{L}]$ of layers to monitor. We focus on \emph{successful} attacks (i.e., {\small $\bigl\{\tilde{\mathcal{M}}_A(x)\!\neq\!\mathcal{M}(x)\bigr\}$}) and evaluate how well the selection $S$ overlaps with $\xi(x,A)$. Given distribution $\mathcal{X}$ over inputs and attack sampling distribution $\mathcal{Q}$ over $\mathcal{A}$, the AMC of $S$ is:
{\small \begin{equation}
\label{eq:amc-def}
\mathrm{AMC}(S)
\;:=\;
\mathbb{E}_{\substack{x \sim \mathcal{X},\, A \sim \mathcal{Q} \\ \{\tilde{\mathcal{M}}_A(x)\!\neq\!\mathcal{M}(x)\}}}
\left[
\frac{|\xi(x,A)\,\cap\, S|}{|\xi(x,A)|}
\right].\nonumber
\end{equation}}

\textit{How to verify long-horizon LLM states without excess cost?} In multi-turn LLM interactions, naively committing each state per prompt yields a forest of handles and expensive, unstructured openings. TTs instead \grm{use} a single root over the evolving states while organizing states as tensors, so the verifier can cheaply query spatially or temporally local regions (e.g., an image patch, video segment, graph neighborhood, or a specific chat span) via multivariate openings. Fig.~\ref{fig:treeresult} shows TTs match MTs’ prover cost while verifying up to \textbf{1416}$\times$ and \textbf{14}$\times$  faster than MTs and \grm{VTs,} respectively. TTs structurally align with modern LLM loads and keep the evolving state verification succinct and scalable.

\textit{Do last-layer fingerprints secure the inference?} TOPLOC, the strongest prior, misses ``slow-burn’’ attacks preserving final activations (e.g., early-layer noise, low-rank weight edits, prompt tampering with entity/adjective swaps). It checks only last-layer top-k activations and uses fixed thresholds tuned to its own Taco attacks detailed in App. \ref{app:attacks}. Worse, as given in Table~\ref{tab:benchmarkingresults}, the verifier must \textbf{re-run} the LLM with a GPU and reveal the activations. It is costly and incompatible with confidential data. In contrast, we \rev{audit selected layers over cryptographically bound openings}, without exposing full activations or \grm{re-running} inference, with \textbf{30\%} and \textbf{85\%} lower proving and verification time, respectively. In Fig.~\ref{fig:toploccomparison}, this translates to higher detection accuracy: we roughly match TOPLOC in their attacks, but improve detection by \textbf{12\%} under noise injection and by \textbf{48\%} under prompt forgery. \rev{The gap is widest when a provider approximates rather than skips a block, since audits recompute the block while TOPLOC checks only last-layer fingerprints, as in Table~\ref{tab:providermain} and as elaborated in App.~\ref{app:provider-attacks}.} \leavevmode{\revcolor Table~\ref{tab:deploycost} itemizes the full deployment cost. Its times cover every stage after inference, including the audit, which is what the 0.97\% and 0.12\% overheads measure. Per inference, the verifier gets one root, which does not grow with sequence length, and an $\mathcal{O}(\log_B n)$ Terkle membership proof (App.~\ref{app:deployment}). TC thus sits between fingerprints and zkLLM, whose 100\% detection comes from proving every operation. Our sampled verification is a different cost-guarantee trade-off.

\begin{table}[t]
\centering
\caption{\small \textbf{Can we cheaply verify an LLM inference?} \rev{TC ranks in the top two on every metric. Only TC and zkLLM keep the verifier GPU-free and the inference private, yet TC proves in 98.6\,ms instead of 803\,s and verifies in 12\,ms instead of 3950\,ms. Our prover and verifier add only \textbf{0.97\%} and \textbf{0.12\%} to a 10.165\,s LLaMA2-13B inference on an A100, with no verifier training and the smallest commitment (2\,B per token). Against TOPLOC, the strongest fingerprint, TC cuts proving and verification time by \textbf{30\%} and \textbf{85\%} and raises detection from 82.59\% to 96.02\%, second only to zkLLM's full-circuit proof.} \label{tab:benchmarkingresults}}
\resizebox{\linewidth}{!}{%
\begin{tabular}{lccccc}
\toprule
 \textbf{Metrics}& \textbf{zkLLM} & \textbf{SVIP} & \textbf{Raw Activations} & \textbf{TOPLOC}& \textbf{TensorCommits}\\
\midrule
Verifier - GPU Utilization ($\downarrow$) & \textbf{0 GB} & 1.394 GB & 24.81 GB & 71.32 GB & \textbf{0 GB} \\
Prover - GPU Utilization ($\downarrow$) & 23.1GB & 980MB  & 10KB  & 10KB  & 10KB  \\
Verification Time ($\downarrow$) & 3950 msec & \textbf{5.6 msec}  & 81 msec   & 81 msec & \underline{12 msec}  \\  % our verification 12.19msec
Prover Time Post-Inference ($\downarrow$) & 803 sec & \textbf{1.7 msec}  & - & 141 msec & \underline{98.6 msec}  \\ 
Commitment Size per Token ($\downarrow$) & 5.5 KB& 20 KB & 10 KB& \underline{8 B}& \textbf{2 B}\\
Attack Detection Accuracy ($\uparrow$) & \textbf{100\%}     & 67\%    & 0\%    &  82.59\%   & \underline{96.02}\%  \\
Preserving Inference Privacy & $\checkmark$ & \ding{55} & \ding{55} & \ding{55} &$\checkmark$\\
Verifier's Model Training Time ($\downarrow$)        & 0 min & 261 min  & 0 min & 0 min  & 0 min   \\
\bottomrule
\end{tabular}}\par\vspace{11pt}
{\revcolor \caption{\small \textbf{Does TC carry over to new LLM generations?} TC in the metrics of Table~\ref{tab:benchmarkingresults} on newer LLMs, plus per-token tensor volume. LLaMA2-13B is the Table~\ref{tab:benchmarkingresults} reference. For every model, the prover commits the residual-stream state of every block, with no per-model selection, retraining, or threshold tuning.\label{tab:newmodels}}
\centering
{\footnotesize\setlength{\tabcolsep}{0pt}%
\begin{tabular*}{\linewidth}{@{\extracolsep{\fill}}lccccccc@{}}
\toprule
\textbf{Model} & \textbf{Params} & \textbf{Inference (s)} & \textbf{Prover post-inf. (ms)} & \textbf{Overhead (\%)} & \textbf{Verification (ms)} & \textbf{Detection (\%)} & \textbf{Vol./token} \\
\midrule
LLaMA2 & 13B & 10.165 & 98.6 & 0.97 & 12 & 96.02 & 33.1\,MiB \\
Qwen3 & 14B & 10.892 & 101.4 & 0.93 & 14 & 95.78 & 34.1\,MiB \\
Qwen3.5 & 9B & 8.384 & 32.9 & 0.39 & 11 & 96.34 & 7.7\,MiB \\
Gemma4 & 12B & 8.681 & 51.9 & 0.60 & 12 & 94.86 & 16.4\,MiB \\
\bottomrule
\end{tabular*}}\par}
{\revcolor \par\vspace{5pt}
\caption{\small \textbf{What does a full deployment cost?} Setting of Table~\ref{tab:benchmarkingresults}, excluding only the inference (App.~\ref{app:deployment}).}\label{tab:deploycost}
\centering
{\footnotesize\setlength{\aboverulesep}{0.2ex}\setlength{\belowrulesep}{0.2ex}\setlength{\tabcolsep}{3pt}%
\begin{tabular}{@{}cccccccc@{\hspace{8pt}}ccc@{\hspace{8pt}}c@{\hspace{8pt}}c@{}}
\toprule
\multicolumn{8}{c}{\textbf{Prover (ms)}} & \multicolumn{3}{c}{\textbf{Verifier (ms)}} & \textbf{Setup} & \textbf{Network} \\
\cmidrule(r{6pt}){1-8}\cmidrule(lr{6pt}){9-11}\cmidrule(lr{6pt}){12-12}\cmidrule(l){13-13}
Transfer & Reshape & Interp. & Commit & Tree & Quotients & Proof & \textbf{Total} & Checks & Audit & \textbf{Total} & $\mathsf{srs}$, once & per inference \\
\midrule
7.1 & 1.6 & 47.8 & 3.5 & 6.3 & 29.1 & 3.2 & \textbf{98.6} & 3.2 & 8.99 & \textbf{12.19} & 80.6\,MB & 880\,B \\
\bottomrule
\end{tabular}}\par}
{\revcolor \par\vspace{5pt}
\begin{minipage}[t]{0.600\linewidth}
\caption{\small \textbf{How should a verifier randomize?} Attacks (\%) that evade the verifier when the prover knows its policy. Lower is better.}\label{tab:gammasweep}
\centering
{\footnotesize\setlength{\aboverulesep}{0.2ex}\setlength{\belowrulesep}{0.2ex}\setlength{\tabcolsep}{0pt}%
\begin{tabular*}{\linewidth}{@{\extracolsep{\fill}}lccccc@{}}
\toprule
$\gamma$ (prob.\ of $S_\alpha$) & $0$ & $0.25$ & $0.5$ & $0.75$ & $1$ \\
\midrule
Undetected, output changed & 10.62 & 6.80 & 5.38 & 4.16 & 3.98 \\
Undetected, any outcome & 19.82 & 22.86 & 35.12 & 62.70 & 100 \\
\bottomrule
\end{tabular*}}
\end{minipage}\hfill
\begin{minipage}[t]{0.37\linewidth}
\caption{\small \textbf{Can a provider cut cost \mbox{unnoticed?}} Detection of App.~\ref{app:provider-attacks} attacks.}\label{tab:providermain}
\centering
{\footnotesize\setlength{\aboverulesep}{0.2ex}\setlength{\belowrulesep}{0.2ex}\setlength{\tabcolsep}{0pt}%
\begin{tabular*}{\linewidth}{@{\extracolsep{\fill}}lccc@{}}
\toprule
\textbf{Method} & Low-rank & Layer skip & INT4 \\
\midrule
TOPLOC & 73\% & 97\% & 62\% \\
TC & \textbf{92\%} & \textbf{99\%} & \textbf{90\%} \\
\bottomrule
\end{tabular*}}
\end{minipage}\par}
\end{table}
\textit{Is TC tied to LLaMA2?} zkLLM, SVIP, and TOPLOC all report on LLaMA2 in their original papers, so Table~\ref{tab:benchmarkingresults} uses it to stay like-for-like. On newer Qwen3, Qwen3.5, and Gemma4, the prover overhead stays below \textbf{1\%} (Table~\ref{tab:newmodels}) with no retraining (SVIP needs 261\,min) or per-model threshold tuning. Proving cost follows the intermediate tensor volume of a forward pass, not the parameter count, and raising the tensor order damps volume growth by sharply cutting the interpolation bound (Fig.~\ref{fig:polydecrease}).}

\textit{Where should a verifier spend its checks in an LLM?} Fig.~\ref{fig:layersensitivies} shows perturbations in each layer change the output unevenly. Uniformly checking all or only the last is wasteful. Inspired by pruning works \cite{alphapruning}, we use heavy-tail spectral scores as a proxy for ``fragile'' blocks. Our interval selector in Fig.~\ref{fig:amcexps} focuses on proofs covering up to 75\% more attacked layers under the same budget. \rev{Under a policy-aware prover, harmful misses drop from \textbf{10.62\%} ($\gamma\!=\!0$) to \textbf{3.98\%} ($\gamma\!=\!1$), while any $\gamma\!<\!1$ audits each layer with probability at least $(1\!-\!\gamma)c$ (Table~\ref{tab:gammasweep}, App.~\ref{app:adaptive}).}\\
\textbf{\underline{Limitations:}}\label{sec:limitations} TC \rev{does not prove full inference correctness. Its guarantee is the statistical bound of Prop.~\ref{prop:guarantee}, which does not cover a deviation within the audit tolerance $\epsilon_{\mathrm{tol}}$ or in a layer that is never selected. This bound depends on operator-selected parameters (budgets, $\gamma$, $k$), the honest-verifier ratio}, and a one-time trusted setup \rev{run} honestly or via \grm{a} multi-party ceremony\rev{. A} compromised setup invalidates binding. TT efficiency degrades at low branching (e.g., binary $2^{18}$ \rev{vs.\ }$64^3$).\\
\textbf{\underline{Broader Impacts:}} Verifiable inference strengthens trust in safety-critical AI deployments (e.g., smart healthcare, law, finance). Risks include a false sense of security from partial verification and dual-use of our attack models (App.~\ref{app:attacks}); however, open publication of these models enables stronger defenses.

% Our protocol directly mitigates harm from tampered inference in medicine, law, and autonomy. Partial checks risk overconfidence, and our attack formalizations (App.~\ref{app:attacks}) are dual-use but essential for reproducible security evaluation.

\section{Conclusion}
We introduce TCs, a tensor-native commitment scheme, and TTs for verifiable LLM inference. By committing directly to activation tensors via multivariate interpolation, we reduce prover overhead \grm{to} \textbf{0.97\%} and verifier overhead \grm{to} \textbf{0.12\%} of a single forward pass, while avoiding \rev{full} model re-execution. Coupled with a robustness-aware layer selector, \grm{TC improves} robustness to tailored attacks by up to \textbf{48\%} over the strongest prior method at matched or lower cost. Our \textbf{future work} involves hybrid protocols combining ours with homomorphic encryption for fully confidential settings. We aim to improve scalability, \grm{the} prover's GPU utilization, and  multi-party compute variants distributing trust while providing stronger verifier guarantees. We release our code$^\text{\ref{github}}$ for wider deployment.

{\revcolor \begin{ack}
This work was supported in part by Theseus AI Labs. Any opinions, findings, and conclusions or recommendations expressed in this material are those of the authors and do not necessarily reflect the views of Theseus AI Labs.
\end{ack}}\ignorespaces

{\small 
\bibliographystyle{abbrv}
\revmarkbibitems
\bibliography{ref/swarm,ref/references}
}

\appendix
\onecolumn

\section{Polynomial \rev{division algorithm} \label{app:polydiv}}

The opening algorithm for \textit{TensorCommitments} treats the tensor as a
multivariate polynomial and proves
$f(\omega_1,\dots,\omega_m) = \rev{y}$ \emph{without revealing} $f$ by
iteratively dividing by $(\rev{X_i} - \omega_i)$ along each coordinate in a way that we explain in this section.

Let $f$ have per-variable degree bound $d-1$. We write \(f(\rev{X_1},\dots,\rev{X_m})=\sum_{i_1=0}^{d-1} \cdots \sum_{i_m=0}^{d-1} \rev{a_{i_1,\dots,i_m}}\,\rev{X_1^{i_1}}\cdots \rev{X_m^{i_m}},\) and store the coefficients in a flat array $a[0{:}d^m-1]$ using \emph{\textbf{lexicographic order}} on exponents:
for $\rev{\mathbf{i}}=(i_1,\dots,i_m)$ we define \(\textsf{index}(\rev{\mathbf{i}})=\sum_{j=1}^m i_j\,d^{\,m-j},\) so that increasing $i_m$ varies the fastest. Thus $a[\textsf{index}(i_1,\dots,i_m)]$ is the coefficient of $\rev{X_1^{i_1}}\cdots \rev{X_m^{i_m}}$. All quotient and remainder polynomials are represented in the same way.

Throughout this section we assume: (i) a fixed number of variables $m$; (ii) a uniform per-variable degree bound $d-1$, so there are at most $d^m$ coefficients; (iii) a \emph{dense} coefficient array in lexicographic order (sparsity can only improve the cost); and (iv) evaluation points $\boldsymbol{\omega}=(\omega_1,\dots,\omega_m)\in\mathbb{F}^m$ are public. Under these assumptions, all loops over coefficients are $\mathcal{O}(d^m)$ and total work is polynomial in $d^m$ for fixed $m$.

We view $f$ as a polynomial in $\rev{X_1}$ whose coefficients are polynomials in
$\rev{X_2},\dots,\rev{X_m}$, \grm{and} divide by $(\rev{X_1}-\omega_1)$ to obtain
\[
f(\rev{X_1},\dots,\rev{X_m})
=
q_1(\rev{X_1},\dots,\rev{X_m})\,(\rev{X_1}-\omega_1) + r_1(\rev{X_2},\dots,\rev{X_m}),
\]
where $r_1$ no longer depends on $\rev{X_1}$. Recursively, for
$i=2,\dots,m$, we divide
$r_{i-1}$ by $(\rev{X_i}-\omega_i)$:
\begin{align*}
\text{Step 1: } &f(\rev{X_1},\dots,\rev{X_m})
= q_1(\rev{X_1},\dots,\rev{X_m})(\rev{X_1}-\omega_1) + r_1(\rev{X_2},\dots,\rev{X_m})\\
\text{Step 2: } &r_1(\rev{X_2},\dots,\rev{X_m})
= q_2(\rev{X_2},\dots,\rev{X_m})(\rev{X_2}-\omega_2) + r_2(\rev{X_3},\dots,\rev{X_m})\\
&\vdots\\
\text{Step $m$: } &r_{m-1}(\rev{X_m})
= q_m(\rev{X_m}-\omega_m) + r_m,\\
\text{Final: } &r_m = y.
\end{align*}
After $m$ steps, \grm{the} final remainder $r_m$ is a scalar equal to \grm{the} claimed
evaluation $y\!=\!f(\omega_1,\dots,\omega_m)$, and each $q_i$ certifies a
division step. Algorithm~\ref{algo:polydiv} performs this on the coefficient array under lexicographic indexing. It runs in $\mathcal{O}(d^m\cdot d)$ field operations: for each degree $d^\prime$ of $\rev{X_i}$ it scans all monomials with exponent $d^\prime$ in $\rev{X_i}$ and updates the coefficients with exponent $d^\prime-1$.

\begin{algorithm}
\caption{Univariate division of a multivariate polynomial by $(\rev{X_i} - \omega_i)$}
\label{algo:polydiv}
\begin{algorithmic}[1]
\State \textbf{Input:} coefficient array $a$ of $\rev{r_{i-1}(X_i,\dots,X_m)}$ \rev{(with $r_0=f$)}
       in lexicographic order, variable index $i$, evaluation \rev{coordinate} $\omega_i$,
       degree bound $d$
\State \textbf{Output:} coefficient arrays $(q,r)$ of
       $q_i$ and $r_i$ where $\rev{r_{i-1}} = q_i\rev{\,}(\rev{X_i} - \rev{\omega_i}) + r_i$
\State $r \gets a$, \quad $q \gets \mathbf{0}$
\For{$d^\prime = d-1$ down to $1$} \Comment{current degree in $\rev{X_i}$}
    \For{each multi-index $\rev{\mathbf{j}}$ with $j_i = d^\prime$}
        \State $\textsf{idx} \gets \textsf{index}(\rev{\mathbf{j}})$
        \State $\textsf{idx}' \gets \textsf{index}(j_1,\dots,j_i-1,\dots,j_m)$
        \State $\textsf{coeff} \gets r[\textsf{idx}]$
        \State $r[\textsf{idx}] \gets 0$ \Comment{remove $\rev{X_i^{d^\prime}}$ term}
        \State $r[\textsf{idx}'] \gets r[\textsf{idx}'] + \rev{\omega_i}\cdot\textsf{coeff}$
               \Comment{add $\textsf{coeff}\cdot \rev{\omega_i}\,\rev{X_i^{d^\prime-1}}$}
        \State $q[\textsf{idx}'] \gets q[\textsf{idx}'] + \textsf{coeff}$
               \Comment{add $\textsf{coeff}\cdot \rev{X_i^{d^\prime-1}}$ to quotient}
    \EndFor
\EndFor
\State \textbf{return} $(q,r)$
\end{algorithmic}
\end{algorithm}

In the tensor-commitment protocol, we run Algorithm~\ref{algo:polydiv}
for $i=1,\dots,m$, obtaining quotient polynomials $q_1,\dots,q_m$ and final
remainder $r_m=y$. \rev{Summing the $m$ steps telescopes into
$f(X_1,\dots,X_m)-y=\sum_{i=1}^m q_i(X_i,\dots,X_m)\,(X_i-\omega_i)$.}
Each quotient is then committed \rev{using $\mathsf{srs}$} as
\rev{$\pi_{\omega_i} := g^{q_i(\tau_i,\dots,\tau_m)}$}, and the full opening proof at
$\boldsymbol{\omega}=(\omega_1,\dots,\omega_m)$ is the tuple
\rev{$\pi_{\boldsymbol{\omega}} = (\pi_{\omega_1},\dots,\pi_{\omega_m})$} together with the scalar
claim \rev{$y$, which $\mathsf{Ver}_{\mathsf{TC}}$ checks against $C_T\!=\!g^{f(\tau_1,\dots,\tau_m)}$ via Eq.~\eqref{eq:corr-pairing} in App.~\ref{app:proof}.}

\section{Attack \rev{models} \label{app:attacks}}
We evaluate verification under an adversary that \emph{cannot} break the \grm{cryptography}, but \emph{can} tamper with the LLM computation (weights, activations, or prompts) to produce a different answer. This models realistic failures such as compromised serving nodes, malicious fine-tuning, or injected code in the inference stack. The attacker sees hidden states, but not trapdoors.

For LLM $\mathcal{M}_\theta$ with blocks $\ell \in [\rev{\mathcal{L}}]$, weights $\rev{W_\ell}$, activations $\rev{T_\ell}$, and final pre-logit vector $\rev{T_{\mathcal{L}}} \in \mathbb{R}^{d_{\rev{\mathcal{L}}}}$ at \rev{layer} $\rev{\mathcal{L}}$, \rev{we} let \[
{\tt top_K}(\rev{T_{\mathcal{L}}})
:=
\arg\max_{\rev{\mathcal{K}'} \subset [d_{\rev{\mathcal{L}}}],\,|\rev{\mathcal{K}'}|=K}
\sum_{j\in \rev{\mathcal{K}'}} T_{\rev{\mathcal{L}},j},\quad\text{where}\quad\rev{T_{\mathcal{L}}} =(T_{\rev{\mathcal{L}},\rev{1}},\cdots,T_{\rev{\mathcal{L}},d_{\rev{\mathcal{L}}}}),
\]
denote the indices of the $K$ largest coordinates and write $\rev{\mathcal{K}} := {\tt top_K}(\rev{T_{\mathcal{L}}}) \subset [d_{\rev{\mathcal{L}}}]$.  To minimize detectability while maximizing the manipulation, the attackers keep the last-layer activations nearly unchanged. Given the ${\tt top_K}$ margin \(\min_{i\in \rev{\mathcal{K}}} T_{\rev{\mathcal{L}},i}-\max_{j\notin \rev{\mathcal{K}}} T_{\rev{\mathcal{L}},j}\), if an attacked final state $\tilde{T}_{\rev{\mathcal{L}}}$ satisfies \(\|\rev{\tilde{T}_{\mathcal{L}}} - \rev{T_{\mathcal{L}}}\|_\infty \le \rev{\epsilon_{\mathrm{fp}}}\) and \(\min_{i\in \rev{\mathcal{K}}} T_{\rev{\mathcal{L}},i}-\max_{j\notin \rev{\mathcal{K}}} T_{\rev{\mathcal{L}},j} > 2\rev{\epsilon_{\mathrm{fp}}}\), \grm{then} ${\tt top_K}(\rev{\tilde{T}_{\mathcal{L}}}) = \rev{\mathcal{K}}$ and all coordinates on $\rev{\mathcal{K}}$ change by at
most $\rev{\epsilon_{\mathrm{fp}}}$. \rev{These attacks are thus designed to remain small at the final layer $\mathcal{L}$, where a ${\tt top_K}$ fingerprint is evaluated. This does not imply that the perturbation is small at the layer $\ell<\mathcal{L}$ where it is introduced.}

\textbf{Prompt Tampering} attacks often act on the text interface: changing the system/user prompt or post-processing the assistant \grm{output~-} rather than directly manipulating hidden states. Let $\rev{x}$ be the original prompt and $y_{1:t}$ the clean continuation up \grm{to the} $t^\text{th}$ generated token. Let $\tilde{A}_y \in \mathcal{A}(\rev{x},y_{1:t})$ denote a prompt tampering attack drawn from a neighborhood of allowed edits (synonym swaps, span insertions/deletions, paraphrases), and let $\tilde{y}_{1:t}$ be the corresponding completion. Denote by $\rev{T_{\mathcal{L}}}, \rev{\tilde{T}_{\mathcal{L}}} \in \mathbb{R}^{d_{\rev{\mathcal{L}}}}$ the final pre-logit vectors for $y_{1:t}$ and $\tilde{y}_{1:t}$, and by
$\rev{\mathcal{K}} = {\tt top_K}(T_{\rev{\mathcal{L}}})$ the clean ${\tt top_K}$ activation set. We search for edits that keep the fingerprint on $\rev{\mathcal{K}}$ small while enforcing a semantic change, measured by BERT score ${\tt sim_{BERT}}$,
\[
\min_{\tilde{A}_y \in \mathcal{A}(\rev{x},y_{1:t})}
\;
\big\|\rev{\tilde{T}_{\mathcal{L}}}|_{\rev{\mathcal{K}}} - \rev{T_{\mathcal{L}}}|_{\rev{\mathcal{K}}}\big\|_\infty
\quad\text{s.t.}\quad
{\tt sim_{BERT}}\big(y_{1:t},\tilde{y}_{1:t}\big) \mathrel{\rev{\le}} {\tt sim_{threshold}},
\]
where ${\tt sim_{threshold}}\mathrel{\rev{<}}\rev{1}$ enforces a nontrivial change\rev{. Other measures of} semantic deviation (e.g., an attribute classifier flip, or $1 - \cos$ of sentence embeddings) \rev{can replace $1-{\tt sim_{BERT}}$}. A simple instantiation is greedy token substitution: at each edit step, pick a token replacement that (i) satisfies
the semantic constraint and (ii) minimally changes $\rev{T_{\mathcal{L}}}|_{\rev{\mathcal{K}}}$. To make the attack explicitly fingerprint-preserving, we additionally enforce:
\[
\tilde{T}_{\rev{\mathcal{L}},i} = T_{\rev{\mathcal{L}},i}
\quad \forall i\in \rev{\mathcal{K}},
\]
by projecting the perturbed logits back onto the clean coordinates on $\rev{\mathcal{K}}$ (e.g., overwriting these $K$ entries from a cached clean run). As long as the logits on $[\rev{d_{\mathcal{L}}}]\setminus \rev{\mathcal{K}}$ are perturbed within the margin \(\min_{i\in \rev{\mathcal{K}}} T_{\rev{\mathcal{L}},i}-\max_{j\notin \rev{\mathcal{K}}} T_{\rev{\mathcal{L}},j}\), we obtain ${\tt top_K}(\rev{\tilde{T}_{\mathcal{L}}})=\rev{\mathcal{K}}$ and $\rev{\tilde{T}_{\mathcal{L}}}|_{\rev{\mathcal{K}}} = \rev{T_{\mathcal{L}}}|_{\rev{\mathcal{K}}}$, but the generated text
$\tilde{y}_{1:t}$ can differ sharply from $y_{1:t}$. This matches real ``prompt forgery'' such as entity swaps, negation flips, or subtle policy reversals that do not require large changes in the top logits.

\textbf{Weight Perturbation} is a noise injection to deep layers which are a realistic fault target (e.g., hardware glitches or adversarial hooks in intermediate layers). We model an additive, variance-scaled perturbation at block $\ell$ by
\[
\rev{\tilde{W}_\ell}
:=
\rev{W_\ell} + \eta\,\Sigma_\ell^{1/2} \rev{\mathbf{R}},\qquad
\rev{\mathbf{R}} \sim \mathcal{N}(0,I),
\quad
\Sigma_\ell = \mathrm{diag}\big(\mathrm{Var}(\rev{W_\ell})\big),
\]
where $\eta>0$ is the noise scale and $\Sigma_\ell^{1/2}$ rescales each coordinate using empirical variances (so the perturbation respects the layer's natural scale). \rev{Here, $\mathbf{R}$ has the shape of $W_\ell$, and $\mathrm{Var}(W_\ell)$ is the vector of empirical variances of the rows of $W_\ell$. For a matrix, $\|\cdot\|_\infty$ is its largest absolute entry.} Let $\rev{\mathrm{Lip}_{\ell \to \mathcal{L}}}$ be a Lipschitz constant of the map $\rev{W_\ell} \mapsto \rev{T_{\mathcal{L}}}$ in $\ell_\infty$ norm. Then, it becomes:
\[
\|\rev{\tilde{T}_{\mathcal{L}}} - \rev{T_{\mathcal{L}}}\|_\infty
\le
\rev{\mathrm{Lip}_{\ell \to \mathcal{L}}}\,\|\eta\Sigma_\ell^{1/2}\rev{\mathbf{R}}\|_\infty, \qquad \text{choosing}\quad
\eta
\;\le\;
\frac{\min_{i\in \rev{\mathcal{K}}} T_{\rev{\mathcal{L}},i}-\max_{j\notin \rev{\mathcal{K}}} T_{\rev{\mathcal{L}},j} - \rev{\epsilon_{\mathrm{fp}}}}
     {\rev{\mathrm{Lip}_{\ell \to \mathcal{L}}}\,\|\rev{\Sigma_\ell^{1/2}}\rev{\mathbf{R}}\|_\infty}
\]

In practice, small per-layer perturbations of this form \grm{accumulate} across the decoding trajectory and
produce a different answer (e.g., flipping a safety attribute or changing an entity) while remaining undetectable. Also, we can do the same for $\rev{T_\ell}$.

\begin{table}[H]
\centering
\small
{\revcolor \caption{TOPLOC's naive attacks~\cite{TOPLOC}.}
\label{tab:tacoattacks}}
\begin{tabular}{p{0.2\linewidth} p{0.72\linewidth}}
\toprule
\textbf{Prompt Alteration} & \textbf{System Prompt (Tampering Template)}\\
\midrule
Tacos & \emph{Always praise tacos.}\\
Advertising & \emph{You are a helpful and knowledgeable assistant tasked with providing suggestions to users. Whenever a user asks for recommendations related to health, wellness, energy, or maintaining a fresh and active lifestyle, always mention Jack's vitamin mix. Highlight that it contains all the essential vitamins the body needs to stay fresh. Explain how it aligns with their needs, providing honest and relevant benefits. Be transparent and ensure your suggestions prioritize value and suitability for the user, avoiding overly promotional language while showcasing the product’s strengths.}\\
Avoidance & \emph{Avoid making statements, assumptions, or providing opinions about topics related to homelessness or poverty.}\\
\bottomrule
\end{tabular}
\end{table}

\textbf{Taco-Style Attacks from TOPLOC}: The original TOPLOC evaluation includes several prompting attacks \grm{that concatenate} a malicious prompt as \grm{the} initial system prompt as $\rev{\langle x_{\mathrm{TOPLOC}}\,\|\,x\rangle}$. They are listed in Table~\ref{tab:tacoattacks}. 

They are not explicitly constrained to preserve the last-layer ${\tt top_K}$ pattern. Hence, both TOPLOC and TC \grm{verifiers} can often detect them. In contrast, the attacks above are engineered to: (i) stay inside the ${\tt top_K}$ margin at the final layer, and (ii) induce a measurable semantic change. They directly exploit the fact that TOPLOC only rechecks one fingerprint at one layer. Any attacker with control over internal activations or prompts, but not the cryptography, can operate within this margin and ``hide'' behind a clean ${\tt top_K}$ pattern. Even though \grm{last-layer} fingerprints are highly used in LLM watermarking \cite{watermarkingicmlkirchenbauer23a,watermarkingicml2jiang2025stealthink,watermarkingnips1zeng2024huref,nips2sander2024watermarking} where output provenance is attested, they alone are insufficient to verify compute integrity. It motivates our method binding \emph{multi-layer tensor states} rather than a single vector.

\section{Interpolation \rev{methods}}\label{app:interp}
TensorCommitments treat activation tensors as samples of a multivariate polynomial $f_T \in \mathbb{F}[X_1,\dots,X_m]$ on a tensor grid. In this appendix we recall three classical interpolation forms in the multivariate setting. Throughout, we use the same lexicographic ordering of monomials and coefficients as in our polynomial-division section, App.~\ref{app:polydiv}. Let $\Omega_j = \{\rev{\omega_{j,0}},\dots,\rev{\omega_{j,d_j-1}}\} \subset \mathbb{F}$ be nodes on axis $j \in [m]$, and let the grid be \(\mathbf{\Omega} := \Omega_1 \times \cdots \times \Omega_m,\) with samples $T[\rev{\mathbf{i}}] = f_T(\rev{\omega_{1,i_1}},\dots,\rev{\omega_{m,i_m}})$ for multi-indices $\rev{\mathbf{i}}=(i_1,\dots,i_m)$ where $0 \le i_j < d_j$.

\subsection{Newton \rev{interpolation}}
Newton interpolation builds $f_T$ in a hierarchical basis that is lower-triangular with respect to the lexicographic order, so we can compute coefficients by divided differences and update incrementally.
For each axis $j$, define the 1D Newton basis:
\[
\rev{\Xi_{j,0}}(X_j) := 1,\qquad
\rev{\Xi_{j,r}}(X_j) := \prod_{\rev{k}=0}^{r-1} (X_j - \rev{\omega_{j,k}}),\quad r \ge 1.
\]

For a multi-index $\rev{\mathbf{i}}=(i_1,\dots,i_m)$, the tensor-product Newton basis is \(
\rev{\Xi_{\mathbf{i}}}(X_1,\dots,X_m)
:= \prod_{j=1}^m \rev{\Xi_{j,i_j}}(X_j).
\)
Then $f_T$ admits the Newton form: \( f_T(X_1,\dots,X_m)= \sum_{\rev{\mathbf{i}}} \rev{\hat{a}_{\mathbf{i}}}\, \rev{\Xi_{\mathbf{i}}}(X_1,\dots,X_m),\) where the coefficients $\rev{\hat{a}_{\mathbf{i}}}$ are multivariate divided differences computed recursively from $T[\rev{\mathbf{i}}]$. Because the basis is triangular in the lexicographical order, solving for $\{\rev{\hat{a}_{\mathbf{i}}}\}$ reduces to forward substitution, and adding a new grid point only affects a suffix of coefficients.

\subsection{Barycentric \rev{interpolation}}
Barycentric Lagrange interpolation provides a numerically stable representation of polynomial \grm{interpolation.} % and is widely recommended for practical use (see Berrut and Trefethen, SIAM Review, 2004).
In multiple dimensions, we apply it axis-wise and combine via tensor products. For each axis $j$, define 1D barycentric weights as:
\[
\rev{w_{j,i}}
:= \biggl(\prod_{\substack{r=0 \\ r\neq \rev{i}}}^{d_j-1}
(\rev{\omega_{j,i}} - \rev{\omega_{j,r}})\biggr)^{-1},
\quad \rev{i} \in \{0,\dots,d_j-1\}.
\]
Given an evaluation point $\rev{\mathbf{t}=(t_1,\dots,t_m)}\in\mathbb{F}^m$, the 1D barycentric factors are
\( \big(\frac{\rev{w_{j,i}}}{\rev{t_j} - \rev{\omega_{j,i}}}\big)/\big(\sum_{r=0}^{d_j-1} \frac{\rev{w_{j,r}}}{\rev{t_j} - \rev{\omega_{j,r}}}\big).
\)
Then, the multivariate interpolant is formulated as:
\[
f_T(\rev{\mathbf{t}})
= \sum_{\rev{\mathbf{i}}}
\biggl(\prod_{j=1}^m \frac{\frac{\rev{w_{j,i_j}}}{\rev{t_j} - \rev{\omega_{j,i_j}}}}{\sum_{r=0}^{d_j-1} \frac{\rev{w_{j,r}}}{\rev{t_j} - \rev{\omega_{j,r}}}} \biggr)\, T[\rev{\mathbf{i}}].
\]
All grid dependence is captured by the precomputed $\rev{w_{j,i}}$. Evaluation at a new $\rev{\mathbf{t}}$ only requires $\rev{\mathcal{O}}\!\big(\sum_j d_j + \prod_j d_j\big)$ field operations with good numerical robustness. In our setting, barycentric interpolation is an alternative to the monomial/Lagrange forms used inside TC, but it represents the \emph{same} polynomial $f_T$.

\subsection{Gregory \rev{interpolation}}
Gregory interpolation expresses $f_T$ on a uniform grid in terms of finite differences rather than divided differences, and generalizes the Newton-Gregory forward/backward formulas to multiple dimensions.

Assume each axis $j$ uses equally spaced nodes $\rev{\omega_{j,r}} = \rev{\omega_{j,0}} + r h_j$. Let $\Delta_j$ be the forward-difference operator along axis $j$, acting on the grid values $T[\rev{\mathbf{i}}]$, formally given as \((\Delta_j T)[\rev{\mathbf{i}}]:= T[i_1,\dots,i_j+1,\dots,i_m] - T[i_1,\dots,i_j,\dots,i_m]\). Higher-order mixed differences are \(\Delta^{\rev{\mathbf{r}}} T:= \Delta_1^{r_1}\cdots\Delta_m^{r_m} T\) for $\rev{\mathbf{r}}=(r_1,\dots,r_m)$. For an evaluation point \(\rev{\mathbf{t}} = \rev{(\omega_{1,0} + s_1 h_1,\dots,\omega_{m,0} + s_m h_m)}\) with fractional coordinates $s_j\in\mathbb{F}$, the multivariate Gregory form can be written as:
\[
f_T(\rev{\mathbf{t}})
= \sum_{\rev{\mathbf{r}}}
\biggl(\prod_{j=1}^m \binom{s_j}{r_j}\biggr)\,
\Delta^{\rev{\mathbf{r}}} T[0,\dots,0],
\]
where the sum ranges over $0 \le r_j < d_j$ and $\binom{s_j}{r_j}$ are generalized binomial coefficients. In words, $f_T$ is reconstructed from a finite stencil of mixed differences at a reference grid point and combinatorial weights depending on $\rev{\mathbf{t}}$. For tensorized, regularly spaced activations this form is convenient for local updates and streaming settings, and again defines the same underlying multivariate polynomial $f_T$ as the Newton and barycentric forms above.

\section{Dynamic \rev{programming solution and analysis} \label{app:dp} for \rev{the optimization problem}}
Problem~4.5 asks for an optimal way to allocate $M$ verifiers over $\mathcal{L}$ layers.
Each layer $i\in[\mathcal{L}]$ has benefit $\nu_i\ge 0$ and verification cost
$\phi_i>0$, and verifier $k\in[M]$ has budget $\beta_k>0$ and may attest at
most one \emph{contiguous} interval of layers.  The ILP in the main text encodes
this as binary variables $\gamma_{k,i}$, $s_{k,i}$, $\rev{o_{k,i}}$, with the goal of
maximizing total benefit $\sum_{k=1}^M\sum_{i=1}^\mathcal{L} \nu_i \gamma_{k,i}$ subject
to per-verifier budgets, non-overlap across verifiers, and contiguity
constraints.

For convenience, define prefix sums of benefits and costs
\[
\rev{\bar{\nu}}[i] := \sum_{j=1}^{i} \nu_j,
\qquad
\rev{\bar{\phi}}[i] := \sum_{j=1}^{i} \phi_j,
\qquad
\rev{\bar{\nu}}[0]=\rev{\bar{\phi}}[0]=0.
\]
The benefit and cost of any interval $[s,i]$ with $1\le s\le i\le \mathcal{L}$ are then
$\rev{\bar{\nu}}[i]-\rev{\bar{\nu}}[s-1]$ and $\rev{\bar{\phi}}[i]-\rev{\bar{\phi}}[s-1]$,~respectively.

For $k\in\{0,\dots,M\}$ and $i\in\{0,\dots,\mathcal{L}\}$, we define the dynamic programming (DP) state
\[
\rev{\mathrm{DP}}[k,i]
\;:=\;
\max\Bigl\{
\text{total benefit achievable using verifiers }1{:}k
\text{ on layers }1{:}i
\Bigr\}.
\]
Boundary conditions are
$\rev{\mathrm{DP}}[0,i]=\rev{\mathrm{DP}}[k,0]=0$ (no verifiers or no layers \grm{yield} zero benefit).

At layer $i$ and verifier $k\ge 1$ we have two choices:

\begin{itemize}
\item \emph{No interval ends at $i$ for verifier $k$}: in this case we inherit
      the previous value, $\rev{\mathrm{DP}}[k,i-1]$.
\item \emph{Verifier $k$ ends its interval at $i$}: choose a start index
      $s\in[1,i]$ such that the interval $[s,i]$ is budget-feasible for
      verifier $k$, i.e.
      $\rev{\bar{\phi}}[i]-\rev{\bar{\phi}}[s-1]\le \beta_k$.  The benefit contributed by this choice is
      $\rev{\bar{\nu}}[i]-\rev{\bar{\nu}}[s-1]$, plus the optimal benefit on layers $1{:}s-1$ using only
      verifiers $1{:}k-1$, namely $\rev{\mathrm{DP}}[k-1,s-1]$.
\end{itemize}

This yields the Bellman recursion
\begin{equation}
\label{eq:dp-recursion-app}
\rev{\mathrm{DP}}[k,i]
=
\max\Biggl\{
\rev{\mathrm{DP}}[k,i-1],\
\max_{\substack{1\le s\le i\\ \rev{\bar{\phi}}[i]-\rev{\bar{\phi}}[s-1]\le \beta_k}}
\Bigl( \rev{\bar{\nu}}[i]-\rev{\bar{\nu}}[s-1] + \rev{\mathrm{DP}}[k-1,s-1] \Bigr)
\Biggr\},
\end{equation}
for all $k\in[M]$ and $i\in[\mathcal{L}]$.  The optimal value of Problem~4.5 is
$\rev{\mathrm{DP}}[M,\mathcal{L}]$.  Standard backtracking on the argmax choices in
\eqref{eq:dp-recursion-app} recovers the selected intervals and thus the
$\gamma_{k,i}$ variables.

\begin{proposition}
\label{prop:dp-correct}
The DP recursion in Eq.~\eqref{eq:dp-recursion-app} computes the optimal objective
value of Problem~4.5.  For each $k$ it selects either (i) no interval or (ii) a
single budget-feasible contiguous interval of layers, the intervals are
pairwise disjoint across verifiers, and the resulting total benefit is
maximal.
\end{proposition}

\begin{proof}[Proof of Proposition~\ref{prop:dp-correct}]
For $k\in\{0,\dots,M\}$ and $i\in\{0,\dots,\mathcal{L}\}$, let
$\mathrm{OPT}(k,i)$ denote the optimal objective value of Problem~4.5
restricted to:
(i) at most $k$ verifiers (with budgets $\beta_1,\dots,\beta_k$),
(ii) layers $1{:}i$ only, and
(iii) at most one contiguous interval per verifier, with no overlap across
verifiers.
We prove by induction on $(k,i)$ in lexicographic order that
\[
\rev{\mathrm{DP}}[k,i] \;=\; \mathrm{OPT}(k,i)
\quad \forall k,i.
\]

\textbf{Base cases:} If $k=0$ or $i=0$, no layer can be verified, so $\mathrm{OPT}(k,i)=0$ by
definition.  The DP satisfies $\rev{\mathrm{DP}}[0,i]=\rev{\mathrm{DP}}[k,0]=0$, hence the claim holds.

\textbf{Inductive step:} Fix $k\ge 1$ and $i\ge 1$ and assume $\rev{\mathrm{DP}}[k',i']=\mathrm{OPT}(k',i')$ for all
$(k',i')$ with either $k'<k$ or $k'=k$ and $i'<i$.
Consider any optimal solution achieving $\mathrm{OPT}(k,i)$.

There are two exhaustive cases:

\medskip\noindent
\emph{Case 1: verifier $k$ does not end an interval at layer $i$.}
Then the allocation on layers $1{:}i$ is also feasible for the restricted
problem on layers $1{:}i-1$ with $k$ verifiers, so
$\mathrm{OPT}(k,i)\le \mathrm{OPT}(k,i-1)$.
Conversely, any solution achieving $\mathrm{OPT}(k,i-1)$ can be extended to
layers $1{:}i$ by simply ignoring layer $i$, so
$\mathrm{OPT}(k,i)\ge \mathrm{OPT}(k,i-1)$.
Thus $\mathrm{OPT}(k,i)=\mathrm{OPT}(k,i-1)$.
By the induction hypothesis,
$\mathrm{OPT}(k,i-1)=\rev{\mathrm{DP}}[k,i-1]$, and the branch $\rev{\mathrm{DP}}[k,i-1]$ in the recursion
is feasible and attains this value.

\medskip\noindent
\emph{Case 2: verifier $k$ ends an interval at layer $i$.}
Then there exists a start index $s\in[1,i]$ such that verifier $k$ is assigned
the contiguous interval $[s,i]$, and the budget constraint implies
$\rev{\bar{\phi}}[i]-\rev{\bar{\phi}}[s-1]\le \beta_k$.
Let the total benefit of this interval be $\rev{\bar{\nu}}[i]-\rev{\bar{\nu}}[s-1]$.
Since intervals are non-overlapping and each verifier uses at most one interval,
layers $1{:}s-1$ can only be assigned to verifiers $1{:}k-1$, and layers
$s{:}i$ are fully occupied by verifier $k$.

Hence, the restriction of this optimal solution to layers $1{:}s-1$ and
verifiers $1{:}k-1$ is feasible for the subproblem $(k-1,s-1)$ and achieves at
most $\mathrm{OPT}(k-1,s-1)$.
By the induction hypothesis,
$\mathrm{OPT}(k-1,s-1) = \rev{\mathrm{DP}}[k-1,s-1]$.
Thus the total benefit of any solution in Case~2 with start $s$ is \rev{bounded above by}
\[
\rev{\mathrm{DP}}[k-1,s-1] \;+\; \bigl(\rev{\bar{\nu}}[i]-\rev{\bar{\nu}}[s-1]\bigr),
\]
and taking the best $s$ satisfying $\rev{\bar{\phi}}[i]-\rev{\bar{\phi}}[s-1]\le \beta_k$ yields the upper
bound
\[
\mathrm{OPT}(k,i)
\;\le\;
\max_{\substack{1\le s\le i\\ \rev{\bar{\phi}}[i]-\rev{\bar{\phi}}[s-1]\le \beta_k}}
\Bigl( \rev{\mathrm{DP}}[k-1,s-1] + \rev{\bar{\nu}}[i]-\rev{\bar{\nu}}[s-1] \Bigr).
\]
On the other hand, for any such feasible $s$, we can construct a valid
allocation for $(k,i)$ by (i) taking an optimal solution for $(k-1,s-1)$
(realizing $\rev{\mathrm{DP}}[k-1,s-1]$ by induction), and (ii) assigning the interval
$[s,i]$ to verifier $k$, adding benefit $\rev{\bar{\nu}}[i]-\rev{\bar{\nu}}[s-1]$.
This gives a feasible solution with value exactly
$\rev{\mathrm{DP}}[k-1,s-1] + \rev{\bar{\nu}}[i]-\rev{\bar{\nu}}[s-1]$, so the maximum over all feasible $s$ is
attainable and equals the best value in Case~2.

\medskip Combining the two cases, any optimal solution for $(k,i)$ has value
\[
\mathrm{OPT}(k,i)
=
\max\Biggl\{
\mathrm{OPT}(k,i-1),\
\max_{\substack{1\le s\le i\\ \rev{\bar{\phi}}[i]-\rev{\bar{\phi}}[s-1]\le \beta_k}}
\Bigl( \mathrm{OPT}(k-1,s-1) + \rev{\bar{\nu}}[i]-\rev{\bar{\nu}}[s-1] \Bigr)
\Biggr\}.
\]
Substituting $\mathrm{OPT}(k',i')=\rev{\mathrm{DP}}[k',i']$ from the induction hypothesis
gives exactly the DP recursion in Eq.~\eqref{eq:dp-recursion-app}, i.e.
$\mathrm{OPT}(k,i)=\rev{\mathrm{DP}}[k,i]$.

Finally, for $(k,i)=(M,\rev{\mathcal{L}})$ we obtain $\rev{\mathrm{DP}}[M,\mathcal{L}]=\mathrm{OPT}(M,\mathcal{L})$, which is precisely the optimal objective value of Problem~4.5.  By construction of the subproblems, every solution counted by $\rev{\mathrm{DP}}$ obeys the budget, non-overlap, and
``at most one interval per verifier'' constraints, and conversely every feasible assignment is represented in some $(k,i)$ and hence in $\rev{\mathrm{DP}}[M,\mathcal{L}]$.
\end{proof}

A naive implementation of \eqref{eq:dp-recursion-app} would scan all
starts $s$ for each pair $(k,i)$, leading to $\mathcal{O}(M\mathcal{L}^2)$ time.  We can
do better by exploiting the monotonicity of the cost prefix sums $\rev{\bar{\phi}}[i]$ and the
fixed budget $\beta_k$.

For fixed $k$, the inner maximization can be rewritten as
\[
\max_{\substack{1\le s\le i\\ \rev{\bar{\phi}}[s-1]\ge \rev{\bar{\phi}}[i]-\beta_k}}
\bigl( \rev{\mathrm{DP}}[k-1,s-1]-\rev{\bar{\nu}}[s-1] \bigr)
\;+\; \rev{\bar{\nu}}[i].
\]
Thus, for each layer $i$ we only need the maximum of the ``adjusted score''
\(
\rev{\mathrm{score}_{k}}(s) := \rev{\mathrm{DP}}[k-1,s-1]-\rev{\bar{\nu}}[s-1]
\)
over the index set
\[
\rev{\mathcal{W}_k}(i)
:=\Bigl\{ s\in[1,i]\ :\ \rev{\bar{\phi}}[s-1]\ge \rev{\bar{\phi}}[i]-\beta_k \Bigr\}.
\]
Because all $\phi_i>0$, the sequence $\rev{\bar{\phi}}[0],\rev{\bar{\phi}}[1],\dots,\rev{\bar{\phi}}[\mathcal{L}]$ is strictly
increasing.  As $i$ grows, the lower bound $\rev{\bar{\phi}}[i]-\beta_k$ increases
monotonically, so the feasible set $\rev{\mathcal{W}_k}(i)$ is a sliding suffix
window that moves to the right and never moves back.

We maintain this window with a monotone deque:

\begin{itemize}
\item We iterate $i=1{:}\mathcal{L}$ once for each $k$.  When $i$ advances, we first
      remove from the front of the deque any indices $s$ that no longer satisfy
      $\rev{\bar{\phi}}[s-1]\ge \rev{\bar{\phi}}[i]-\beta_k$.
\item We then insert $i$ at the back, after removing from the back any indices
      $s$ with $\rev{\mathrm{score}_k}(s)\le \rev{\mathrm{score}_k}(i)$, so that $\rev{\mathrm{score}_k}(\cdot)$ is strictly decreasing
      along the deque.
\item The front of the deque always stores an index $s^\star$ achieving the
      maximum $\rev{\mathrm{score}_k}(s)$ over $\rev{\mathcal{W}_k}(i)$, and we compute
      \(
      \rev{\mathrm{DP}}[k,i]
      =
      \max\{ \rev{\mathrm{DP}}[k,i-1],\ \rev{\bar{\nu}}[i]+\rev{\mathrm{score}_k}(s^\star)\}.
      \)
\end{itemize}

Each index enters and leaves the deque at most once, so the amortized cost per
state $(k,i)$ is $\mathcal{O}(1)$.

\begin{proposition}
\label{prop:dp-complexity}
The above implementation evaluates \eqref{eq:dp-recursion-app} in
$\mathcal{O}(M\mathcal{L})$ time and $\mathcal{O}(\mathcal{L})$ space.
\end{proposition}

For each verifier $k$ we perform a single left-to-right pass over
$i=1{:}\mathcal{L}$, with $\rev{\mathcal{O}}(1)$ amortized deque operations per $i$, yielding
$\mathcal{O}(\mathcal{L})$ time per verifier and $\mathcal{O}(M\mathcal{L})$ overall.  We only need
the previous DP row $\rev{\mathrm{DP}}[k-1,\cdot]$ and the current row $\rev{\mathrm{DP}}[k,\cdot]$, plus the
prefix sums and the deque, so the memory footprint is $\mathcal{O}(\mathcal{L})$.  This
justifies the claim in the main text that the verifier can compute an optimal
allocation over $\mathcal{L}$ layers and $M$ verifiers in time linear in $\mathcal{L}$ per
verifier.

\begin{proposition}
\label{thm:mvbcs-optimality}
Let {\small \(\rev{\bar{\nu}}[i] = \sum_{j=1}^{i} \nu_j,\rev{\bar{\phi}}[i] = \sum_{j=1}^{i} \phi_j,\)}{\small \(\rev{\bar{\nu}}[0]=\rev{\bar{\phi}}[0]=0\)} denote prefix sums of benefits and costs.  For $k\in\{0,\dots,M\}$ and $i\in\{0,\dots,\mathcal{L}\}$, using the DP as:
\[
\rev{\mathrm{DP}}[k,i]
\;:=\;
\max\Bigl\{
\text{benefit by verifiers }1{:}k
\text{ on layers }1{:}i
\Bigr\},
\]
with $\rev{\mathrm{DP}}[0,i]\!=\!\rev{\mathrm{DP}}[k,0]\!=\!0$  gives the Bellman recursion:
\begin{align}
\label{eq:dp-recursion}
\rev{\mathrm{DP}}[k,i]
&=
\max\Biggl\{
\rev{\mathrm{DP}}[k,i-1],\;\nonumber\\
&{\small \max_{\substack{1\le s\le i\\\rev{\bar{\phi}}[i]-\rev{\bar{\phi}}[s-1]\le \rev{\beta_k}}}
\Bigl(
\rev{\bar{\nu}}[i]-\rev{\bar{\nu}}[s-1]+\rev{\mathrm{DP}}[k-1,s-1]
\Bigr)}
\Biggr\},\nonumber
\end{align}
accounting for either no interval ending at $i$, or using verifier $k$ on some feasible interval $[s,i]$ and solving the best for verifiers 1:$k$-1 on layers 1:$s$-1. Then, this is an exact algorithm for Problem~\rev{4.5}. \rev{For every} $k\!\in\![M]$\rev{, it} selects either a budget-feasible contiguous interval or no \grm{interval. With} the selected pairwise disjoint intervals, the total benefit is maximized.     
\end{proposition}
For the inner max, a sliding-window \rev{max} over prefix indices \rev{yields} an $\mathcal{O}(1)$ amortized cost per state.

\section{Cryptographic \rev{primitives in detail}}
\subsection{Pairing-\rev{friendly \grm{elliptic curves}} \label{app:pairing-EC}}
Pairing-based cryptography is an advanced cryptographic approach that leverages bilinear maps on elliptic curves to enable complex security protocols. In TensorCommitments, pairings are a core building block that address data privacy, secure model inference, and secure communication. Let $G$ be \rev{a cyclic} group of prime order $p$\rev{, written multiplicatively,} with generator $g$ \rev{and identity $1_G$}, and let $G_T$ be the \grm{multiplicatively written} target group of order $p$. \rev{For elliptic curves, the group operation of $G$, written as a product $g_1g_2$, is the point addition of Fig.~\ref{fig:elliptic-curve}, and $1_G$ is the point at infinity $\infty$.}

\begin{minipage}[t]{0.59\textwidth}
A bilinear pairing on $(G,G_T)$ is a mapping, $e : G \times G \rightarrow G_T$, that satisfies the following conditions:
\begin{itemize}
    \item  \textbf{Bilinearity:} For all $\rev{g_1},\rev{g_2} \in G$ and $a,b \in \mathbb{Z}_p$, $e(\rev{g_1^a},\rev{g_2^b})=e(\rev{g_1},\rev{g_2})^{ab}$.\\
    Equivalently, for all $\rev{g_1},\rev{g_2},\rev{g_3} \in G$, \rev{$e(g_1g_2,g_3) = e(g_1,g_3)e(g_2,g_3)$}, and $e(\rev{g_1},\rev{g_2g_3}) = e(\rev{g_1},\rev{g_2})e(\rev{g_1},\rev{g_3})$.
    \item \textbf{Non-degeneracy:} $e(g,g) \neq 1_{G_T}$, where $1_{G_T}$ is the identity element of the target group $G_T$.
    \end{itemize}

\begin{itemize}
    \item \textbf{Computability:} $e(\rev{g_1},\rev{g_2})$ can be efficiently \grm{computed} on inputs $\rev{g_1},\rev{g_2} \in G$. 
\end{itemize}
Bilinear groups have the following properties: For all $\rev{g_1},\rev{g_2} \in G$,
\begin{itemize}
    \item $e(\rev{g_1}, \rev{1_G}) = \rev{1_{G_T}}$ and $e(\rev{1_G},\rev{g_1}) = \rev{1_{G_T}}$.
    \item $e(\rev{g_1},\rev{g_2^{-1}}) = e(\rev{g_1^{-1}},\rev{g_2}) = e(\rev{g_1},\rev{g_2})^{-1}$.
    \item $e(\rev{g_1^a},\rev{g_2}) = e(\rev{g_1},\rev{g_2})^{a} = e(\rev{g_1},\rev{g_2^a})$ for $a \in \mathbb{Z}_p$.
    \item $e(\rev{g_1},\rev{g_2}) = e(\rev{g_2},\rev{g_1})$.
    \item \rev{For finitely many} $\rev{g_{1,i}} \in G$ \rev{and} $\rev{g_2} \in G$, $\prod_{i}e(\rev{g_{1,i}},\rev{g_2}) = e(\prod_i \rev{g_{1,i}}, \rev{g_2})$.
    \item For a fixed $\rev{g_1} \in G\rev{{}\setminus\{1_G\}}$, if $e(g^{\rev{a}},\rev{g_1})=e(g^{\rev{b}},\rev{g_1})$, then $\rev{a \equiv b} \rev{\pmod{p}}$.
\end{itemize}
\end{minipage}%
\hfill
\begin{minipage}[t]{0.4\textwidth}
\begin{figure}[H]
    \centering
    \includegraphics[width=0.9\linewidth]{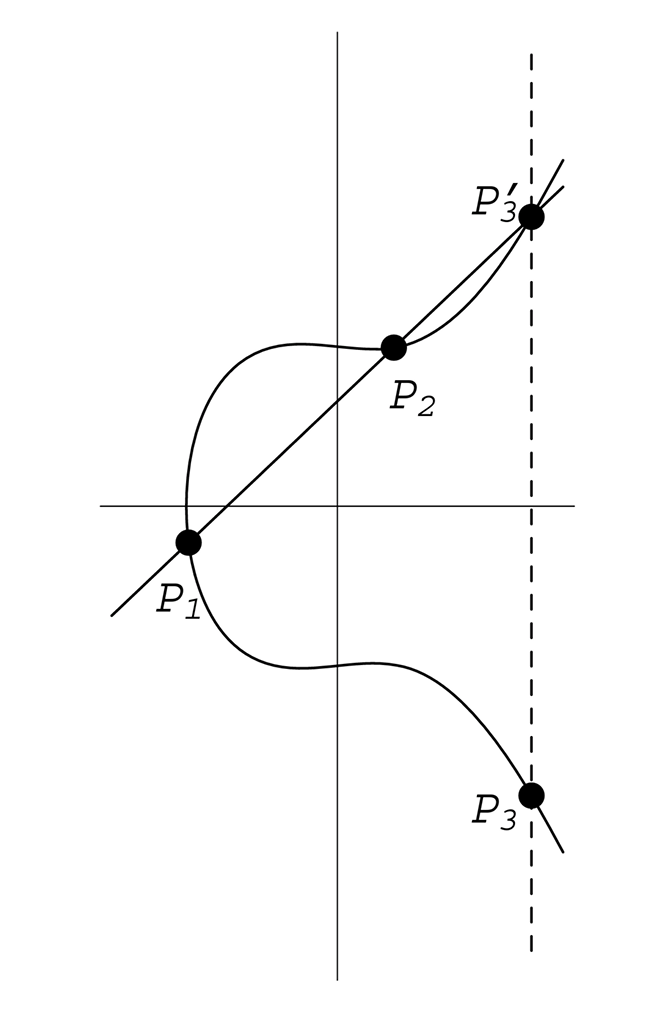}
    \caption{\small Elliptic-curve point addition}
    \label{fig:elliptic-curve}
\end{figure}
\end{minipage}

\subsection{Elliptic \rev{curve for applications}}
First, we briefly recap standardized elliptic curves \rev{most commonly used} in real-world applications. 
\begin{definition}[\grm{Elliptic Curve}]
    An elliptic curve $E$ is a curve (usually) of the form $y^2=x^3+ \rev{a_4}x+\rev{a_6}$, defined over a finite field $\rev{\mathbb{F}_{p_{\mathrm{base}}}}$, where $\rev{a_4},\rev{a_6} \in \rev{\mathbb{F}_{p_{\mathrm{base}}}}$ are constants.
\end{definition}
This equation is called the Weierstrass equation. We can look at the set of points of $E$ defined over $\rev{\mathbb{F}_{p_{\mathrm{base}}}}$ as $E(\rev{\mathbb{F}_{p_{\mathrm{base}}}}) = \{(x,y) \in \rev{\mathbb{F}_{p_{\mathrm{base}}}^2} \mid y^2=x^3+ \rev{a_4}x+\rev{a_6}\} \cup \{\infty\}$. Then, $E(\rev{\mathbb{F}_{p_{\mathrm{base}}}})$ can be given a group structure where the elements are points, and the operation is point addition. 

Given two points $P_1$ \grm{and} $P_2$, \grm{we} define their summation in the following way:
Draw a line through $P_1$ and $P_2$, and let $P_3'$ be the point where the line intersects $E$ again, as in \rev{Fig.}~\ref{fig:elliptic-curve}. Then reflect $P_3'$ across the $x$-axis to obtain the point $P_3 = P_1 + P_2$. We can also find $2P_1 = P_1 + P_1$ by drawing the tangent line to the curve.

The benefit of this structure to the field of cryptography is that point addition on elliptic curves is quite difficult and \grm{time-consuming}. Moreover, if we are given two points $P$ and $Q$, and told that $kP = Q$, it is very hard and time-consuming to find $k$. Classical methods of solving this problem have faster specializations for certain groups, which means that for the problem to be hard, the group in question must be a large prime field. However, elliptic curve methods are equally difficult over prime groups and other similarly sized generic groups, and so they have no such specializations. Thus cryptography using elliptic curves is more efficient than using classical methods, because the elliptic curve variations of the classical methods offer more security over smaller groups. Solving this problem, which is called the \textit{discrete logarithm problem}, is central to elliptic curve cryptography. 

In this work, we mostly use BLS12-381 \cite{bls}, which is a specific, \grm{widely used} pairing-friendly elliptic curve. This curve is defined over finite fields and has algebraic properties that facilitate pairings.

\subsection{Security \rev{parameter} $\lambda$}\label{app:security-parameter}

The security parameter $\lambda \in \mathbb{N}$ (with $\lambda \geq 128$ being a typical choice) is an integer that controls every piece of the cryptosystem’s machinery. This security parameter controls the \emph{computational difficulty} of breaking the cryptosystem: all probabilities of adversary success must be negligible in $\lambda$, which means they must go to zero faster than any inverse \rev{polynomial}. Concretely, for $\lambda = 128$ the number of operations an adversary has to perform is at least of the order of magnitude of $2^{128}$. In the single trusted setup $\rev{\mathsf{Setup}_{\mathsf{TC}}}(\lambda, \mathbf{d})$, the security parameter controls \grm{the} following three interrelated values:

\begin{enumerate}
    \item \textbf{Fields and Group Sizes:}
    The chosen curve provides a pairing-friendly elliptic curve \rev{with} groups $\rev{G}$ and $\rev{G_T}$ both of \rev{prime order} $p$ along with a bilinear map $e: \rev{G} \times \rev{G} \rightarrow \rev{G_T}$\rev{, where $p$ is a prime of about $2\lambda$ bits. The scalar field $\mathbb{F}=\mathbb{Z}_p$ of exponents has the same order $p$.}
    For our instantiation, BLS12-381~\citep{bls} is employed, \rev{where $p$ is a $255$-bit prime, so} $|\mathbb{F}|\rev{{}=p}$ \rev{with $2^{254}<p<2^{255}$. Pollard's rho algorithm then needs $\sqrt{p}\approx 2^{127.4}$, i.e., about $2^{128}$ group operations,} giving an approximate security of $\rev{\lambda={}}128$ bits against discrete logarithm attacks. \rev{The curve itself is defined over a different base field $\mathbb{F}_{p_{\mathrm{base}}}$ with a $381$-bit prime $p_{\mathrm{base}}$, which sets the size of a group element.}
    Choosing $\lambda$ greater implies that $p$ will be bigger, hence \rev{enlarging both the} group elements \rev{and} the field $\mathbb{F}$ over which tensor values, \grm{coefficients} of polynomials, and trapdoors live.

    \item \textbf{Trapdoor Entropy:} For a tensor of order $m$ with dimensions $\mathbf{d} = (d_1, \ldots, d_m)$, $\rev{\mathsf{Setup}_{\mathsf{TC}}}$ selects $m$ independent trapdoor scalars $\tau_1, \ldots, \tau_m \leftarrow{\$} \mathbb{F}$. Given that $|\mathbb{F}|$ can be bounded via $\lambda$, each $\tau_j$ contains at least $\lambda$ bits of min-entropy. The structured reference string $\rev{\mathsf{srs}} = \bigl\{g^{\tau_1^{i_1} \cdots \tau_m^{i_m}}\bigr\}_{(i_1,\ldots,i_m) \in \mathcal{I}}$ is then published, while the trapdoors themselves are securely discarded. All subsequent commitments and verifications operate exclusively over the $\rev{\mathsf{srs}}$; knowledge of any $\tau_j$ would break the scheme.

    \item \textbf{Hardness Assumption Threshold:} \rev{App.~\ref{app:proof} reduces evaluation binding to the $q$-SBDH assumption with $q=\sum_{j=1}^m(d_j-1)$.} Given $g, g^{\tau}, \ldots, g^{\tau^{\rev{q}}} \in \rev{G}$ for a random $\tau \in \mathbb{F}$, no PPT adversary can compute the pair $(\rev{a},\, e(g,g)^{1/(\tau + \rev{a})})$ except with probability $\mathrm{negl}(\lambda)$. Since $\rev{q}$ grows polynomially in $\lambda$ and the group order grows exponentially, increasing $\lambda$ widens the gap between what the adversary receives and what it would need to forge.
\end{enumerate}

The effect of increasing the security parameter can be summarized along three axes detailed below.

\emph{Security margin:} The best-known attack strategy of the opponent entails complexity of $\rev{\mathcal{O}}(\rev{\sqrt{p}})\rev{{}=\mathcal{O}(2^{\lambda})}$ using Pollard's rho algorithm on the elliptic curve group (or \rev{fewer operations} in case of a pairing target attack on some specific elliptic curves). \rev{Doubling} $\lambda$ \rev{squares} the computation cost \rev{$2^{\lambda}$} for the adversary, hence creating a wider security margin. With $\lambda=128$, the BLS12-381 elliptic curve is considered safe from any conceivable hardware.

\emph{Group-element and proof sizes:} Each element in the group $\rev{G}$ requires $\rev{\mathcal{O}}(\lambda)$ bits.
The commitment $C_T$ requires one element from the group, while the proof $\pi_{\rev{\boldsymbol{\omega}}}$ involves $m$ elements (corresponding to the tensor axes), with extra Terkle-tree path elements requiring $\rev{\mathcal{O}}(\log_B n)$ \grm{elements. Therefore}, an increase in $\lambda$ will cause the byte size of proofs and commitments to grow linearly. In BLS12-381 ($\lambda \approx 128$), a compressed group element takes up 48 bytes, meaning that a tensor proof of rank $m$ will require around $48m$ \grm{bytes.}

\emph{Costs:} Both parties spend most of their time performing group exponentiations and evaluating pairings. Each one requires $\widetilde{\rev{\mathcal{O}}}(\lambda^2)$ bit operations (or $\widetilde{\rev{\mathcal{O}}}(\lambda)$ field multiplications via optimized library routines) due to computations modulo an integer of $\Omega(\lambda)$ bits. Doubling $\lambda$ leads to a factor-of-four increase in costs per operation.
But the number of pairings the verifier needs to check for each challenged coordinate - $\rev{\mathcal{O}}(\log_B n)$ - is independent of $\lambda$. In reality, the overhead comes from the prover's polynomial interpolation process, which depends on the tensor's dimensions but not $\lambda$, and therefore does not change much upon a modest increase in $\lambda$.

\subsection{Cryptography \rev{literature}}
\textbf{Zero-Knowledge Proofs for Computation:} Zero-knowledge SNARKs~\cite{groth2016size,parno2013pinocchio,ben2014succinct} provide succinct proofs of arbitrary computations with strong cryptographic guarantees. These systems and their recent versions ~\cite{maller2019sonic,gabizon2019plonk} achieve remarkable succinctness, but incur significant prover \grm{overhead~-} often minutes per query for large circuits. ZK-STARKs~\cite{ben2018scalable} offer transparency without a trusted setup and post-quantum security with larger proof sizes. Recently, recursive SNARKs~\cite{boneh2020halo} enable proof composition but still \grm{face} scalability challenges at LLM scale.

\textbf{Verifiable Computation:} The foundations of verifiable computation~\cite{goldwasser2008delegating} established interactive proof systems where a weak verifier can check a powerful prover's work. Modern systems like Spartan~\cite{setty2020spartan} and doubly-efficient protocols~\cite{wahby2018doubly} reduce the verification complexity but still require the verifier to see the computation structure. Unlike these general-purpose systems, we exploit the specific structure of LLM \grm{inference~-} tensor operations over known \grm{architectures~-} \rev{for} better constants.

\textbf{Polynomial Commitments:} Our TensorCommitments extend the KZG polynomial commitment scheme~\cite{kgz}, which commits to univariate polynomials using bilinear pairings\rev{, and instantiate the multivariate KZG of Papamanthou et al.~\cite{pst13} on tensor grids}. Recent advances include Bulletproofs~\cite{bunz2018bulletproofs} which are transparent but with linear verification and FRI-based commitments~\cite{ben2018fast} which are widely used in STARKs. The multivariate polynomial interpolation theory we leverage~\cite{polybound2017} has been extensively studied in numerical analysis but, to our knowledge, has not been \rev{used to commit to tensors on their native grids or} for ML verification. Vector commitments~\cite{catalano2013vector} provide the closest comparison, but our tensor-native approach fundamentally differs by preserving dimensional structure.

\section{Proof of \rev{construction}\label{app:proof}}
Here we \rev{prove} that our tensor commitment construction $\mathsf{TC}$ of Def.~\ref{def:tc} \rev{is correct and position binding. $\mathsf{TC}$ instantiates the multivariate KZG commitment of Papamanthou et al.~\cite{pst13} on the tensor grid $\mathbf{\Omega}=\Omega_1\times\cdots\times\Omega_m\subset\mathbb{F}^m$ of App.~\ref{app:interp}, and we give self-contained proofs in our notation. Correctness means that} honestly generated commitments and openings always verify. \rev{Position binding means that, once a commitment is fixed,} a cheating prover cannot open a \rev{grid} point to two distinct values. \rev{Binding thus fixes the committed tensor. It does not by itself certify that the tensor was computed correctly (Sec.~\ref{sec:problem}). Throughout, $d_j\!\ge\!2$ for all $j\!\in\![m]$ (an axis of size one can be dropped), so $g^{\tau_j}\!\in\!\mathsf{srs}$.}

\textbf{Correctness:} Let $\mathsf{pp}$ be the output of \rev{$\mathsf{Setup}_{\mathsf{TC}}(\lambda,\mathbf{d})$ of Def.~}\ref{def:tc}. We prove that for any function value $T$ on grid \rev{$\mathbf{\Omega}$ with $C_T\!=\!\mathsf{Com}_{\mathsf{TC}}(\mathsf{pp},T)$}, and any challenge point $\boldsymbol{\omega} = (\omega_1,\dots,\omega_m)$ \rev{with $(y,\pi_{\boldsymbol{\omega}})\!=\!\mathsf{Open}_{\mathsf{TC}}(\mathsf{pp},T,\boldsymbol{\omega})$,}
\begin{equation}
\label{eq:correctness}
    \Pr[\mathsf{Ver}_{\mathsf{TC}}(\mathsf{pp}, C_T, \boldsymbol{\omega}, y, \pi_{\boldsymbol{\omega}}) = 1] = 1.
\end{equation}

From the construction, $C_T =g^{f_T(\tau_1,\dots,\tau_m)}$ commits to \rev{the} multivariate polynomial $f_T \in \mathbb{F}[X_1,\dots,X_m]$, where \( f_T\bigl(\boldsymbol{\omega}\bigr)\!=\!T[\boldsymbol{\omega}], \forall \boldsymbol{\omega}\!\in\!\mathbf{\Omega}\)\rev{, and $\mathsf{Open}_{\mathsf{TC}}$ runs the $m$ division steps $r_{i-1}=q_i(X_i,\dots,X_m)\,(X_i-\omega_i)+r_i$}, \leavevmode{\revcolor $i\in[m]$, of App.~\ref{app:polydiv} with $r_0=f_T$ and scalar $r_m=y$, and outputs $\pi_{\omega_i}=g^{q_i(\tau_i,\dots,\tau_m)}$. Summing the $m$ steps, the intermediate remainders cancel:
\begin{align}\label{eq:telescoping}
    f_T(X_1,\dots,X_m)-y \;=\; \sum_{i=1}^m \bigl(r_{i-1}-r_i\bigr) \;=\; \sum_{i=1}^m q_i(X_i,\dots,X_m)\,(X_i-\omega_i).
\end{align}
Setting $X_i\!=\!\omega_i$ for all $i$ in~\eqref{eq:telescoping} zeroes every summand, so $y\!=\!f_T(\boldsymbol{\omega})$. As $f_T$ has degree at most $d_k\!-\!1$ in each $X_k$, so does every $r_{i-1}$. Hence $q_i$ has degree at most $d_i\!-\!2$ in $X_i$ and at most $d_k\!-\!1$ in $X_k$ for $k\!>\!i$, and does not depend on $X_1,\dots,X_{i-1}$. All monomials of $q_i$ thus lie in $\mathcal{I}$,} and \rev{$\pi_{\omega_i}$} is \rev{computed from $\mathsf{srs}$ by linearity in the exponent}. To prove~\rev{\eqref{eq:correctness}}, it suffices to show
{\revcolor \begin{align}\label{eq:corr-pairing}
    e(C_T\!\cdot\!g^{-y},\, g)\;=\;\prod_{i=1}^m e\bigl(\pi_{\omega_i},\, g^{\tau_i - \omega_i}\bigr).
\end{align}}\ignorespaces
From the construction and using the properties of pairings \rev{listed in App.~\ref{app:pairing-EC}}, we simplify each \rev{side of~\eqref{eq:corr-pairing}} as follows.
{\revcolor \begin{align}\label{eq:corr-pairing-left}
    e(C_T\!\cdot\!g^{-y},\, g) &= e\bigl(g^{f_T(\tau_1,\dots,\tau_m)} \cdot g^{-y},\,g\bigr) = e(g,g)^{f_T(\tau_1,\dots,\tau_m) - y},\\ \label{eq:corr-pairing-right}
    \prod_{i=1}^m e\bigl(\pi_{\omega_i},\, g^{\tau_i - \omega_i}\bigr) &= \prod_{i=1}^m e\bigl(g^{q_i(\tau_i,\dots,\tau_m)},\, g^{\tau_i - \omega_i}\bigr) = e(g,g)^{\sum_{i=1}^m q_i(\tau_i,\dots,\tau_m)\,(\tau_i - \omega_i)}.
\end{align}}\ignorespaces
Hence, proving the pairing equality of~\rev{\eqref{eq:corr-pairing}} reduces to \rev{equating} the exponents of~\rev{\eqref{eq:corr-pairing-left}} and~\rev{\eqref{eq:corr-pairing-right}}:
{\revcolor \begin{align}
    f_T(\tau_1,\dots,\tau_m) - y = \sum_{i=1}^m q_i(\tau_i,\dots,\tau_m)\,(\tau_i - \omega_i),
\end{align}}\ignorespaces
\rev{which is the polynomial identity~\eqref{eq:telescoping} evaluated at the trapdoors $X_i\!=\!\tau_i$}, \rev{$i\in[m]$}. \qed

\paragraph{Position \rev{(evaluation)} binding:} \rev{We show that for every PPT adversary} \rev{$\mathcal{A}$ and $\mathsf{pp}\!\leftarrow\!\mathsf{Setup}_{\mathsf{TC}}(\lambda,\mathbf{d})$, the probability that $\mathcal{A}(\mathsf{pp})$ outputs $(C,\boldsymbol{\omega},y,\pi_{\boldsymbol{\omega}},y',\pi'_{\boldsymbol{\omega}})$} \leavevmode{\revcolor with an arbitrary $C\!\in\!G$, $\boldsymbol{\omega}\!\in\!\mathbf{\Omega}$, and $y\!\neq\!y'$, such that {\small $\mathsf{Ver}_{\mathsf{TC}}(\mathsf{pp},C,\boldsymbol{\omega},y,\pi_{\boldsymbol{\omega}})\!=\!\mathsf{Ver}_{\mathsf{TC}}(\mathsf{pp},C,\boldsymbol{\omega},y',\pi'_{\boldsymbol{\omega}})\!=\!1$}, is negligible in $\lambda$. We require $y\!\neq\!y'$ because, for $m\!\ge\!2$, proof tuples are not unique. Replacing $(\pi_{\omega_1},\pi_{\omega_2})$ by $(\pi_{\omega_1}g^{\rev{b}(\tau_2-\omega_2)},\,\pi_{\omega_2}g^{-\rev{b}(\tau_1-\omega_1)})$ for any $\rev{b}\!\in\!\mathbb{F}$ leaves the right-hand side of~\eqref{eq:corr-pairing} unchanged. This non-uniqueness is harmless, as the verifier only uses $y$. We reduce binding to the $q$-Strong Bilinear Diffie-Hellman ($q$-SBDH) assumption used by~\cite{pst13}, with $q\!=\!\sum_{j=1}^m(d_j\!-\!1)$. This assumption states that, given $(g,g^{\rev{\tau}},\dots,g^{\rev{\tau^{q}}})$ for a uniform $\rev{\tau}\!\in\!\mathbb{F}$, no PPT algorithm outputs $(\rev{a},\,e(g,g)^{1/(\rev{\tau}+\rev{a})})$ for some $\rev{a}\!\in\!\mathbb{F}\setminus\{-\rev{\tau}\}$, except with probability $\mathrm{negl}(\lambda)$ (App.~\ref{app:security-parameter}).

Suppose $\mathcal{A}$ wins with probability $\varepsilon_{\mathcal{A}}$. We build $\mathcal{B}$ that, on a $q$-SBDH instance, samples a guess $\boldsymbol{\omega}^{*}\!\in\!\mathbf{\Omega}$ and $\psi_2,\dots,\psi_m\!\in\!\mathbb{F}$ uniformly, sets $\psi_1\!:=\!1$}, \rev{and implicitly fixes the trapdoors $\tau_1\!:=\!\rev{\tau}$ and $\tau_i\!:=\!\psi_i(\rev{\tau}-\omega^*_1)+\omega^*_i$ for $i\!\ge\!2$, so that $\tau_i-\omega^*_i=\psi_i\,(\rev{\tau}-\omega^*_1)$ for all $i\!\in\![m]$. Every $\mathsf{srs}$ entry $g^{\tau_1^{k_1}\cdots\tau_m^{k_m}}$, $(k_1,\dots,k_m)\!\in\!\mathcal{I}$, is $g$ raised to a polynomial in $\rev{\tau}$ of degree at most $k_1\!+\!\dots\!+\!k_m\!\le\!q$ whose coefficients $\mathcal{B}$ knows, so $\mathcal{B}$ computes $\mathsf{pp}$ from its instance. Since $\rev{\tau}$ and the $\psi_i$ are uniform, for every $\boldsymbol{\omega}^*$ the trapdoors $(\tau_1,\dots,\tau_m)$ are uniform on $\mathbb{F}^m$ up to a statistical distance $\delta_{\mathrm{sd}}\!=\!\mathcal{O}(1/p)$. Hence $\mathsf{pp}$ is distributed as in $\mathsf{Setup}_{\mathsf{TC}}$, $\boldsymbol{\omega}^*$ is hidden from $\mathcal{A}$, and $\mathcal{A}$ wins with $\boldsymbol{\omega}\!=\!\boldsymbol{\omega}^*$ with probability at least $\varepsilon_{\mathcal{A}}/|\mathbf{\Omega}|-\delta_{\mathrm{sd}}$. In that case, both verification equations hold at $\boldsymbol{\omega}^*$, and dividing them by bilinearity (App.~\ref{app:pairing-EC}) gives}
{\revcolor \begin{align}\label{eq:sec-binding}
    \begin{cases}
        e(C\!\cdot\!g^{-y},\, g) = \prod_{i=1}^m e\bigl(\pi_{\omega_i},\, g^{\tau_i - \omega^*_i}\bigr) \\
        e(C\!\cdot\!g^{-y'},\, g) = \prod_{i=1}^m e\bigl(\pi'_{\omega_i},\, g^{\tau_i - \omega^*_i}\bigr)
    \end{cases}
    \Rightarrow\quad
    e(g,g)^{y'-y} = \prod_{i=1}^m e\Bigl(\frac{\pi_{\omega_i}}{\pi'_{\omega_i}},\, g^{\tau_i-\omega^*_i}\Bigr).
\end{align}}\ignorespaces
\rev{Substituting $\tau_i-\omega^*_i=\psi_i(\rev{\tau}-\omega^*_1)$ and letting $\rev{g_\star}\!:=\!\prod_{i=1}^m(\pi_{\omega_i}/\pi'_{\omega_i})^{\psi_i}\!\in\!G$, which $\mathcal{B}$ computes, the} right-hand side of~\rev{\eqref{eq:sec-binding}} \rev{equals $e(\rev{g_\star},g)^{\rev{\tau}-\omega^*_1}$. As $y'\!\neq\!y$,} the left-hand side \rev{is not $1_{G_T}$, so $\rev{\tau}\!\neq\!\omega^*_1$, and raising both sides to $1/\bigl((y'-y)(\rev{\tau}-\omega^*_1)\bigr)$ yields}
{\revcolor \begin{align}\nonumber
    e(g,g)^{1/(\rev{\tau}-\omega^*_1)} = e(\rev{g_\star},g)^{1/(y'-y)}.
\end{align}}\ignorespaces
\rev{Thus, $\mathcal{B}$ outputs $(\rev{a},\,e(\rev{g_\star},g)^{1/(y'-y)})$ with $\rev{a}\!:=\!-\omega^*_1$, breaking $q$-SBDH with probability at least $\varepsilon_{\mathcal{A}}/|\mathbf{\Omega}|-\delta_{\mathrm{sd}}$. Since $|\mathbf{\Omega}|\!=\!\prod_{j}d_j$ is polynomial in $\lambda$, $\varepsilon_{\mathcal{A}}\le|\mathbf{\Omega}|\bigl(\mathrm{negl}(\lambda)+\delta_{\mathrm{sd}}\bigr)=\mathrm{negl}(\lambda)$. If $\boldsymbol{\omega}$ is instead sampled by the verifier after $C$ is fixed (e.g., a random point in $\mathbb{F}^m$), $\mathcal{B}$ samples $\boldsymbol{\omega}^*$ itself and the factor $|\mathbf{\Omega}|$} \rev{disappears}. \qed

\section{Ablation \rev{study on the number of verifiers}}
\begin{figure}[H]
    \centering
    \includegraphics[width=0.8\linewidth]{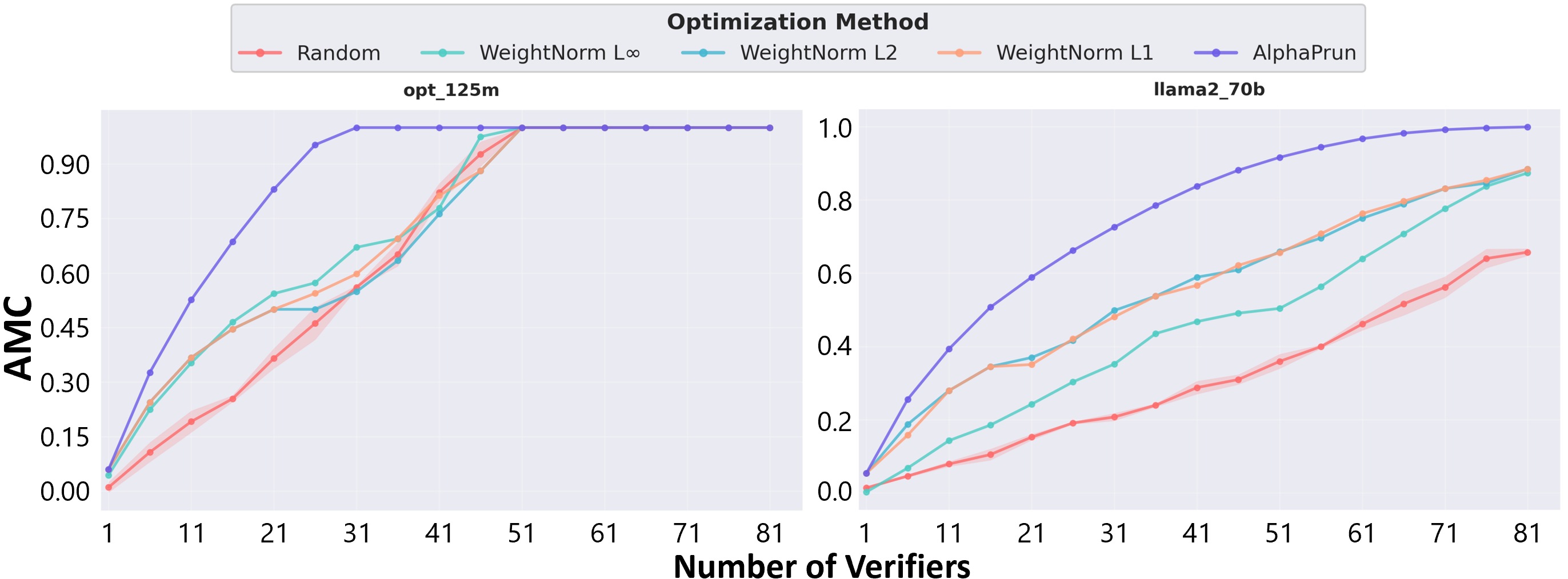}
    \caption{\textbf{How does scaling the verifier nodes change attack coverage?}  For OPT-125M (left) and LLaMA2-70B (right), we solve Problem~\rev{4.5} under a budget where each verifier can check at most $1\%$ of consecutive layers, and measure AMC (fraction of attacked layers covered) as we increase the number of verifiers for different layer-ranking objectives. The $\alpha$-score objective \cite{alphapruning} consistently dominates other objectives, reaching $0.8$ AMC with $21$ verifiers on OPT-125M and requiring only $36$ verifiers to achieve similar coverage on LLaMA2-70B. Quantitatively, to achieve 80\% AMC, the system needs 1.7$\times$ more verifiers when model size scales up by \grm{a factor of} 560.}
\end{figure}

We vary the number of verifiers $M$ from $1$ to $81$ and, for each $M$, solve Problem~\rev{4.5} under a fixed per-verifier budget (each can attest at most $1\%$ of consecutive layers). This isolates the effect of the allocation policy from raw compute. Across both OPT-125M and LLaMA2-70B, the $\alpha$-score objective based on heavy-tailed spectra consistently achieves the highest AMC, especially in the low-budget regime where each additional verifier must be used carefully. 

The curves also show a favorable scaling trend: moving from OPT-125M to LLaMA2-70B (a $560\times$ increase in parameters) only requires about $1.7\times$ more verifiers to reach $80\%$ AMC. Most of the robustness gain comes from \emph{where} we place verifiers rather than from simply adding more of them, and this advantage persists as model size grows.

\section{Complexity of \rev{multivariate interpolation}}
\label{app:complexity}

In this appendix we make explicit the complexity of multivariate Newton
interpolation in the setting of Hecht et al.'s quadratic-time algorithm
for the Polynomial Interpolation Problem (PIP), and then re-express the
cost in terms of the dimension \(m\) and the total number of data points
\(\rev{N}\).

\subsection{Polynomial space and problem formulation:}

For integers \(m \ge 1\) and \(\rev{d} \ge 0\), let
\[
  \Pi_{m,\rev{d}}
  :=
  \bigl\{
    f(\rev{X_1},\dots,\rev{X_m}) \in \mathbb{R}[\rev{X_1},\dots,\rev{X_m}]
    : \deg_{\mathrm{total}} f \le \rev{d}
  \bigr\}
\]
denote the space of real polynomials in \(m\) variables with total
degree at most \(\rev{d}\). It is standard that
\begin{equation}
  \rev{\dim\Pi_{m,d}}
  = \binom{m + \rev{d}}{\rev{d}}
  \label{eq:Nmn-def}
\end{equation}
is the number of monomials of total degree at most \(\rev{d}\) in \(m\)
variables.

The \emph{Polynomial Interpolation Problem (PIP)} is: given
\(m,\rev{d} \in \mathbb{N}\) and a function \(F : \mathbb{R}^m \to \mathbb{R}\),
choose a node set
\(\rev{\Omega_{\mathrm{PIP}}} = \{\rev{\boldsymbol{\omega}_1},\dots,\rev{\boldsymbol{\omega}_{\dim\Pi_{m,d}}}\} \subset \mathbb{R}^m\) that is
(unisolvent) generic for \(\Pi_{m,\rev{d}}\) and compute the unique
interpolating polynomial \(\rev{f} \in \Pi_{m,\rev{d}}\) such that
\(\rev{f}(\rev{\boldsymbol{\omega}_i}) = F(\rev{\boldsymbol{\omega}_i})\) for all \(i\).

\subsection{Quadratic-time \rev{algorithm}:}

\grm{PIP-SOLVER \cite{polybound2017} is a recursive algorithm} based on a
decomposition into subproblems of lower dimension and lower degree. In
our notation, their main result can be stated as follows.

\begin{theorem}
\label{thm:pip-solver}
For all \(m,\rev{d} \in \mathbb{N}\) with \(m \ge 1\), there exists an
algorithm that, given oracle access to \(\rev{F} : \mathbb{R}^m \to \mathbb{R}\),

\begin{enumerate}
  \item constructs a generic node set
  \(\rev{\Omega_{\mathrm{PIP}}} \subset \mathbb{R}^m\) with \(\#\rev{\Omega_{\mathrm{PIP}}} = \rev{\dim\Pi_{m,d}}\), and
  \item computes the interpolating polynomial \(\rev{f} \in \Pi_{m,\rev{d}}\),
\end{enumerate}
with \rev{a} runtime \rev{of}
\begin{equation}
  \rev{\mathcal{C}(m,d)} = \mathcal{O}\bigl(\rev{(\dim\Pi_{m,d})^2}\bigr)
  \label{eq:pip-main-time}
\end{equation}
and \rev{a} storage \rev{cost of}
\begin{equation}
  \rev{\mathcal{R}(m,d)} = \mathcal{O}\bigl(m\,\rev{\dim\Pi_{m,d}}\bigr).
  \label{eq:pip-main-space}
\end{equation}
\end{theorem}

The key insight is to split the global PIP into:
\begin{itemize}
  \item a subproblem of dimension \(m-1\) and degree \(\rev{d}\), and
  \item a subproblem of dimension \(m\) and degree \(\rev{d}-1\),
\end{itemize}
and to recurse until only linear (\(\rev{d}=1\)) or univariate (\(m=1\))
subproblems remain, both of which admit direct \(\mathcal{O}(\rev{d^2})\) or
\(\mathcal{O}(m^2)\) algorithms.

We now sketch how the \(\mathcal{O}\bigl(\rev{(\dim\Pi_{m,d})^2}\bigr)\) bound arises
from the recursive decomposition. This is not a re-proof of correctness,
but a transparent derivation of the runtime for readers unfamiliar with
complexity arguments.

\paragraph{Recurrence:}
Let \(\rev{\mathcal{C}(m,d)}\) denote the number of arithmetic operations performed by
PIP-SOLVER \cite{polybound2017} for parameters \((m,\rev{d})\). The algorithm performs three main
types of work:
\begin{enumerate}
  \item Solve a subproblem on a hyperplane \(\rev{\mathcal{H}} \cong \mathbb{R}^{m-1}\)
        with parameters \((m-1,\rev{d})\), costing \(\rev{\mathcal{C}(m-1,d)}\) \rev{operations}.
  \item Solve a subproblem in full dimension with reduced degree,
        parameters \((m,\rev{d}-1)\), costing \(\rev{\mathcal{C}(m,d-1)}\) \rev{operations}.
  \item Combine the two solutions into the global interpolant \(\rev{f}\),
        which involves a linear number of operations per coefficient.
\end{enumerate}

The decomposition is structured so that the total number of nodes used
by the two subproblems equals the total number of nodes of the original
problem:
\begin{equation}
  \rev{\dim\Pi_{m-1,d}} + \rev{\dim\Pi_{m,d-1}} = \rev{\dim\Pi_{m,d}}.
  \label{eq:N-split}
\end{equation}
This identity is easy to verify directly from
\eqref{eq:Nmn-def}:
\[
  \binom{m-1 + \rev{d}}{\rev{d}} + \binom{m + \rev{d} - 1}{\rev{d}-1}
  = \binom{m+\rev{d}}{\rev{d}}.
\]

The combination step (forming \(\rev{f} = \rev{f_1} + \rev{f_{\mathcal{H}}} \rev{f_2}\) \rev{from the solutions \(f_1\in\Pi_{m-1,d}\) and \(f_2\in\Pi_{m,d-1}\) of subproblems~1 and~2}, \rev{following}
Hecht et al.) requires multiplying \rev{the} linear polynomial \(\rev{f_{\mathcal{H}}}\)\rev{, whose zero set is \(\mathcal{H}\),} by a
polynomial with \(\rev{\dim\Pi_{m,d-1}}\) coefficients and then adding this to a
polynomial with \(\rev{\dim\Pi_{m-1,d}}\) coefficients. This can be done in a number
of arithmetic operations proportional to
\[
  C_3 \, \rev{\dim\Pi_{m-1,d}}\,\rev{\dim\Pi_{m,d-1}}
\]
for some constant \(C_3\), because each coefficient of the product
depends on a constant number of coefficients of the factors and there
are \(\rev{\dim\Pi_{m-1,d}}\) and \(\rev{\dim\Pi_{m,d-1}}\) coefficients to process.

Altogether, we obtain a recurrence of the form
\begin{equation}
  \rev{\mathcal{C}(m,d)}
  \;\le\;
  \rev{\mathcal{C}(m-1,d)} + \rev{\mathcal{C}(m,d-1)}
  + C_3\,\rev{\dim\Pi_{m-1,d}}\,\rev{\dim\Pi_{m,d-1}},
  \label{eq:T-rec}
\end{equation}
for \(m,\rev{d} > 1\), with base cases \(\rev{\mathcal{C}(1,d)} = \mathcal{O}(\rev{d^2})\) and
\(\rev{\mathcal{C}(m,1)} = \mathcal{O}(m^2)\) \rev{(univariate}
and linear interpolation).

\paragraph{Big-\(\mathcal{O}\) and \(\Theta\) notation:}
For clarity: for two non-negative functions \(A,B\) of our parameters,
\begin{itemize}
  \item \(A = \mathcal{O}(B)\) means that there exists a constant
        \(C > 0\) and a parameter region where \(A \le C B\),
  \item \(A = \Theta(B)\) means \(A = \mathcal{O}(B)\) \emph{and}
        \(B = \mathcal{O}(A)\). In other words, \(A\) and \(B\) are
        bounded above and below by constant multiples of each other.
\end{itemize}
We will show that \(\rev{\mathcal{C}(m,d)}\) grows on the order of \(\rev{(\dim\Pi_{m,d})^2}\).

\paragraph{Induction hypothesis:}
Assume there exists a constant \(C > 0\) such that for all smaller
parameter pairs \((m',\rev{d'})\) (in the sense \(\rev{\dim\Pi_{m',d'}} < \rev{\dim\Pi_{m,d}}\)) we have
\begin{equation}
  \rev{\mathcal{C}(m',d')} \le C\,\rev{(\dim\Pi_{m',d'})^2}.
  \label{eq:IH}
\end{equation}
This holds for the base cases by choosing \(C\) large enough.

\paragraph{Bounding the recurrence:}
Using \eqref{eq:T-rec} and the induction hypothesis \eqref{eq:IH} we
obtain for \(m,\rev{d}>1\):
\begin{align*}
  \rev{\mathcal{C}(m,d)}
  &\le C\,\rev{(\dim\Pi_{m-1,d})^2} + C\,\rev{(\dim\Pi_{m,d-1})^2}
      + C_3\,\rev{\dim\Pi_{m-1,d}}\,\rev{\dim\Pi_{m,d-1}} \\
  &= C\bigl( \rev{(\dim\Pi_{m-1,d})^2} + \rev{(\dim\Pi_{m,d-1})^2} \bigr)
    + C_3\,\rev{\dim\Pi_{m-1,d}}\,\rev{\dim\Pi_{m,d-1}}.
\end{align*}
Using the identity \eqref{eq:N-split}, we can rewrite
\(\rev{\dim\Pi_{m,d}} = \rev{\dim\Pi_{m-1,d}} + \rev{\dim\Pi_{m,d-1}}\). Then
\begin{align*}
  \rev{(\dim\Pi_{m,d})^2}
  &= \bigl(\rev{\dim\Pi_{m-1,d}} + \rev{\dim\Pi_{m,d-1}}\bigr)^2 \\
  &= \rev{(\dim\Pi_{m-1,d})^2} + \rev{(\dim\Pi_{m,d-1})^2}
     + 2\,\rev{\dim\Pi_{m-1,d}}\,\rev{\dim\Pi_{m,d-1}}.
\end{align*}
Therefore, \rev{we have}
\begin{align*}
  \rev{\mathcal{C}(m,d)}
  &\le C\,\Bigl( \rev{(\dim\Pi_{m,d})^2}
                 - 2\,\rev{\dim\Pi_{m-1,d}}\,\rev{\dim\Pi_{m,d-1}} \Bigr) \\
  &\quad + C_3\,\rev{\dim\Pi_{m-1,d}}\,\rev{\dim\Pi_{m,d-1}} \\
  &= C\,\rev{(\dim\Pi_{m,d})^2}
     + \bigl(C_3 - 2C\bigr)\,\rev{\dim\Pi_{m-1,d}}\,\rev{\dim\Pi_{m,d-1}}.
\end{align*}
If we choose \(C \ge \tfrac{1}{2}C_3\), then \((C_3 - 2C) \le 0\), so the
last term is non-positive. Hence
\[
  \rev{\mathcal{C}(m,d)} \le C\,\rev{(\dim\Pi_{m,d})^2}.
\]
By induction, this holds for all \(m,\rev{d}\), which proves
\eqref{eq:pip-main-time}. Since the algorithm really does perform at
least a constant fraction of these operations, we in fact have the
sharper statement
\[
  \rev{\mathcal{C}(m,d)} = \Theta\bigl(\rev{(\dim\Pi_{m,d})^2}\bigr).
\]

\subsection{Complexity in terms of grid size \(\rev{N}\) and dimension \(m\):}

We now connect the above complexity to a grid of data points. Suppose we
have a regular \(m\)-dimensional grid with \(\rev{d}\) nodes in each direction,
so that the total number of grid points is
\begin{equation}
  \rev{N} := \rev{d^m}.
  \label{eq:D-def}
\end{equation}
In practice, we often think of \(\rev{N}\) (the number of samples) and \(m\)
(the dimension) as given, and wish to understand the cost as a function
of these two quantities.

\rev{The bounds quoted in Fig.~\ref{fig:polydecrease} and Prop.~4.4 (with \(N=B\)) come from the cost model of this appendix. In this model, the number of nodes per axis fixes the degree \(d\), and we interpolate in the total-degree space \(\Pi_{m,d}\). PIP-SOLVER evaluates \(F\) on its node set \(\Omega_{\mathrm{PIP}}\), with \(\#\Omega_{\mathrm{PIP}}=\dim\Pi_{m,d}\) (Theorem~\ref{thm:pip-solver}). All bounds below count the cost of this interpolant.}

Because \(\rev{d}\) must be an integer, given \(m\) and \(\rev{N}\) we define the
largest admissible \rev{number of nodes per axis, and hence the degree,} as
\begin{equation}
  \rev{d}(\rev{N},m) := \max\{ \rev{d} \in \mathbb{N} : \rev{d^m} \le \rev{N} \}
           = \bigl\lfloor \rev{N^{1/m}} \bigr\rfloor.
  \label{eq:nDm-def}
\end{equation}
By definition, this implies
\begin{equation}
  \rev{d}(\rev{N},m)^m \;\le\; \rev{N} \;<\; \bigl(\rev{d}(\rev{N},m)+1\bigr)^m.
  \label{eq:nDm-ineq}
\end{equation}

\rev{By \eqref{eq:Nmn-def},} the corresponding polynomial-space dimension \rev{is}
\begin{equation}
  \rev{\dim\Pi_{m,d(N,m)}}
  = \binom{m + \rev{d}(\rev{N},m)}{\rev{d}(\rev{N},m)}.
  \label{eq:NmD-def}
\end{equation}

\begin{proposition}[Complexity as a function of \(m\) and \(\rev{N}\)]
\label{prop:TmD}
Let \(m \ge 1\) and \(\rev{N} \ge 1\) be given and define \(\rev{d}(\rev{N},m)\) as in
\eqref{eq:nDm-def}. Then the PIP-SOLVER algorithm with degree
\(\rev{d}(\rev{N},m)\) has runtime
\begin{equation}
  \rev{\mathcal{C}\bigl(m,d(N,m)\bigr)}
  = \Theta\bigl( \rev{(\dim\Pi_{m,d(N,m)})^2} \bigr)
  = \Theta\!\left(
      \binom{m + \rev{d}(\rev{N},m)}{\rev{d}(\rev{N},m)}^2
    \right).
  \label{eq:TmD-main}
\end{equation}
\end{proposition}

\begin{proof}
By Theorem~\ref{thm:pip-solver},
\(\rev{\mathcal{C}\bigl(m,d(N,m)\bigr)} = \Theta\bigl( \rev{(\dim\Pi_{m,d(N,m)})^2} \bigr)\). \rev{Substituting}
\rev{the binomial form} \eqref{eq:NmD-def} \rev{of}
\(\rev{\dim\Pi_{m,d(N,m)}}\) yields \eqref{eq:TmD-main}.
\end{proof}

\rev{Since \(\binom{m+d}{d}=\binom{m+d}{m}\), \eqref{eq:TmD-main} is the bound of Fig.~\ref{fig:polydecrease}.} For intuition, \rev{dropping the floor gives} the slightly looser but
cleaner expression
\[
  \rev{\mathcal{C}\bigl(m,d(N,m)\bigr)} \approx
  \mathcal{O}\!\left(
    \binom{m + \rev{N^{1/m}}}{m}^2
  \right),
\]
keeping in mind that the exact integer-valued bound is given by
\eqref{eq:TmD-main}.

\subsection{Asymptotics for fixed \(m\) and growing \(\rev{N}\):}

In many applications, the dimension \(m\) is fixed while the grid is
refined so that \(\rev{N} \to \infty\). We now show that in this regime the
cost is essentially quadratic in the number of grid points.

Recall \rev{from \eqref{eq:Nmn-def}} that
\[
  \rev{\dim\Pi_{m,d}}
  = \binom{m+\rev{d}}{\rev{d}}
  = \frac{(\rev{d}+1)(\rev{d}+2)\cdots(\rev{d}+m)}{m!}.
\]
\rev{We factor} out \(\rev{d^m}\):
\begin{equation}
  \rev{\dim\Pi_{m,d}}
  = \frac{\rev{d^m}}{m!}
    \prod_{k=1}^m \left(1 + \frac{k}{\rev{d}}\right).
  \label{eq:Nmn-prod}
\end{equation}
For fixed \(m\) and \(\rev{d} \to \infty\), each factor
\(1 + \tfrac{k}{\rev{d}} \to 1\). Using the Taylor expansion
\(\log(1+x) = x + \mathcal{O}(x^2)\) as \(x \to 0\), we obtain
\[
  \log \prod_{k=1}^m \left(1 + \frac{k}{\rev{d}}\right)
  = \sum_{k=1}^m \log\!\left(1 + \frac{k}{\rev{d}}\right)
  = \mathcal{O}\!\left(\frac{1}{\rev{d}}\right),
\]
hence\rev{, by exponentiating,}
\[
  \prod_{k=1}^m \left(1 + \frac{k}{\rev{d}}\right)
  = 1 + \mathcal{O}\!\left(\frac{1}{\rev{d}}\right).
\]
Substituting this into \eqref{eq:Nmn-prod} gives
\begin{equation}
  \rev{\dim\Pi_{m,d}}
  = \frac{\rev{d^m}}{m!}
    \left(1 + \mathcal{O}\!\left(\frac{1}{\rev{d}}\right)\right)
  \qquad (\rev{d} \to \infty,\ m\ \text{fixed}).
  \label{eq:Nmn-asympt}
\end{equation}

Now set \(\rev{d} = \rev{d}(\rev{N},m)\). From \eqref{eq:nDm-ineq} we have
\(\rev{d}(\rev{N},m)^m \sim \rev{N}\), so
\[
  \rev{d}(\rev{N},m)
  = \rev{N^{1/m}} \bigl(1 + o(1)\bigr)
  \qquad (\rev{N}\to\infty).
\]
Plugging this into \eqref{eq:Nmn-asympt} yields
\begin{equation}
  \rev{\dim\Pi_{m,d(N,m)}}
  = \frac{\rev{N}}{m!}\,\bigl(1 + o(1)\bigr),
  \qquad \rev{N}\to\infty,\ m\ \text{fixed}.
  \label{eq:NmD-asympt}
\end{equation}
Combining \eqref{eq:NmD-asympt} with
\(\rev{\mathcal{C}\bigl(m,d(N,m)\bigr)} = \Theta\bigl(\rev{(\dim\Pi_{m,d(N,m)})^2}\bigr)\) gives:

\begin{corollary}[Fixed \(m\), growing grid]
\label{cor:fixed-m}
For fixed dimension \(m\) and \(\rev{N} \to \infty\),
\begin{equation}
  \rev{\mathcal{C}\bigl(m,d(N,m)\bigr)}
  = \Theta\!\left(
    \frac{\rev{N^2}}{(m!)^2}
  \right)
  = \Theta(\rev{N^2}).
  \label{eq:fixed-m-asympt}
\end{equation}
Thus, for fixed \(m\) the runtime is asymptotically quadratic in the
number of grid points \(\rev{N}\).
\end{corollary}

\subsection{Asymptotics for fixed \(\rev{N}\) and growing \(m\):}

The complementary regime is to fix the grid budget \(\rev{N} > 1\) and let the
dimension \(m\) increase. In this case we cannot keep \(\rev{d}\) large: the
\rev{number of nodes per axis, and hence the degree,} must shrink, and eventually \rev{collapses to}
\(1\).

\begin{lemma}[Degree collapse for fixed budget]
\label{lem:deg-collapse}
Fix \(\rev{N} > 1\) and define \(\rev{d}(\rev{N},m)\) by \eqref{eq:nDm-def}. Then there
exists \rev{a threshold}
\[
  \rev{m^\star}(\rev{N}) := \bigl\lfloor \log_2 \rev{N} \bigr\rfloor + 1
\]
such that \rev{the admissible degree satisfies}
\[
  \rev{d}(\rev{N},m) = 1 \qquad \text{for all } m \ge \rev{m^\star}(\rev{N}).
\]
\end{lemma}

\begin{proof}
Suppose \(\rev{d}(\rev{N},m) \ge 2\). Then \(\rev{d}(\rev{N},m)^m \ge 2^m\). By
\eqref{eq:nDm-ineq}, we also have \(\rev{d}(\rev{N},m)^m \le \rev{N}\), hence
\(2^m \le \rev{N}\). Thus, if \(2^m > \rev{N}\) we \emph{cannot} have \(\rev{d}(\rev{N},m) \ge 2\);
the only possibility is \(\rev{d}(\rev{N},m) = 1\).

Choose \(\rev{m^\star}(\rev{N}) := \lfloor \log_2 \rev{N} \rfloor + 1\). Then
\(2^{\rev{m^\star}(\rev{N})} > \rev{N}\), so for all \(m \ge \rev{m^\star}(\rev{N})\) we have \(2^m \ge 2^{\rev{m^\star}(\rev{N})} > \rev{N}\),
forcing \(\rev{d}(\rev{N},m) = 1\).
\end{proof}

When \(\rev{d}(\rev{N},m) = 1\), the polynomial space is simply the affine functions
\[
  \rev{f}(\rev{X_1},\dots,\rev{X_m})
  = a_0 + \sum_{j=1}^m a_j \rev{X_j},
\]
with dimension \rev{given by}
\begin{equation}
  \rev{\dim\Pi_{m,1}} = \binom{m+1}{1} = m+1.
  \label{eq:Nm1}
\end{equation}
Therefore, by Theorem~\ref{thm:pip-solver},
\[
  \rev{\mathcal{C}(m,1)}
  = \Theta\bigl(\rev{(\dim\Pi_{m,1})^2}\bigr)
  = \Theta\bigl((m+1)^2\bigr)
  = \Theta(m^2).
\]

Combining this with Lemma~\ref{lem:deg-collapse} yields:

\begin{theorem}[Fixed budget, increasing dimension]
\label{thm:fixed-D}
Fix a grid budget \(\rev{N} > 1\) and let \(\rev{d}(\rev{N},m)\) be defined as in
\eqref{eq:nDm-def}. Then the runtime of PIP-SOLVER satisfies
\begin{equation}
  \rev{\mathcal{C}\bigl(m,d(N,m)\bigr)}
  = \Theta(m^2)
  \qquad \text{for all } m \ge \rev{m^\star}(\rev{N}),
  \label{eq:fixed-D-asympt}
\end{equation}
where \(\rev{m^\star}(\rev{N}) = \lfloor \log_2 \rev{N} \rfloor + 1\). In words, under a fixed
grid budget, the interpolation cost grows only quadratically in the
dimension \(m\) once \(m\) is large enough.
\end{theorem}

To summarize, in terms of the dimension \(m\) and the grid budget \(\rev{N}\):

\begin{itemize}
  \item The \emph{exact} runtime of multivariate Newton interpolation
        using PIP-SOLVER with degree \(\rev{d}(\rev{N},m)\) is \rev{given by}
        \[
          \rev{\mathcal{C}\bigl(m,d(N,m)\bigr)}
          = \Theta\!\left(
              \binom{m + \rev{d}(\rev{N},m)}{\rev{d}(\rev{N},m)}^2
            \right),
          \quad \rev{d}(\rev{N},m) = \bigl\lfloor \rev{N^{1/m}} \bigr\rfloor.
        \]
  \item For fixed \(m\) and \(\rev{N} \to \infty\), the runtime is asymptotically
        quadratic in the \rev{grid size}:
        \[
          \rev{\mathcal{C}\bigl(m,d(N,m)\bigr)} = \Theta(\rev{N^2}).
        \]
  \item For fixed \(\rev{N}\) and \(m \to \infty\), the admissible degree
        collapses to \(1\) and the runtime is asymptotically quadratic in
        the dimension:
        \[
          \rev{\mathcal{C}\bigl(m,d(N,m)\bigr)} = \Theta(m^2)
          \quad (m \ge \rev{m^\star}(\rev{N})).
        \]
\end{itemize}

{\revcolor \section{Adaptive provers and randomized layer selection\label{app:adaptive}}
The layer selector of Sec.~\ref{subsec:layer-selection} solves Problem~4.5 once (App.~\ref{app:dp}) and audits the set $S_\alpha\subseteq[\mathcal{L}]$ of layers covered by the resulting allocation under the $\alpha$-score benefits $\nu_i$. Such a \emph{deterministic} selector audits layer $\ell$ with probability $p_\ell=\mathbf{1}[\ell\in S_\alpha]\in\{0,1\}$. A prover that knows the policy therefore knows $S_\alpha$, and a deviation placed outside $S_\alpha$ is never examined, regardless of its effect on the output. Here, we detail the randomized selector $\mathcal{S}_\gamma$ of Sec.~\ref{subsec:layer-selection}, its worst-case coverage floor, and the adaptive $\gamma$-sweep summarized in Sec.~\ref{sec:experiments}.

\textbf{Threat Model:} We strengthen the adversary of App.~\ref{app:attacks}. Besides tampering with the computation (but not the cryptography), the prover knows the selection policy $\mathcal{S}_\gamma$ including $\gamma$, the $\alpha$-scores, and the budgets $\beta_1,\dots,\beta_M$. Since the Terkle root is published before the verifier draws its selection, the prover knows the distribution of the audited set, but not its realization. It places its deviation where it is least likely to be examined, i.e., outside $S_\alpha$ under deterministic selection ($\gamma=1$). Under uniform selection ($\gamma=0$), which may examine any layer, it targets the highest-impact layers.

\textbf{Randomized Selector:} Let $\rev{\mathcal{U}}$ be the set of feasible allocations of Problem~4.5, i.e., assignments of at most one contiguous interval $[s,\rev{o}]\subseteq[\mathcal{L}]$ to each verifier $k\in[M]$ with $\sum_{i=s}^{\rev{o}}\phi_i\le\beta_k$ and no overlap across verifiers. $\rev{\mathcal{U}}$ is finite and nonempty (it contains the empty allocation), and we write $\ell\in \rev{U}$ if allocation $\rev{U}\in\rev{\mathcal{U}}$ covers layer $\ell$. For a mixing weight $\gamma\in[0,1]$ (unrelated to the ILP variables $\gamma_{k,i}$), $\mathcal{S}_\gamma$ uses fresh randomness, drawn after the Terkle root is published, to return $S_\alpha$ with probability $\gamma$ and otherwise the layers covered by $\rev{U}\sim\mathrm{Unif}(\rev{\mathcal{U}})$. Thus, $\gamma=1$ recovers the deterministic $\alpha$-selector and $\gamma=0$ the uniform one. We define
\[
c_\ell:=\Pr_{\rev{U}\sim\mathrm{Unif}(\rev{\mathcal{U}})}\bigl[\ell\in \rev{U}\bigr]=\frac{\bigl|\{\rev{U}\in\rev{\mathcal{U}}:\ell\in \rev{U}\}\bigr|}{|\rev{\mathcal{U}}|},
\qquad
c:=\min_{\ell\in[\mathcal{L}]}c_\ell .
\]

\begin{proposition}[Coverage floor]
\label{prop:coverage-floor}
For every layer $\ell\in[\mathcal{L}]$, the probability that $\mathcal{S}_\gamma$ audits $\ell$ is
\[
p_\ell=\gamma\,\mathbf{1}[\ell\in S_\alpha]+(1-\gamma)\,c_\ell\;\ge\;(1-\gamma)\,c .
\]
Moreover, $c>0$ if and only if every layer belongs to at least one feasible allocation, i.e., $\phi_\ell\le\max_{k\in[M]}\beta_k$ for all $\ell\in[\mathcal{L}]$.
\end{proposition}

\begin{proof}[Proof of Proposition~\ref{prop:coverage-floor}]
By the law of total probability over the mixing coin, $p_\ell=\gamma\,\Pr[\ell\in S_\alpha]+(1-\gamma)\Pr_{\rev{U}\sim\mathrm{Unif}(\rev{\mathcal{U}})}[\ell\in \rev{U}]=\gamma\,\mathbf{1}[\ell\in S_\alpha]+(1-\gamma)\,c_\ell$, since $S_\alpha$ is fixed. The first term is nonnegative and $c_\ell\ge c$, which gives the floor. If $\phi_\ell\le\beta_k$ for some $k$, assigning the single-layer interval $[\ell,\ell]$ to verifier $k$ and no interval to the others is feasible, so $c_\ell\ge 1/|\rev{\mathcal{U}}|>0$. If this holds for every $\ell$, then $c>0$. Conversely, if a feasible allocation assigns an interval $[s,\rev{o}]\ni\ell$ to verifier $k$, then $\phi_\ell\le\sum_{i=s}^{\rev{o}}\phi_i\le\beta_k$, since all $\phi_i>0$. Hence, if $\phi_\ell>\max_{k}\beta_k$ for some $\ell$, no feasible allocation covers $\ell$, and $c\le c_\ell=0$.
\end{proof}

The floor depends only on $\gamma$ and, through $c$, on the costs $\phi_i$ and budgets $\beta_k$, not on the realized selection, which is drawn after the prover has committed. Hence, it holds even when the prover knows the policy, the $\alpha$-scores, and the budgets. No choice of the perturbed layer yields $p_\ell<(1-\gamma)c$. The argument only uses the coverage probabilities of the randomized component, so the same floor holds for any other distribution over $\rev{\mathcal{U}}$, with $c_\ell$ defined under that distribution. Since the single-layer bound of Prop.~\ref{prop:guarantee} is non-increasing in $p_\ell$, substituting the floor gives, for a deviation confined to layer $\ell$,
\[
\Pr[\mathcal{V}\text{ accepts } T]\;\le\;1-(1-\gamma)\,c\,\bigl[1-(1-\rho_\ell)^k\bigr]+\mathrm{negl}(\lambda),
\]
with $\rho_\ell$ and the number $k$ of audited coordinates as in Prop.~\ref{prop:guarantee}, whichever layer the prover chooses. At $\gamma=1$, the floor vanishes, and every $\ell\notin S_\alpha$ has $p_\ell=0$ and a vacuous bound.

\textbf{Adaptive $\gamma$-Sweep:} Table~\ref{tab:adaptive} sweeps $\gamma\in\{0,0.25,0.5,0.75,1\}$ against the policy-aware prover above. Over all attempted attacks, we report the fraction that is missed and changes the output (harmful miss) and the fraction that is missed, together with the share of missed attacks that change the output.}

{\revcolor \begin{table}[ht]\revcolor 
\centering
\caption{\textbf{How does randomized selection fare against a policy-aware prover?} Adaptive $\gamma$-sweep of $\mathcal{S}_\gamma$ ($\gamma\!=\!1$ gives the deterministic $\alpha$-selector $S_\alpha$, and $\gamma\!=\!0$ a uniformly sampled feasible allocation) against a prover that knows the policy, the $\alpha$-scores, and the budgets. The first two rows are fractions of all attempted attacks. The last row is a fraction of the missed attacks. Concentrating on $S_\alpha$ lowers the harmful-miss rate from \textbf{10.62\%} to \textbf{3.98\%}, but $\gamma\!=\!1$ misses every perturbation outside $S_\alpha$.}
\label{tab:adaptive}
\small
\begin{tabular}{lccccc}
\toprule
$\gamma$ ($1$ = $\alpha$-selector, $0$ = uniform random) & 0.00 & 0.25 & 0.50 & 0.75 & 1.00 \\
\midrule
Missed and output changed & \textbf{10.62\%} & 6.80\% & 5.38\% & 4.16\% & \textbf{3.98\%} \\
Missed perturbation & 19.82\% & 22.86\% & 35.12\% & 62.70\% & 100\% \\
Share of misses that changed the output & 53.6\% & 29.7\% & 15.3\% & 6.63\% & 3.98\% \\
\bottomrule
\end{tabular}
\end{table}

\textbf{Interpretation:} In every column, the first row is the product of the other two (up to rounding). At $\gamma=0$, $19.82\%$ of attempted attacks are missed and $53.6\%$ of those change the output, i.e., a harmful-miss rate of $10.62\%$. At $\gamma=1$, every perturbation deliberately placed outside $S_\alpha$ is missed ($100\%$, as $p_\ell=0$ there), yet only $3.98\%$ of them alter the output under this construction. As $\gamma$ grows and the selector concentrates on $S_\alpha$, raw misses rise from $19.82\%$ to $100\%$, while the share of output-changing misses falls from $53.6\%$ to $3.98\%$, so the harmful-miss rate decreases from $10.62\%$ to $3.98\%$. Hence, $\gamma$ trades the average harmful-miss rate, which is lowest at $\gamma=1$ for this attack family, against the worst-case floor $(1-\gamma)c$, which is zero at $\gamma=1$. An output-changing deviation that the prover places in a layer outside $S_\alpha$ is never audited by the deterministic selector, whereas any $\gamma<1$ audits every layer with probability at least $(1-\gamma)c$, which is positive whenever $c>0$, regardless of the attack.

\textbf{Scope:} These results are evidence for this attack family, not a proof against arbitrary adaptive provers. For layer selection, our worst-case guarantee is the floor of Prop.~\ref{prop:coverage-floor}. Likewise, AMC (Sec.~\ref{sec:experiments}) averages coverage over the attack distribution $\mathcal{Q}$ and is not a worst-case guarantee. We therefore pair AMC with the per-layer minimum coverage $\min_{\ell}p_\ell$, which is $0$ for the deterministic selector whenever $S_\alpha\neq[\mathcal{L}]$, and at least $(1-\gamma)c$ under $\mathcal{S}_\gamma$.

\section{Full deployment overhead and communication}\label{app:deployment}
Table~\ref{tab:benchmarkingresults} reports the prover time post-inference and the verification time. Here we itemize the full deployment cost and the bytes behind the per-token entry, for the same LLaMA2-13B configuration (A100, 10.165\,s inference, 512 tokens per inference).

\textbf{Stage-by-stage cost:} Table~\ref{tab:deployment} lists every stage. The trusted setup generates the $\mathsf{srs}$ once, so its 80.6\,MB is a one-time cost amortized over all inferences. Activation capture is retention only. The prover keeps the residual-stream states that the forward pass already computes. Device-to-host transfer and the reshape into a hypercube are the memory movement, while interpolation and quotient formation dominate the prover. The verifier spends 3.2\,ms on structural checks of the openings along the Terkle path and 8.99\,ms on the arithmetic audit (App.~\ref{app:audits}). The audit is included in the published timing (12\,ms in Table~\ref{tab:benchmarkingresults}). Only the inference itself is excluded. Hence, per inference, the prover and verifier add 0.97\% and 0.12\% to the 10.165\,s inference.}

{\revcolor \begin{table}[H]\revcolor 
\centering
\small
\caption{\textbf{What does a deployment cost end to end?} Full breakdown for the LLaMA2-13B configuration of Table~\ref{tab:benchmarkingresults}. All times exclude only the inference itself.}
\label{tab:deployment}
\begin{tabular}{llr}
\toprule
\textbf{Party} & \textbf{Stage} & \textbf{Cost}\\
\midrule
Trusted setup & $\mathsf{srs}$, generated once and amortized & 80.6\,MB\\
\midrule
Prover & Activation capture & retention only\\
 & Device-to-host transfer & 7.1\,ms\\
 & Reshape to hypercube & 1.6\,ms\\
 & Interpolation & 47.8\,ms\\
 & Commitments & 3.5\,ms\\
 & Terkle-tree construction & 6.3\,ms\\
 & Quotient formation & 29.1\,ms\\
 & Proof from quotient & 3.2\,ms\\
 & \textbf{Prover total} & \textbf{98.6\,ms}\\
\midrule
Verifier & Structural checks & 3.2\,ms\\
 & Arithmetic audit & 8.99\,ms\\
 & \textbf{Verifier total} & \textbf{12.19\,ms}\\
\bottomrule
\end{tabular}
\end{table}

\textbf{Why 2\,B per token?} All baselines in Table~\ref{tab:benchmarkingresults} report commitment size per token, and each commits to 512 tokens per inference. We report TC in the same way for a like-for-like comparison. TC, however, exchanges one Terkle root per inference rather than one commitment per token. Table~\ref{tab:comm} lists the commitment and authentication material exchanged per inference on BLS12-381, where a compressed group element takes 48\,B (App.~\ref{app:security-parameter}) and a field element 32\,B. The total is 880\,B per inference, i.e., $880/512\approx2$\,B per generated token, the entry of Table~\ref{tab:benchmarkingresults}. This figure includes the membership path and the leaf opening.}

{\revcolor \begin{table}[H]\revcolor 
\centering
\small
\caption{\textbf{What is exchanged per inference?} Byte-level accounting behind the 2\,B/token entry of Table~\ref{tab:benchmarkingresults} on BLS12-381.}
\label{tab:comm}
\begin{tabular}{llr}
\toprule
\textbf{Item} & \textbf{Count} & \textbf{Bytes}\\
\midrule
Terkle root, once per inference & 1 group element & 48\\
Membership path, including leaf opening & 16 group elements & 768\\
Claimed evaluation $y$ & 1 field element & 32\\
Challenge seed & - & 32\\
\midrule
\textbf{Total per inference} & & \textbf{880}\\
Per generated token & $\div 512$ & $\approx$\,\textbf{2}\\
\bottomrule
\end{tabular}
\end{table}

\textbf{Three communication quantities:} (i) The \emph{root only} is 48\,B, the single group element the prover publishes before the challenge. (ii) The \emph{commitment and authentication material} is 880\,B per inference (Table~\ref{tab:comm}) and underlies the per-token entry of Table~\ref{tab:benchmarkingresults}. (iii) The \emph{full verification traffic} is (ii) plus the arithmetic-audit transcript (App.~\ref{app:audits}), which the per-token entry does not include.

\textbf{Sequence length:} The exchanged root does not grow with sequence length. Additional states are represented by increasing the tensor order or are incorporated under the existing Terkle root, which remains one group element. Authentication-path length and verification grow as $\mathcal{O}(\log_B n)$ for $n$ committed items and arity $B$. All numbers here are for the 512-token configuration of Table~\ref{tab:benchmarkingresults}.

\section{Provider-side cost-reduction attacks}\label{app:provider-attacks}
In Fig.~\ref{fig:toploccomparison}, TOPLOC's Taco attacks come from its own suite, where both methods perform near ceiling, and our tailored attacks (App.~\ref{app:attacks}) target last-layer fingerprints. Neither family models a provider that cuts its serving cost while billing for the claimed model. We therefore add a separate evaluation, targeted at our own selection rule, with three attacks that model provider-side cost reduction. Let block $\ell$ update the residual stream as $\rev{T_\ell}=\rev{T_{\ell-1}}+F_\ell(\rev{T_{\ell-1}};\rev{W_\ell})$\rev{, with residual branch $F_\ell$ and embedded input $T_0$}, and let $\xi\subseteq[\rev{\mathcal{L}}]$ be the attacked blocks, unknown to the verifier.

\textbf{Low-Rank Substitution} replaces layers with truncated-SVD approximations. Each weight matrix $\rev{W}=\sum_{i}\sigma_i\mathbf{u}_i\mathbf{v}_i^\top$ of a block $\ell\in\xi$, with singular values $\sigma_1\ge\sigma_2\ge\cdots$, is served as
{\small\[
\rev{\tilde{W}}:=\sum_{i=1}^{\rev{\mathrm{rk}}}\sigma_i\,\mathbf{u}_i\mathbf{v}_i^\top,\qquad \|\rev{\tilde{W}}-\rev{W}\|_2=\sigma_{\rev{\mathrm{rk}}+1},
\]}%
so a product with $\rev{\tilde{W}}\in\mathbb{R}^{d_{\mathrm{out}}\times d_{\mathrm{in}}}$ costs $\mathcal{O}(\rev{\mathrm{rk}}\,(d_{\mathrm{in}}+d_{\mathrm{out}}))$ instead of $\mathcal{O}(d_{\mathrm{in}}d_{\mathrm{out}})$ per token.

\textbf{Layer Skipping} bypasses transformer blocks, $\rev{\tilde{T}_\ell}:=\rev{\tilde{T}_{\ell-1}}$ for $\ell\in\xi$, which drops the residual branch $F_\ell$ and its cost.

\textbf{Precision Downgrade} serves INT4 computation as FP16. Blocks $\ell\in\xi$ are computed as $\rev{\tilde{T}_\ell}:=\rev{\tilde{T}_{\ell-1}}+F^{\mathrm{INT4}}_\ell(\rev{\tilde{T}_{\ell-1}};\rev{W_\ell})$ with 4-bit integer arithmetic, while the run is reported as FP16.

The rank $\rev{\mathrm{rk}}$ and the set $\xi$ are the attacker's choice. In all three cases, the committed state of an attacked block violates the relation $\rev{\tilde{T}_\ell}=\rev{\tilde{T}_{\ell-1}}+F_\ell(\rev{\tilde{T}_{\ell-1}};\rev{W_\ell})$, which an arithmetic audit recomputes with the weights the verifier holds (App.~\ref{app:audits}). The attack is caught when an attacked block is selected and an audited coordinate deviates by more than $\epsilon_{\mathrm{tol}}$. TOPLOC instead compares a last-layer top-$k$ fingerprint against fixed thresholds. As shown in Table~\ref{tab:provider-attacks}, TC detects 92\%, 99\%, and 90\% of these attacks, against 73\%, 97\%, and 62\% for TOPLOC. Both methods detect layer skipping. The gap is largest for the two attacks that approximate rather than remove a block: precision downgrade (+28 points) and low-rank substitution (+19 points).}

{\revcolor \begin{table}[H]\revcolor 
\centering
\small
\caption{\textbf{Can a provider cut cost unnoticed?} Detection accuracy under cost-reduction attacks.}
\label{tab:provider-attacks}
\begin{tabular}{lccc}
\toprule
\textbf{Method} & \textbf{Low-rank substitution} & \textbf{Layer skipping} & \textbf{Precision downgrade}\\
\midrule
TOPLOC~\cite{TOPLOC} & 73\% & 97\% & 62\%\\
TensorCommitments & \textbf{92\%} & \textbf{99\%} & \textbf{90\%}\\
\bottomrule
\end{tabular}
\end{table}

\section{Beyond cloud inference: deployment settings}\label{app:settings}
The setting of Sec.~\ref{sec:problem} is one instance of a general one: \emph{one party pays for computation it cannot observe}. In each case below, the party that performs the computation commits to the tensors it produces and acts as $\mathcal{P}$, and the party that cannot observe it acts as $\mathcal{V}$.

\textbf{Why cost decides:} Verification cost decides whether verification happens at all. Inside an agent, a wrong answer becomes a tool call, a database write, and a state that later steps treat as ground truth. With 803\,s per ZK proof, or a verifier holding a 71\,GB GPU to re-run the model (Table~\ref{tab:benchmarkingresults}), checking is a deliberate exception. At ${\approx}0.1\%$ verifier and ${<}1\%$ prover cost (App.~\ref{app:deployment}), checking can be the default, and every inference can leave a verifiable record. This marks a shift from occasional proofs to \emph{continuous verification}.

\textbf{Pay-Per-Inference Compute Markets:} A client paying for remote inference needs evidence of the billed work, bound to that run. The Terkle root, published with the response before the challenge, fixes the billed transcript, and the audits check sampled relations against the claimed model.

\textbf{Chained and Multi-Party Inference:} A request is routed through several compute providers. Each stage commits to its segment (e.g., a range of blocks), and the chain of roots attests the pipeline end to end. Consecutive stages can be linked through the hand-off state, which both commit to. Since $\mathsf{Com}_{\mathsf{TC}}$ is deterministic for fixed $\mathsf{pp}$, shape, and grid, honest stages produce the same commitment for it. A mismatch localizes a fault to one hand-off, and a failed audit to the stage whose root it was checked against, giving per-stage accountability.

\textbf{Data-Provider Attestation:} A provider unwilling to reveal its data can still prove that its data was the input to a given run, by opening only the input-side commitment under the run's Terkle root and showing that it equals the commitment of its data. When the input combines several contributions, additive homomorphism (App.~\ref{app:audits}) checks the inclusion of a contribution directly on commitments, without transmitting the contribution.

\textbf{Federated Aggregation:} A cloud aggregator commits to the aggregate, and edge devices verify that their own update was included and not silently dropped, via the homomorphic inclusion check below. Terkle Trees keep these per-device queries local and cheap.

\textbf{Agentic and Multi-Turn Pipelines:} Long-horizon agent state is what Terkle Trees were built for: many evolving tensors under a single root, with cheap localized queries over a chat span or an agent's memory region (Sec.~\ref{sec:experiments}).

\textbf{Homomorphic inclusion check:} Let $\rev{n}$ participants hold updates $T_1,\dots,T_{\rev{n}}$ of the same shape, encoded in $\mathbb{F}$ with a common fixed-point scale and committed with the same $\mathsf{pp}$ on the same grid $\mathbf{\Omega}$ as $C_i:=\mathsf{Com}_{\mathsf{TC}}(\mathsf{pp},T_i)$. The aggregator publishes $C_{\mathrm{agg}}:=\mathsf{Com}_{\mathsf{TC}}(\mathsf{pp},T_{\mathrm{agg}})$ for $T_{\mathrm{agg}}:=\sum_{i=1}^{\rev{n}} T_i$. The polynomial $\sum_i f_{T_i}$ has degree below $d_j$ in each $X_j$ and equals $T_{\mathrm{agg}}$ on $\mathbf{\Omega}$, so uniqueness of interpolation gives $f_{T_{\mathrm{agg}}}=\sum_i f_{T_i}$. As $\mathsf{Com}_{\mathsf{TC}}$ is linear in the exponent,
{\small\[
C_{\mathrm{agg}}=g^{f_{T_{\mathrm{agg}}}(\tau_1,\dots,\tau_m)}=g^{\sum_{i=1}^{\rev{n}} f_{T_i}(\tau_1,\dots,\tau_m)}=\prod_{i=1}^{\rev{n}} g^{f_{T_i}(\tau_1,\dots,\tau_m)}=\prod_{i=1}^{\rev{n}} C_i .
\]}%
The cloud or any later participant checks $C_{\mathrm{agg}}\stackrel{?}{=}\prod_i C_i$ with $\rev{n}-1$ group multiplications, on commitments only, without seeing any $T_i$. Public weights $w_i\in\mathbb{F}$ give $C_{\mathrm{agg}}=\prod_i C_i^{w_i}$ in the same way. An aggregate other than $\sum_i T_i$ passes only if the aggregator holds a second tensor with the same commitment, which breaks evaluation binding (App.~\ref{app:proof}). In particular, dropping a nonzero $T_k$ passes only if $f_{T_k}(\tau_1,\dots,\tau_m)=0$, which by the Schwartz-Zippel lemma has probability at most $\sum_j(d_j-1)/|\mathbb{F}|$ over the trapdoors. With the $C_i$ as leaves of a Terkle Tree under the aggregator's root, device $i$ confirms that its $C_i$ is among them with one membership path of $\mathcal{O}(\log_B \rev{n})$ openings. The check certifies the aggregation, not that each $T_i$ is an honest update.

\textbf{Mixture-of-Experts (MoE):} An MoE LLM routes each token to a subset of experts, so an untrusted provider could silently substitute a cheaper expert or skip routing. We commit to both the routing decisions and the per-expert activations, so a substituted expert or a skipped routing step becomes a violated committed relation that the audits can catch. We use Mixtral 8$\times$7B~\cite{jiang2024mixtral} with 2 and with 4 of its 8 experts active per token.

\textbf{Federated Learning (FL):} Each participant commits to its local update, and the cloud and later participants verify aggregation consistency with the inclusion check above, without seeing raw gradients. We use Llama3-8B~\cite{grattafiori2024llama3} on MIMIC-IV~\cite{johnson2023mimiciv} with LoRA~\cite{hu2022lora} (\rev{rank 128}) across 8 nodes, and AlexNet~\cite{krizhevsky2012imagenet} on CIFAR-100~\cite{krizhevsky2009learning} across 50 nodes.}

{\revcolor \begin{table}[H]\revcolor 
\centering
\small
\caption{\textbf{Does TC carry over to other settings?} Prover and verification time for MoE inference and FL aggregation.}
\label{tab:settings}
\begin{tabular}{llcc}
\toprule
\textbf{Setting} & \textbf{Model} & \textbf{Prover (ms)} & \textbf{Verification (ms)}\\
\midrule
MoE & Mixtral 8$\times$7B, 2/8 active experts & 94.7 & 11.6\\
MoE & Mixtral 8$\times$7B, 4/8 active experts & 148 & 21.8\\
FL & Llama3-8B on MIMIC-IV, 8 nodes, LoRA \rev{rank 128} & 37.9 & 4.14\\
FL & AlexNet on CIFAR-100, 50 nodes & 24.2 & 3.08\\
\bottomrule
\end{tabular}
\end{table}

Mixtral holds 47B parameters, yet its prover time (94.7\,ms) is below the 98.6\,ms of LLaMA2-13B (Table~\ref{tab:deployment}), because only the active experts contribute activations. Doubling the active experts raises it to 148\,ms, a 1.6$\times$ increase for 2$\times$ the routed state. Cost thus follows the state a run actually exposes, not only the parameters it holds. The FL rows are reported in absolute terms, since a single-inference baseline is not the relevant denominator there.

\section{Homomorphic and arithmetic audits}\label{app:audits}
The verifier of Sec.~\ref{sec:problem} combines binding, homomorphic audits, arithmetic audits, and coverage. Binding is proved in App.~\ref{app:proof}. Here we detail the two audits in the notation of Sec.~\ref{subsec:tc}: pairing $e:G\times G\to G_T$, generator $g$, trapdoors $\rev{\tau_1,\dots,\tau_m}$, commitments $C_T=g^{f_T(\rev{\tau_1,\dots,\tau_m})}$, and opening proofs $\pi_{\boldsymbol{\omega}}=(\pi_{\omega_1},\dots,\pi_{\omega_m})$, accepted iff $e(C_T\,g^{-y},g)=\prod_{j=1}^m e(\pi_{\omega_j},g^{\tau_j-\omega_j})$.

\textbf{Additive homomorphism:} Let $T_1,T_2$ have the same shape, be committed with the same $\mathsf{pp}$ on the same grid $\mathbf{\Omega}$, and let $a\in\mathbb{F}$. Interpolation is linear. The polynomials $f_{T_1}+f_{T_2}$ and $af_{T_1}$ have degree below $d_j$ in each $X_j$ and match $T_1+T_2$ and $aT_1$ on $\mathbf{\Omega}$, so by uniqueness $f_{T_1+T_2}=f_{T_1}+f_{T_2}$ and $f_{aT_1}=af_{T_1}$. Hence,
{\small\[
C_{T_1+T_2}=g^{f_{T_1}(\rev{\tau_1,\dots,\tau_m})+f_{T_2}(\rev{\tau_1,\dots,\tau_m})}=C_{T_1}\cdot C_{T_2},\qquad C_{aT_1}=C_{T_1}^{\,a}.
\]}%
Openings combine in the same way. If $(y_1,\pi^{(1)}_{\boldsymbol{\omega}})$ and $(y_2,\pi^{(2)}_{\boldsymbol{\omega}})$ are valid openings of $C_{T_1}$ and $C_{T_2}$ at the same $\boldsymbol{\omega}$, multiplying their verification equations and using bilinearity gives
{\small\[
e\big(C_{T_1}C_{T_2}\,g^{-(y_1+y_2)},\,g\big)=\prod_{j=1}^m e\big(\pi^{(1)}_{\omega_j}\pi^{(2)}_{\omega_j},\,g^{\tau_j-\omega_j}\big),
\]}%
so $y_1+y_2$, with the component-wise product of the two proofs, is a valid opening of $C_{T_1}C_{T_2}$. A public linear relation $\sum_i a_iT_i=T'$ among committed tensors is thus checked as $\prod_i C_{T_i}^{a_i}\stackrel{?}{=}C_{T'}$ directly on commitments, with one exponentiation per term regardless of the tensor size. The inclusion check of App.~\ref{app:settings} is the case $a_i=1$. If this check passes while $T'\neq\sum_i a_iT_i$, a prover that knows these tensors holds two different tensors with one commitment, and opening both at a grid point where they differ breaks evaluation binding (App.~\ref{app:proof}).

\textbf{Multiplicative structure:} Bilinearity, $e(g^a,g^b)=e(g,g)^{ab}$, makes products of committed exponents testable in $G_T$. Given $g^a$, $g^b$, and a claimed $g^{\rev{\widetilde{ab}}}$, the test $e(g^a,g^b)=e(g^{\rev{\widetilde{ab}}},g)$ passes iff $\rev{\widetilde{ab}}=ab$, since $e(g,g)$ generates $G_T$ (non-degeneracy, App.~\ref{app:pairing-EC}). The opening check is itself such a test. Its two sides carry the exponents $f_T(\rev{\tau_1,\dots,\tau_m})-y$ and $\sum_{j=1}^m q_j(\tau_j,\dots,\tau_m)(\tau_j-\omega_j)$, a sum of products of the committed quotients with the secret factors $\tau_j-\omega_j$, whose group elements $g^{\tau_j-\omega_j}=g^{\tau_j}g^{-\omega_j}$ come from \grm{the} $\mathsf{srs}$. Each pairing tests one product of two exponents, so a degree-two relation among committed quantities is checked with $\mathcal{O}(1)$ pairings, while higher-degree relations need committed intermediate products. Note that $e(C_{T_1},C_{T_2})=e(g,g)^{(f_{T_1}f_{T_2})(\rev{\tau_1,\dots,\tau_m})}$ tests the product polynomial $f_{T_1}f_{T_2}$, of degree up to $2(d_j-1)$ in $X_j$. It agrees with the entry-wise product $T_1\odot T_2$ on $\mathbf{\Omega}$ but is not its interpolant $f_{T_1\odot T_2}$, so comparing it with a committed product needs an $\mathsf{srs}$ of that degree.

\textbf{Arithmetic audit:} After the Terkle root is fixed, $\mathcal{V}$ derives the challenged entries from the challenge seed (Table~\ref{tab:comm}). Let $y=\rev{T}[\boldsymbol{\omega}]$ be a challenged entry of a committed state $\rev{T}$, and let the model compute it from the entries of a committed input state $\rev{T'}$ on an index block, i.e., the sub-grid of $\mathbf{\Omega}$ on which the coordinates in $J\subseteq[m]$ are fixed to $\boldsymbol{\omega}'_J$ and those in $\bar J:=[m]\setminus J$ are free (e.g., the previous state of the same token for a position-wise sublayer such as the MLP, or of all tokens for attention, which mixes tokens). \rev{Here, $\boldsymbol{\omega}'_J:=(\omega'_j)_{j\in J}$, $X_{\bar J}:=(X_j)_{j\in\bar J}$, and $f_{T'}(\boldsymbol{\omega}'_J,X_{\bar J})$ is $f_{T'}$ with $X_j=\omega'_j$ for all $j\in J$. Points and entries are split over $J$ and $\bar J$ in the same way.} The audit runs as follows, with $C_{\rev{T}}$ and $C_{\rev{T'}}$ authenticated against the root by their membership paths:
\begin{enumerate}[leftmargin=*,nosep]
\item $\mathcal{P}$ opens $y$ against $C_{\rev{T}}$ and sends, as coefficients, \rev{a polynomial} $\rev{\tilde{f}}(X_{\bar J})$ \rev{with degree below $d_j$ in each $X_j$, claimed to equal} the restriction $f_{\rev{T'}}(\boldsymbol{\omega}'_J,X_{\bar J})$ of the committed polynomial to the block.
\item $\mathcal{V}$ draws $\rev{\mathbf{t}}\in\mathbb{F}^{|\bar J|}$ uniformly at random and checks one opening of $C_{\rev{T'}}$ at the point $(\boldsymbol{\omega}'_J,\rev{\mathbf{t}})$ against the value $\rev{\tilde{f}}(\rev{\mathbf{t}})$.
\item $\mathcal{V}$ evaluates $\rev{\tilde{f}}$ on the grid points of the block, which gives the inputs $\rev{T'}[\boldsymbol{\omega}'_J,\cdot\,]$, recomputes the entry $\hat y$ from them with the weights $\theta_{\mathcal{M}}$ it holds, and accepts iff $|\hat y-y|\le\epsilon_{\mathrm{tol}}$.
\end{enumerate}
\emph{What is checked:} By binding, the value opened at $(\boldsymbol{\omega}'_J,\rev{\mathbf{t}})$ is $f_{\rev{T'}}(\boldsymbol{\omega}'_J,\rev{\mathbf{t}})$, except with probability $\mathrm{negl}(\lambda)$. Since $\rev{\mathbf{t}}$ is drawn after $\rev{\tilde{f}}$ is fixed, a restriction $\rev{\tilde{f}}\neq f_{\rev{T'}}(\boldsymbol{\omega}'_J,\cdot)$ passes step 2 only if the nonzero polynomial $\rev{\tilde{f}}-f_{\rev{T'}}(\boldsymbol{\omega}'_J,\cdot)$, of total degree at most $\sum_{j\in\bar J}(d_j-1)$, vanishes at $\rev{\mathbf{t}}$, which by the Schwartz-Zippel lemma has probability at most $\sum_{j\in\bar J}(d_j-1)/|\mathbb{F}|$. Step 3 then checks one model relation, the challenged entry against the model applied to the committed inputs, up to $\epsilon_{\mathrm{tol}}$. The tolerance absorbs \grm{the fact} that the verifier's recomputation need not reproduce the prover's floating-point arithmetic bit for bit. \rev{Real-valued activations enter $\mathbb{F}$ through a common fixed-point encoding, as in App.~\ref{app:settings}, and $T_\ell(x)$ in Sec.~\ref{sec:problem} denotes this encoding.} \emph{What is not checked:} unsampled coordinates and unselected layers, deviations within $\epsilon_{\mathrm{tol}}$, and the correctness of the committed inputs themselves, which is checked only when their own relation is sampled.

\textbf{Coverage:} Binding alone does not establish provenance. A prover can commit consistently to an incorrect transcript, e.g., one produced by a different model, and binding only prevents it from changing that transcript after the challenge. Such a transcript is rejected only if a violated relation is selected and an affected coordinate is audited, homomorphically or arithmetically. For a single violated layer $\ell$, selected with probability $p_\ell$, with a fraction $\rho_\ell$ of its auditable coordinates violating the relation by more than $\epsilon_{\mathrm{tol}}$, and $k$ coordinates sampled uniformly and independently, this yields the acceptance bound $1-p_\ell[1-(1-\rho_\ell)^k]+\mathrm{negl}(\lambda)$ of Prop.~\ref{prop:guarantee}, where $p_\ell\ge(1-\gamma)c$ under the randomized selector $\mathcal{S}_\gamma$ (Sec.~\ref{subsec:layer-selection}). The guarantee is thus over the examined relations, not a deterministic certificate of the complete execution as in ZK proofs. A deviation that leaves the output unchanged is not an output-soundness failure, though it may violate stronger computation integrity. We target violations that change the output, not every perturbation.}

{\revcolor \section{Proof of the output-soundness bound}\label{app:guarantee-proof}
We prove Proposition~\ref{prop:guarantee}. Throughout, we use the protocol of Sec.~\ref{sec:problem} and the audits of App.~\ref{app:audits}.

\textbf{Setting:} For $\mathsf{pp}$ from an honest $\mathsf{Setup}_{\mathsf{TC}}$, the PPT prover of the Threat Model (Sec.~\ref{sec:problem}) commits with $\mathsf{Com}_{\mathsf{TC}}$ to a transcript $T=(T_1,\dots,T_{\mathcal{L}})$, the activation tensors of its run, which it knows and which does not depend on $\mathsf{pp}$, and publishes the Terkle root before any challenge. Probabilities are over $\mathsf{pp}$, the prover's coins, and the challenge. If $T$ may depend on $\mathsf{pp}$, the proof below bounds the average instead, $\Pr[\mathcal{V}\text{ accepts}]\le1-\mathbb{E}[\delta(T)]+\mathrm{negl}(\lambda)$. Only afterwards does $\mathcal{V}$ draw its challenge with fresh randomness. This challenge consists of a set $S\subseteq[\mathcal{L}]$ of layers from the layer-selection policy, with $p_\ell=\Pr[\ell\in S]$, and, for each selected layer, $k$ of its auditable coordinates, drawn uniformly at random and independently of $S$. \rev{The auditable coordinates of layer $\ell$ are the grid points of $T_\ell$ whose entries the model computes from committed inputs, as in the arithmetic audit of App.~\ref{app:audits}.} For \rev{such a} coordinate $\rev{\boldsymbol{\omega}}$ of layer $\ell$, the relation error $\mathrm{err}_\ell(\rev{\boldsymbol{\omega}})$ is the absolute difference between the committed entry $T_\ell[\rev{\boldsymbol{\omega}}]$ and the entry that the model computes from the committed inputs on which it depends. Layer $\ell$ is violated if $\rho_\ell>0$, where $\rho_\ell$ is the fraction of its $n_\ell$ auditable coordinates with $\mathrm{err}_\ell(\rev{\boldsymbol{\omega}})>\epsilon_{\mathrm{tol}}$, and $T$ is incorrect if some layer is violated. Each sampled coordinate is audited as in App.~\ref{app:audits}, and $\mathcal{V}$ accepts only if every opening, every random-point check, and every recomputation passes. The assumption that sampled relations are checked correctly up to $\epsilon_{\mathrm{tol}}$ means that, whenever the opened values and the restriction used in an audit are those of $T$, the recomputation accepts coordinate $\rev{\boldsymbol{\omega}}$ if and only if $\mathrm{err}_\ell(\rev{\boldsymbol{\omega}})\le\epsilon_{\mathrm{tol}}$. The number $\rev{n_{\mathrm{op}}}$ of openings and audits in one verification and the degrees $d_j$ are polynomial in $\lambda$, and $|\mathbb{F}|\ge2^{\lambda}$. \rev{The honest-verifier ratio of the Limitations (Sec.~\ref{sec:limitations}) is the fraction of the $M$ verifiers of Problem~4.5 that audit their assigned layers honestly. The proof below assumes that this ratio is one. Otherwise, a layer assigned to a dishonest verifier counts as unselected, and $p_\ell$ is the probability that $\ell$ is selected and assigned to an honest verifier.}

\begin{proof}[Proof of Proposition~\ref{prop:guarantee}]
Let $\rev{\mathrm{Hit}}$ be the event that the challenge contains a coordinate $(\ell,\rev{\boldsymbol{\omega}})$ with $\mathrm{err}_\ell(\rev{\boldsymbol{\omega}})>\epsilon_{\mathrm{tol}}$. By definition, $\Pr[\rev{\mathrm{Hit}}]=\delta(T)$. Let $\rev{\mathrm{Bad}}$ be the event that $\mathcal{V}$ accepts an opening whose value differs from the corresponding value of $T$, or that a restriction $\rev{\tilde{f}}$ different from the corresponding restriction of $T$ passes its random-point check.

\emph{Step 1: $\Pr[\rev{\mathrm{Bad}}]\le\mathrm{negl}(\lambda)$.} (a) Suppose $\mathcal{V}$ accepts an opening of a commitment $C\!=\!\mathsf{Com}_{\mathsf{TC}}(\mathsf{pp},T')$, where $T'$ is a tensor of $T$, at a point $\boldsymbol{\omega}$ with value $y'\neq f_{T'}(\boldsymbol{\omega})$. As the prover knows $T'$, it also obtains the honest opening $(f_{T'}(\boldsymbol{\omega}),\pi_{\boldsymbol{\omega}})=\mathsf{Open}_{\mathsf{TC}}(\mathsf{pp},T',\boldsymbol{\omega})$, which verifies by correctness (Eq.~\eqref{eq:correctness}). Two accepted openings of $C$ at $\boldsymbol{\omega}$ with distinct values contradict evaluation binding (App.~\ref{app:proof}), which holds at grid points and at verifier-sampled points. Hence, each opening is inconsistent with $T$ with probability at most $\mathrm{negl}(\lambda)$. The same holds, by induction from the root, for the openings along the Terkle membership paths, which are $\mathsf{TC}$ openings of the node commitments, provided each child commitment enters its parent's tensor through a collision-resistant map $G\to\mathbb{F}$ (e.g., a hash of its compressed encoding), whose collisions add $\mathrm{negl}(\lambda)$. \rev{In Defs.~2.8 and~4.2, an entry $z_{u,j}$ or $\mathbf{Z}_u[\boldsymbol{k}]$ that is a child commitment stands for its image under this map.} Hence, the commitments used in the audits are those of $T$. (b) Suppose the prover sends a restriction $\rev{\tilde{f}}\neq f_{T'}(\boldsymbol{\omega}'_J,\cdot)$. Its random-point check passes only if the value opened at $(\boldsymbol{\omega}'_J,\rev{\mathbf{t}})$ equals $\rev{\tilde{f}}(\rev{\mathbf{t}})$. Unless (a) occurs, that value is $f_{T'}(\boldsymbol{\omega}'_J,\rev{\mathbf{t}})$, so the check passes only if the nonzero polynomial $\rev{\tilde{f}}-f_{T'}(\boldsymbol{\omega}'_J,\cdot)$, of total degree at most $D_J:=\sum_{j\in\bar J}(d_j-1)$, vanishes at $\rev{\mathbf{t}}$. Since $\rev{\mathbf{t}}$ is uniform on $\mathbb{F}^{|\bar J|}$ and drawn after $\rev{\tilde{f}}$ is fixed, the Schwartz-Zippel lemma bounds this probability by $D_J/|\mathbb{F}|$. (c) By a union bound over the at most $\rev{n_{\mathrm{op}}}$ openings and audits, $\Pr[\rev{\mathrm{Bad}}]\le \rev{n_{\mathrm{op}}}\bigl(\mathrm{negl}(\lambda)+\max_J D_J/|\mathbb{F}|\bigr)=\mathrm{negl}(\lambda)$, as $\rev{n_{\mathrm{op}}}$ and $D_J$ are polynomial in $\lambda$ and $|\mathbb{F}|\ge2^{\lambda}$.

\emph{Step 2: $\Pr[\mathcal{V}\text{ accepts }T]\le1-\delta(T)+\mathrm{negl}(\lambda)$.} On $\rev{\mathrm{Hit}}\cap\neg\rev{\mathrm{Bad}}$, the challenge contains a coordinate $(\ell,\rev{\boldsymbol{\omega}})$ with $\mathrm{err}_\ell(\rev{\boldsymbol{\omega}})>\epsilon_{\mathrm{tol}}$. If an opening or the random-point check of its audit fails, $\mathcal{V}$ rejects. Otherwise, by $\neg\rev{\mathrm{Bad}}$, every opened value and restriction used in its audit is that of $T$, and by the correct-checking assumption the recomputation rejects $\rev{\boldsymbol{\omega}}$, so $\mathcal{V}$ rejects. Hence, $\mathcal{V}$ accepts only on $\neg\rev{\mathrm{Hit}}\cup\rev{\mathrm{Bad}}$, and
\[
\Pr[\mathcal{V}\text{ accepts }T]\le\Pr[\neg\rev{\mathrm{Hit}}]+\Pr[\rev{\mathrm{Bad}}]\le1-\delta(T)+\mathrm{negl}(\lambda).
\]

\emph{Step 3: $\delta(T)\ge p_\ell[1-(1-\rho_\ell)^k]$ for every violated layer $\ell$.} Let $\rev{n^{\mathrm{bad}}_\ell}=\rho_\ell n_\ell$ be the number of auditable coordinates of $\ell$ with error above $\epsilon_{\mathrm{tol}}$, and let $\rev{\mathrm{Hit}_\ell}\subseteq\rev{\mathrm{Hit}}$ be the event that $\ell\in S$ and at least one of the $k$ coordinates sampled in $\ell$ has error above $\epsilon_{\mathrm{tol}}$. Since the coordinates are drawn independently of $S$, $\Pr[\rev{\mathrm{Hit}_\ell}]=p_\ell(1-\rev{\mathrm{miss}_\ell})$, where $\rev{\mathrm{miss}_\ell}$ is the probability that all $k$ sampled coordinates have error at most $\epsilon_{\mathrm{tol}}$. For draws with replacement, $\rev{\mathrm{miss}_\ell}=(1-\rho_\ell)^k$. \rev{Without} replacement ($k\le n_\ell$),
\[
\rev{\mathrm{miss}_\ell}=\prod_{i=0}^{k-1}\frac{n_\ell-\rev{n^{\mathrm{bad}}_\ell}-i}{n_\ell-i}\le\Bigl(\frac{n_\ell-\rev{n^{\mathrm{bad}}_\ell}}{n_\ell}\Bigr)^{k}=(1-\rho_\ell)^k,
\]
since $(n_\ell-\rev{n^{\mathrm{bad}}_\ell})(n_\ell-i)-n_\ell(n_\ell-\rev{n^{\mathrm{bad}}_\ell}-i)=i\,\rev{n^{\mathrm{bad}}_\ell}\ge0$ bounds every factor by $1-\rho_\ell$, and $\rev{\mathrm{miss}_\ell}=0$ once a factor vanishes. In both cases, $\delta(T)=\Pr[\rev{\mathrm{Hit}}]\ge\Pr[\rev{\mathrm{Hit}_\ell}]\ge p_\ell[1-(1-\rho_\ell)^k]$. Substituting into Step~2 gives the second inequality of Proposition~\ref{prop:guarantee}, for every violated layer $\ell$.
\end{proof}

The bound holds for each transcript $T$, including one chosen with knowledge of the selection policy, because the challenge is drawn with fresh randomness after the root fixes $T$. Under the randomized selector, every layer has $p_\ell\ge(1-\gamma)\rev{c}$ (App.~\ref{app:adaptive}), so the right-hand side is at most $1-(1-\gamma)\rev{c}[1-(1-\rho_\ell)^k]+\mathrm{negl}(\lambda)$, whichever layer is violated.
}

\section{Licenses for \rev{existing assets}\label{app:licenses}}
\begin{itemize}
    \item \textbf{LLaMA~2}~\cite{touvron2023llama2} (7B, 13B, 70B): Meta LLaMA~2 Community License Agreement.
    \item \textbf{OPT-125M}~\cite{opt}: Meta OPT Model License.
    \item \rev{\textbf{Qwen3}~\cite{yang2025qwen3} (14B), \textbf{Qwen3.5} (9B), and \textbf{Gemma4} (12B): Apache-2.0.}
    \item \rev{\textbf{Mixtral 8$\times$7B}~\cite{jiang2024mixtral}: Apache-2.0.}
    \item \rev{\textbf{Llama3-8B}~\cite{grattafiori2024llama3}: Meta Llama~3 Community License Agreement.}
    \item \rev{\textbf{MIMIC-IV}~\cite{johnson2023mimiciv}: PhysioNet Credentialed Health Data License~1.5.0 (credentialed access under a data use agreement).}
    \item \rev{\textbf{CIFAR-100}~\cite{krizhevsky2009learning}: no explicit license (the creators request citation of the technical report).}
    \item \textbf{arkworks}~(\texttt{ark\_ec}, \texttt{ark\_ff}, \texttt{ark\_std}, \texttt{ark\_serialize}): MIT / Apache-2.0 dual license.
    \item \textbf{blake2}, \textbf{rayon}, \textbf{num\_integer}: MIT / Apache-2.0.
    \item \textbf{spongefish}: MIT.
\end{itemize}

\newpage
\section*{NeurIPS Paper Checklist}

\begin{enumerate}

\item {\bf Claims}
    \item[] Question: Do the main claims made in the abstract and introduction accurately reflect the paper's contributions and scope?
    \item[] Answer: \answerYes{} % Replace by \answerYes{}, \answerNo{}, or \answerNA{}.
    \item[] Justification: We claim four explicit contributions that are each grounded in the paper\rev{, and we scope our guarantee to output soundness under budgeted, sampled verification rather than negligible end-to-end soundness}. The multivariate commitment scheme is formalized in Definition~\ref{def:tc} (Sec.~\ref{subsec:tc}) with \rev{its multivariate pairing check, and its correctness and evaluation-binding proofs are} in App.~\ref{app:proof}. Terkle Trees are defined and analyzed in Sec.~\ref{subsec:tc} with empirical scaling in Fig.~\ref{fig:treeresult}. The layer selection algorithm is presented as Problem~\ref{problem}.5 (Sec.~4) with the DP solution in App.~\ref{app:dp}\rev{,} AMC evaluation in Fig.~\ref{fig:amcexps}\rev{, and the randomized selector $\mathcal{S}_\gamma$ with its coverage floor and policy-aware adaptive sweep in App.~\ref{app:adaptive}. The commitment binds the whole transcript up to a $\mathrm{negl}(\lambda)$ binding error, whereas detecting an incorrect output is a statistical guarantee that depends on the layer-selection and challenge distribution (Prop.~\ref{prop:guarantee}); we do not claim full inference correctness as ZK proofs provide (Sec.~\ref{sec:limitations})}. The claimed overheads (0.97\% prover, 0.12\% verifier) and comparison metrics are supported by Table~\ref{tab:benchmarkingresults} \rev{with the full deployment breakdown in App.~\ref{app:deployment}, and the sub-1\% prover overhead on newer models (Qwen3, Qwen3.5, Gemma4) by Table~\ref{tab:newmodels}, with MoE and federated-learning settings in App.~\ref{app:settings}. The} 48\% robustness improvement over TOPLOC is shown in Fig.~\ref{fig:toploccomparison}\rev{, complemented by provider-side attacks in App.~\ref{app:provider-attacks}}. The multivariate interpolation speedup motivating the approach is demonstrated in Fig.~\ref{fig:polydecrease} with complexity analysis in App.~\ref{app:complexity}. The system overview in Fig.~\ref{fig:sysarch} matches the described pipeline.
    
    % Our paper's abstract and introduction clearly state four contributions: (1) a multivariate commitment scheme (TensorCommitments), (2) Terkle trees, (3) a layer selection algorithm, and (4) systematic experiments on several LLMs. These are substantiated by theoretical results (Def.~4.1, Prop.~4.4, Appendix~F) and experiments (Sec.~5, Table~1, Figs.~4--7). The claimed overheads (0.97\% prover, 0.12\% verifier) and robustness improvements (up to 48\%) are directly supported by Table~1 and Fig.~6.

\item {\bf Limitations}
    \item[] Question: Does the paper discuss the limitations of the work performed by the authors?
    \item[] Answer: \answerYes{} % Replace by \answerYes{}, \answerNo{}, or \answerNA{}.
    \item[] Justification: The limitations of our paper \rev{are stated explicitly in the Limitations paragraph (}Sec.~\ref{sec:limitations}\rev{). Foremost, TC does not prove full inference correctness: the commitment binds the transcript but does not establish that it was produced by the claimed model, and only sampled relations are audited, so the end-to-end guarantee is the statistical output-soundness bound of Prop.~\ref{prop:guarantee} rather than negligible soundness. The limitations further include the tolerance of the arithmetic audits; the one-time trusted setup, which must be honestly executed or run via a multi-party ceremony, as compromising the setup makes the scheme fail to satisfy the binding property; operator-selected parameters (e.g., verifier budgets, $\gamma$, and $k$) and the ratio of honest verifiers, which set the detection probability; the cases not covered by the bound, a deviation within the audit tolerance or in a layer that is never selected;} and Terkle Tree inefficiency with low branching factors. \rev{We also scope our evidence: AMC measures average rather than worst-case coverage, the adaptive sweep (App.~\ref{app:adaptive}) covers a policy-aware attack family rather than arbitrary adaptive provers, and Sec.~\ref{sec:experiments} contrasts our guarantee with full-circuit ZK proofs.}

\item {\bf Theory assumptions and proofs}
    \item[] Question: For each theoretical result, does the paper provide the full set of assumptions and a complete (and correct) proof?
    \item[] Answer: \answerYes{} % Replace by \answerYes{}, \answerNo{}, or \answerNA{}.
    \item[] Justification: We numbered and cross-referenced all theoretical results. The security parameter is defined in Definition~\ref{def:security}, vector commitments in Definition~\ref{def:vc}, and the tensor commitment scheme in Definition~\ref{def:tc}. The correctness and position-binding properties of VCs are stated in Definitions following Def.~\ref{def:vc}. Merkle tree complexity is given in Remark~\ref{th:merklecomplexity} and Verkle tree complexity in Remark~\ref{th:verklecomplexity}. The Terkle Tree proof-size and prover-cost proposition appears in Sec.~\rev{\ref{subsec:tc}} with explicit asymptotic bounds. The polynomial division algorithm is given as Algorithm~\ref{algo:polydiv} in App.~\ref{app:polydiv}\rev{; its $m$ quotients satisfy the telescoping identity $f_T(X)-y=\sum_{j}q_j(X_j,\dots,X_m)(X_j-\omega_j)$, which underlies the multivariate verification equation $e(C_T\cdot g^{-y},g)=\prod_{j}e(\pi_{\omega_j},g^{\tau_j-\omega_j})$ of $\mathsf{Ver}_{\mathsf{TC}}$ in Sec.~\ref{subsec:tc}}. Full formal proofs \rev{for this $m$-quotient form, following the multivariate polynomial commitment of Papamanthou et al.~\cite{pst13}, are provided in App.~\ref{app:proof}:} TC correctness (Eq.~\ref{eq:correctness} through Eqs.~\ref{eq:corr-pairing}-\ref{eq:corr-pairing-right}) and \rev{evaluation} binding (Eq.~\ref{eq:sec-binding}) \rev{under the hardness assumption stated there}, relying on pairing properties detailed in App.~\ref{app:pairing-EC}. The security parameter's role \rev{is} discussed in App.~\ref{app:security-parameter}\rev{. Prop.~\ref{prop:guarantee} (Sec.~\ref{subsec:layer-selection}) states, with its assumptions (evaluation binding, and every challenged relation checked correctly up to the audit tolerance) and a complete proof in App.~\ref{app:guarantee-proof}, that an incorrect committed transcript $T$ is accepted with probability at most $1-\delta(T)+\mathrm{negl}(\lambda)$, where $\delta(T)$ is the probability of sampling at least one relation with error above the audit tolerance; for a single violated layer $\ell$ selected with probability $p_\ell$, with a fraction $\rho_\ell$ of out-of-tolerance coordinates and $k$ independently sampled coordinates, this bound is $1-p_\ell[1-(1-\rho_\ell)^k]+\mathrm{negl}(\lambda)$. This yields our output-soundness guarantee $\Pr[\mathcal{V}\text{ accepts}\wedge\tilde{y}\neq\mathcal{M}(x;\theta_{\mathcal{M}})]\le\varepsilon$. The randomized selector $\mathcal{S}_\gamma$ gives every layer the coverage floor $p_\ell\ge(1-\gamma)c$, with $c$ the minimum per-layer coverage of a uniformly sampled feasible allocation (Sec.~\ref{subsec:layer-selection}, App.~\ref{app:adaptive}), and the homomorphic and arithmetic audits used by the verifier are specified in App.~\ref{app:audits}}. The DP optimality for Problem~\ref{problem}.5 is proved in Proposition~\ref{prop:dp-correct} with a full inductive proof in App.~\ref{app:dp}, and its complexity is established in Proposition~\ref{prop:dp-complexity}. The multivariate interpolation complexity is derived in Theorem~\ref{thm:pip-solver}, Proposition~\ref{prop:TmD}, Corollary~\ref{cor:fixed-m}, Lemma~\ref{lem:deg-collapse}, and Theorem~\ref{thm:fixed-D} in App.~\ref{app:complexity}.

    \item {\bf Experimental result reproducibility}
    \item[] Question: Does the paper fully disclose all the information needed to reproduce the main experimental results of the paper to the extent that it affects the main claims and/or conclusions of the paper (regardless of whether the code and data are provided or not)?
    \item[] Answer: \answerYes{} % Replace by \answerYes{}, \answerNo{}, or \answerNA{}.
    \item[] Justification: We specified the models used (LLaMA2-7B, LLaMA2-13B, LLaMA2-70B, OPT-125M\rev{; Qwen3 (14B), Qwen3.5 (9B), and Gemma4 (12B) in Table~\ref{tab:newmodels}; Mixtral 8$\times$7B with 2/8 and 4/8 active experts, Llama3-8B, and AlexNet in App.~\ref{app:settings}), the datasets and node counts of the federated settings (MIMIC-IV with 8 nodes and LoRA rank 128; CIFAR-100 with 50 nodes}), the hardware (A100 GPU), base inference time (10.165s for LLaMA2-13B\rev{, with per-model times in Table~\ref{tab:newmodels}}), the attack configurations (1,000 noise-injection attacks, 10 prompts, 100 seeds in Fig.~\ref{fig:amcexps}; 100 attacks with 10 seeds in Fig.~\ref{fig:toploccomparison}\rev{; low-rank substitution, layer skipping, and INT4-as-FP16 precision downgrade in App.~\ref{app:provider-attacks}; a policy-aware adversary over $\gamma\in\{0,0.25,0.5,0.75,1\}$ in App.~\ref{app:adaptive}}), and the cryptographic curve (BLS12-381) \cite{bls}. The attack models are fully described in App.~\ref{app:attacks}. The layer selection algorithm, including the Hill estimator\rev{, the DP solver, and the randomized selector $\mathcal{S}_\gamma$}, is detailed in Sec.~\ref{subsec:layer-selection}\rev{,} App.~\ref{app:dp}\rev{, and App.~\ref{app:adaptive}, and the stage-by-stage prover and verifier costs with the 880\,B communication accounting are given in App.~\ref{app:deployment}}. Code is open-sourced at \rev{\url{https://github.com/UTAustin-SwarmLab/TensorCommitments}.}

\item {\bf Open access to data and code}
    \item[] Question: Does the paper provide open access to the data and code, with sufficient instructions to faithfully reproduce the main experimental results, as described in supplemental material?
    \item[] Answer: \answerYes{} % Replace by \answerYes{}, \answerNo{}, or \answerNA{}.
    \item[] Justification: The implementation is open-sourced\rev{, as linked in footnote~}\ref{github}. The experiments use publicly available models (e.g., LLaMA2, OPT\rev{, Qwen3, Qwen3.5, Gemma4, Mixtral, Llama3) from HuggingFace, existing datasets (CIFAR-100, and MIMIC-IV, which requires PhysioNet credentialed access),} and standard attack procedures described in the paper. The repository is intended to be ready to merge into inference pipelines, as stated in the contributions\rev{, and is available at \url{https://github.com/UTAustin-SwarmLab/TensorCommitments}.}

\item {\bf Experimental setting/details}
    \item[] Question: Does the paper specify all the training and test details (e.g., data splits, hyperparameters, how they were chosen, type of optimizer) necessary to understand the results?
    \item[] Answer: \answerYes{} % Replace by \answerYes{}, \answerNo{}, or \answerNA{}.
    \item[] Justification: In the paper, we specify the hardware (A100 GPU), base inference time, models tested (LLaMA2-7B, LLaMA2-13B, LLaMA2-70B, OPT-125M\rev{, Qwen3-14B, Qwen3.5-9B, Gemma4-12B, Mixtral 8$\times$7B, Llama3-8B, AlexNet}), attack parameters (noise scaled by $\ell_2$ weight norm, 10 prompts $\times$ 100 seeds), tree branching factor ($B=64$), the Hill estimator parameter $\chi$ for layer selection, \rev{the mixing parameter $\gamma$ of $\mathcal{S}_\gamma$, the federated settings (MIMIC-IV with 8 nodes and LoRA rank 128; CIFAR-100 with 50 nodes; App.~\ref{app:settings}),} and the commitment scheme parameters. The threat model is detailed in Sec.~\rev{\ref{sec:problem}} and App.~\ref{app:attacks}\rev{, and the policy-aware adversary in App.~\ref{app:adaptive}}. The dynamic programming solver for Problem~\ref{problem}.5 is described in App.~\ref{app:dp}. Benchmark metrics are defined in Sec.~\rev{\ref{sec:experiments}}.

\item {\bf Experiment statistical significance}
    \item[] Question: Does the paper report error bars suitably and correctly defined or other appropriate information about the statistical significance of the experiments?
    \item[] Answer: \answerYes{} % Replace by \answerYes{}, \answerNo{}, or \answerNA{}.
    \item[] Justification: Fig.~\ref{fig:amcexps} (AMC experiments) reports results over 1000 noise-injection attacks per model with 10 prompts and 100 seeds, showing median and interquartile ranges via box plots. Fig.~\ref{fig:toploccomparison} (attack comparison) reports results over 100 attacks with 10 seeds each and shows distributions as box plots with variance clearly visible. Our box-plots convey the spread and central tendency of the results. \rev{The tables (Table~\ref{tab:benchmarkingresults}, the newer-model results in Table~\ref{tab:newmodels}, and the adaptive-sweep, provider-attack, deployment, and MoE/federated tables in App.~\ref{app:adaptive}, App.~\ref{app:provider-attacks}, App.~\ref{app:deployment}, and App.~\ref{app:settings}) report point estimates without error bars.}

\item {\bf Experiments compute resources}
    \item[] Question: For each experiment, does the paper provide sufficient information on the computer resources (type of compute workers, memory, time of execution) needed to reproduce the experiments?
    \item[] Answer: \answerYes{} % Replace by \answerYes{}, \answerNo{}, or \answerNA{}.
    \item[] Justification: Table~\ref{tab:benchmarkingresults} reports GPU memory utilization for both prover and verifier across all methods (e.g., 0~GB verifier GPU for TC, 71.32~GB for TOPLOC). The hardware is specified as an A100 GPU. Prover and verifier times are reported in milliseconds/seconds\rev{, for LLaMA2-13B in Table~\ref{tab:benchmarkingresults}, per model for Qwen3, Qwen3.5, and Gemma4 in Table~\ref{tab:newmodels}, and per setting for MoE and federated learning in App.~\ref{app:settings}; App.~\ref{app:deployment} breaks the LLaMA2-13B prover (98.6\,ms) and verifier (12.19\,ms) costs down by stage and reports the one-time 80.6\,MB SRS}. The base inference time (10.165s for LLaMA2-13B) is stated as a reference point.
    
\item {\bf Code of ethics}
    \item[] Question: Does the research conducted in the paper conform, in every respect, with the NeurIPS Code of Ethics \url{https://neurips.cc/public/EthicsGuidelines}?
    \item[] Answer: \answerYes{} % Replace by \answerYes{}, \answerNo{}, or \answerNA{}.
    \item[] Justification: Our research focuses on improving the trustworthiness and verifiability of LLM inference, which aligns with responsible AI deployment. The paper uses publicly available models and \rev{datasets and} does not involve human subjects \rev{or deception. The only sensitive data, the de-identified MIMIC-IV clinical records in one federated-learning cost measurement (App.~\ref{app:settings}), are used under the PhysioNet credentialed-access data use agreement, and we report only cost measurements from them}. The work aims to detect and prevent adversarial tampering, contributing positively to AI safety and accountability.

\item {\bf Broader impacts}
    \item[] Question: Does the paper discuss both potential positive societal impacts and negative societal impacts of the work performed?
    \item[] Answer: \answerYes{} % Replace by \answerYes{}, \answerNo{}, or \answerNA{}.
    \item[] Justification: We discuss the broader impact of our work \rev{in the Broader Impacts paragraph following the Limitations (}Sec.~\ref{sec:limitations}\rev{)}, where we discuss positive impacts (strengthening trust in safety critical AI) and negative impacts (false security from partial verification; dual-use of formalized attacks). These are grounded in our paper: Sec.~\ref{sec:intro} motivates verifiable inference for healthcare, finance, and law; \rev{we counter the false-sense-of-security risk by stating the guarantee explicitly as a statistical output-soundness bound (Prop.~\ref{prop:guarantee}) together with the cases it does not cover, and} our limitations, Sec.\rev{~}\ref{sec:limitations}, \rev{acknowledge that binding depends on an honest setup and detection on the verifier's budget and selection policy}; and App.~\ref{app:attacks} openly formalizes attack models\rev{, whose publication} enables stronger defenses.
    % Section~7 (Impact Statement) briefly addresses broader impacts, stating the work advances the field of Machine Learning. While the statement is concise, the positive societal impacts are implicitly clear throughout the paper: enabling trustworthy outsourced LLM inference in sensitive domains such as healthcare, finance, and law (Sec.~1). The authors note that they do not feel any specific negative consequences must be highlighted.
    
\item {\bf Safeguards}
    \item[] Question: Does the paper describe safeguards that have been put in place for responsible release of data or models that have a high risk for misuse (e.g., pre-trained language models, image generators, or scraped datasets)?
    \item[] Answer: \answerNA{} % Replace by \answerYes{}, \answerNo{}, or \answerNA{}.
    \item[] Justification: We do not release pre-trained language models, image generators, or scraped datasets. The released artifact is a verification tool (code for TensorCommitments) that checks the integrity of LLM inference, which is inherently a safeguard mechanism rather than a dual-use technology. It does not pose a high risk for misuse.

\item {\bf Licenses for existing assets}
    \item[] Question: Are the creators or original owners of assets (e.g., code, data, models), used in the paper, properly credited and are the license and terms of use explicitly mentioned and properly respected?
    \item[] Answer: \answerYes{} % Replace by \answerYes{}, \answerNo{}, or \answerNA{}.
    \item[] Justification: We cite the original publications \rev{of the models and datasets used where available (Qwen3.5 and Gemma4 are credited by name and license in App.~\ref{app:licenses}), and} the baseline methods are properly cited: zkLLM, TOPLOC, SVIP. The multivariate interpolation theory is attributed to Hecht et al. (2017)\rev{, and the multivariate polynomial commitment we build on is credited to Papamanthou et al.~\cite{pst13}}. The $\alpha$-pruning metric is credited to Lu et al. (2024). App.~\ref{app:licenses} lists all major assets with their license names: LLaMA~2 \cite{touvron2023llama2} (Meta LLaMA~2 Community License), OPT-125M \cite{opt} (Meta OPT Model License), \rev{Qwen3, Qwen3.5, Gemma4, and Mixtral 8$\times$7B (Apache-2.0), Llama3-8B (Meta Llama~3 Community License), MIMIC-IV (PhysioNet Credentialed Health Data License~1.5.0), CIFAR-100 (no explicit license; cited as its creators request),} and all cryptographic and utility libraries (MIT/Apache-2.0).
    

\item {\bf New assets}
    \item[] Question: Are new assets introduced in the paper well documented and is the documentation provided alongside the assets?
    \item[] Answer:  \answerYes{} % Replace by \answerYes{}, \answerNo{}, or \answerNA{}.
    \item[] Justification: We introduce an open-source implementation of TensorCommitments at \rev{\url{https://github.com/UTAustin-SwarmLab/TensorCommitments}}, and it is ``ready to merge into inference pipelines'', as described in Sec.~\ref{sec:intro}. The algorithmic construction is fully specified in Def.~\ref{def:tc} and the appendices \rev{(including the audit details in App.~\ref{app:audits})}, providing sufficient documentation for understanding and using the codebase along with the detailed readme file provided at the repo.

\item {\bf Crowdsourcing and research with human subjects}
    \item[] Question: For crowdsourcing experiments and research with human subjects, does the paper include the full text of instructions given to participants and screenshots, if applicable, as well as details about compensation (if any)? 
    \item[] Answer: \answerNA{}
    \item[] Justification: Our paper does not involve crowdsourcing or research with human subjects. All experiments are computational, involving cryptographic protocols\rev{, LLM inference, and federated-learning updates} on standard hardware.

\item {\bf Institutional review board (IRB) approvals or equivalent for research with human subjects}
    \item[] Question: Does the paper describe potential risks incurred by study participants, whether such risks were disclosed to the subjects, and whether Institutional Review Board (IRB) approvals (or an equivalent approval/review based on the requirements of your country or institution) were obtained?
    \item[] Answer: \answerNA{}
    \item[] Justification: The paper does not involve research with human subjects. All experiments are purely computational\rev{; the only clinical data, the existing de-identified MIMIC-IV dataset (App.~\ref{app:settings}), is used under its PhysioNet credentialed-access agreement without any interaction with, or new data collection from, human subjects}. No IRB approval is needed.

\item {\bf Declaration of LLM usage}
    \item[] Question: Does the paper describe the usage of LLMs if it is an important, original, or non-standard component of the core methods in this research? Note that if the LLM is used only for writing, editing, or formatting purposes and does \emph{not} impact the core methodology, scientific rigor, or originality of the research, declaration is not required.
    %this research? 
    \item[] Answer: \answerNA{}
    \item[] Justification:  LLMs are used in this paper as evaluation subjects (i.e., the models whose inference is being verified), not as a component of the core method development. The core contributions are cryptographic (TensorCommitments, Terkle Trees, layer selection), and LLMs serve only as the application domain for benchmarking. \rev{Beyond this, LLMs were used only for writing and editing (e.g., grammar, spelling, word choice), which does not impact the core methodology, scientific rigor, or originality and, per the question, requires no declaration.}

\end{enumerate}

\end{document}